\documentclass[11pt]{article}
\usepackage[margin=1in]{geometry}
\usepackage{amsmath}
\usepackage{amssymb}
\usepackage{amsthm}
\usepackage{forloop}
\usepackage{nicefrac}
\usepackage{mathtools}
\usepackage[table]{xcolor}
\usepackage{url}
\usepackage{bm}
\usepackage{float}
\usepackage[linesnumbered, ruled]{algorithm2e}
\usepackage{hyperref}
\hypersetup{
    colorlinks,
    linkcolor={blue!100!black!70},
    citecolor={blue!75!black},
    urlcolor={blue!75!black}
}
\usepackage[shortlabels]{enumitem}
\usepackage[nameinlink,capitalise]{cleveref}
\usepackage{cancel}
\usepackage{multirow}
\usepackage{makecell}
\usepackage{caption}
\usepackage{subcaption}
\usepackage{framed}
\usepackage{wrapfig}
\usepackage{booktabs}

\newtheorem{theorem}{Theorem}[section]
\newtheorem{lemma}{Lemma}[section]
\newtheorem{corollary}{Corollary}[section]
\newtheorem{definition}{Definition}[section]
\newtheorem{proposition}{Proposition}[section]

\crefname{definition}{Def.}{Defs.}

\newcommand{\defcal} [1]{\expandafter\newcommand\csname cal#1\endcsname{{\cal #1}}}
\newcommand{\defsf} [1]{\expandafter\newcommand\csname sf#1\endcsname{{\sf #1}}}
\newcommand{\defbf} [1]{\expandafter\newcommand\csname bf#1\endcsname{{\bf #1}}}
\newcommand{\defbb} [1]{\expandafter\newcommand\csname bb#1\endcsname{{\mathbb{#1}}}}
\newcommand{\deffrak} [1]{\expandafter\newcommand\csname frak#1\endcsname{{\mathfrak{#1}}}}

\newcounter{ct}
\forLoop{1}{26}{ct}{
    \edef\letter{\Alph{ct}}
    \expandafter\defcal\letter
    \expandafter\defsf\letter
    \expandafter\defbf\letter
    \expandafter\defbb\letter
    \expandafter\deffrak\letter
}
\forLoop{1}{26}{ct}{
    \edef\letter{\alph{ct}}
    \expandafter\defcal\letter
    \expandafter\defsf\letter
    \expandafter\defbf\letter
    \expandafter\defbb\letter
    \expandafter\deffrak\letter
}
\newcommand{\argmin}{\mathop{\mathrm{argmin}}}

\newcommand{\numberthis}{\addtocounter{equation}{1}\tag{\theequation}}

\newcommand{\TVdist}{\mathrm{d}_{\mathrm{TV}}}
\newcommand{\KLdist}{\mathrm{d}_{\mathrm{KL}}}
\newcommand{\mirrorfunc}[1]{
    \ensuremath{\if$#1$\phi \else \phi(#1)\fi}
}
\newcommand{\mirrorfuncdual}[1]{
    \ensuremath{\if$#1$\phi^{*}\else \phi^{*}(#1)\fi}
}
\newcommand{\primtodual}[1]{\nabla\mirrorfunc{#1}}
\newcommand{\dualtoprim}[1]{\nabla\mirrorfuncdual{#1}}
\newcommand{\hessmirror}[1]{\nabla^{2}\mirrorfunc{#1}}
\newcommand{\hessmirrorinv}[1]{\nabla^{2}\mirrorfuncdual{#1}}
\newcommand{\trace}{\mathrm{trace}}
\newcommand{\defeq}{\overset{\mathrm{def}}{=}}

\newcommand{\subscript}[2]{$#1 _ #2$}
\newlist{assumplist}{enumerate}{1}
\setlist[assumplist]{label=(\subscript{\textbf{A}}{{\arabic*}})}

\newcommand{\potential}[1]{
    \ensuremath{\if$#1$f\else f(#1)\fi}
}
\newcommand{\gradpotential}[1]{\nabla\potential{#1}}
\newcommand{\mixingtime}[1]{\tau_{\mathrm{mix}}(#1)}
\newcommand{\interior}[1]{\mathrm{int}(#1)}
\newcommand{\crossratio}[3]{\mathrm{CR}(#1, #2; #3)}

\newcommand{\costgradpotential}{\mathrm{Cost}[\gradpotential{}]}
\newcommand{\costhessmirror}{\mathrm{Cost}[\nabla^{2}\phi]}
\newcommand{\costdualtoprim}{\mathrm{Cost}[\dualtoprim{}]}
\newcommand{\costcholesky}[1]{\mathrm{Cost}[\text{Cholesky}(#1)]}
\newcommand{\costtrisolve}[1]{\mathrm{Cost}[\text{TriSolve(#1)}]}
\newcommand{\costpotential}{\mathrm{Cost}[\potential{}]}

\newcommand{\MAMLA}{\textsf{MAMLA}}
\newcommand{\MALA}{\textsf{MALA}}
\newcommand{\MLA}{\textsf{MLA}}

\newcommand{\MLAFD}{\textsf{MLA\textsubscript{FD}}}
\newcommand{\MLABD}{\textsf{MLA\textsubscript{BD}}}
\newcommand{\primalspace}{\mathcal{K}}
\newcommand{\targetdens}{\pi}
\newcommand{\targetdist}{\Pi}

\makeatletter
\newcommand*{\algotitle}[2]{%
  \stepcounter{algocf}%
  \hypertarget{algocf.title.\theHalgocf}{}%
  \NR@gettitle{#1}%
  \label{#2}%
  \addtocounter{algocf}{-1}%
}
\makeatother

\usepackage[round]{natbib}
\usepackage[splitrule]{footmisc}

\crefname{paragraph}{part}{parts}
\makeatletter
\def\@maketitle{%
  \newpage
  \begin{center}%
  \let \footnote \thanks
    {\Large \bf \@title \par}%
  \end{center}%
  \par
  \vskip 0.5em}
\makeatother

\title{Fast sampling from constrained spaces using the \\
Metropolis-adjusted Mirror Langevin algorithm}

\begin{document}
\maketitle

\begin{center}
{\large
\begin{tabular}{ccc}
    \makecell{Vishwak Srinivasan\(^\star\)\\{\normalsize\texttt{vishwaks@mit.edu}}} & \makecell{Andre Wibisono\(^\dagger\)\\{\normalsize\texttt{andre.wibisono@yale.edu}}} & \makecell{Ashia Wilson\(^\star\)\\{\normalsize\texttt{ashia07@mit.edu}}}
\end{tabular}
\vskip 0.5em

\normalsize
\begin{tabular}{c}
\({}^{\star}\)Department of Electrical Engineering and Computer Science, MIT
\\
[0.25em]
\({}^{\dagger}\)Department of Computer Science, Yale University \\
\end{tabular}
}
\end{center}

\begin{abstract}
  \noindent
  We propose a new method called the Metropolis-adjusted Mirror Langevin algorithm for approximate sampling from distributions whose support is a compact and convex set.
  This algorithm adds an accept-reject filter to the Markov chain induced by a single step of the Mirror Langevin algorithm \citep{zhang2020wasserstein}, which is a basic discretisation of the Mirror Langevin dynamics.
  Due to the inclusion of this filter, our method is unbiased relative to the target, while known discretisations of the Mirror Langevin dynamics including the Mirror Langevin algorithm have an asymptotic bias.
  For this algorithm, we also give upper bounds for the number of iterations taken to mix to a constrained distribution whose potential is relatively smooth, convex, and Lipschitz continuous with respect to a self-concordant mirror function.
  As a consequence of the reversibility of the Markov chain induced by the inclusion of the Metropolis-Hastings filter, we obtain an exponentially better dependence on the error tolerance for approximate constrained sampling.
  We also present numerical experiments that corroborate our theoretical findings.
\end{abstract}

\setcounter{page}{1}

\vspace*{-6mm}
\section{Introduction}
\label{sec:intro}
\vspace*{-2mm}
Continuous distributions supported on high-dimensional spaces are prevalent in various areas of science, more commonly so in machine learning and statistics.
Samples drawn from such distributions can be used to generate confidence intervals for a point estimate, or provide Monte Carlo estimates for functionals of distributions.
This motivates development of efficient algorithms to sample from such distributions.
However, not all distributions are supported on the entire space (for e.g., \(\bbR^{d}\)).
In one dimension, some common examples are the Gamma and Beta distributions (supported on \((0, \infty)\) and \((0, 1)\), respectively), and the latter's generalisation to multiple dimensions is the Dirichlet distribution (supported on a simplex \(\Delta_{d + 1}\), a compact and convex subset of \(\bbR^{d}\)).
Such constrained distributions not only occur in theory, but also in practice: for instance, in several Bayesian models \citep{pakman2014exact}, latent Dirichlet allocation \citep{blei2003latent} for topic modelling, regularised regression \citep{celeux2012regularization}, more recently in language models \citep{kumar2022gradient}, and for modelling metabolic networks \citep{heirendt2019creation,kook2022sampling}.
Analogous to sampling, optimisation over constrained domains are also of interest, and constrained optimisation is generally harder than unconstrained optimisation.
Two popular families of approaches for solving constrained optimisation problems are interior-point methods \citep{nesterov1994interior}, and mirror descent \citep{nemirovskii1983problem}.

\begin{framed}
\textbf{The constrained sampling problem}~~
More formally, in this paper, we are interested in generating (approximate) samples from a target distribution \(\Pi\) with support \(\primalspace\) that is a compact, convex subset of \(\bbR^{d}\), and whose density \(\pi\) is of the form
\begin{equation*}
    \pi(x) \propto e^{-\potential{x}}~.
\end{equation*}
Here, \(f : \primalspace \to \bbR \cup \{\infty\}\) is termed the \emph{potential} of \(\Pi\).
\end{framed}

The constrained sampling problem has been long-studied and a variety of prior approaches exist.
The earliest inventions initially focused on obtaining uniform samples over \(\primalspace\) (i.e., when \(\potential{} = 0\)).
This task had applications in estimating the volume of \(\primalspace\), and these methods were subsequently generalised to sampling from log-concave distributions supported on \(\primalspace\) (equivalently when \(\potential{}\) is convex in its domain).
These inventions include the \textsf{Hit-And-Run} \citep{smith1984efficient,lovasz1999hit,lovasz2003hit}, and \textsf{BallWalk} algorithms \citep{lovasz1993random,kannan1997random}.
Amongst other later modifications and analyses of these algorithms, the \textsf{DikinWalk} \citep{kannan2009random,narayanan2017efficient} is particularly notable due to the use of Dikin ellipsoids which are extensively used in the design and analysis of interior-point methods in optimisation.
A commonality of all the aforementioned methods is that they only require calls to the (unnormalised) density of the target distribution.

To see how the gradients of \(\potential{}\) could be useful, consider the unconstrained setting where \(\primalspace = \bbR^{d}\).
A popular algorithm for sampling in this setting is the unadjusted Langevin algorithm (\ref{eq:ULA}), which is the Euler-Maruyama discretisation of the continuous-time Langevin dynamics (\ref{eq:LD}).
\begin{align*}
    dX_{t} &= -\nabla \potential{X_{t}} dt + \sqrt{2}~dB_{t}
    \tag{\textsf{LD}}\label{eq:LD} \\
    X_{k + 1} - X_{k} &= -h \cdot \gradpotential{X_{k}} + \sqrt{2h} \cdot \xi_{k}; \quad \xi_{k} \sim \calN(0, I_{d})\tag{\textsf{ULA}} \label{eq:ULA}~.
\end{align*}
Without the Brownian motion term \(dB_{t}\), \ref{eq:LD} resembles the gradient flow of \(\potential{}\), and analogously, setting \(\xi_{k} = \bm{0}\) in \ref{eq:ULA} recovers gradient descent with step size \(h\).
In the seminal work of \cite{jordan1998variational}, the Langevin dynamics was shown to be the gradient flow of the KL divergence in the space of probability measures equipped with the Wasserstein metric via the Fokker-Planck equation -- thus making the connection between sampling and optimisation more substantive; see also the paper by \cite{wibisono2018sampling}.
However, in the constrained sampling problem described earlier, the Langevin dynamics (or algorithm) is not applicable since the solutions (or iterates) are no longer guaranteed to remain in \(\primalspace\), which we recall is a compact and convex subset of \(\bbR^{d}\).
The Mirror Langevin dynamics (\ref{eq:MLD}) \citep{zhang2020wasserstein,chewi2020exponential} was proposed to address this inapplicability, which draws from the idea of mirror descent in optimisation.
Specifically, an invertible \emph{mirror map} is used, which is often defined as the gradient of a barrier function (also called the \emph{mirror function}) \(\mirrorfunc{}\) over \(\primalspace\).
A closely related dynamics was proposed in \cite{hsieh2018mirrored} which is also termed Mirror Langevin dynamics, but this is not the focus of this work -- for a detailed comparison between the two dynamics, see \citet[\S 1.2]{zhang2020wasserstein}, and our interest is the dynamics defined as follows.
\begin{equation*}
\label{eq:MLD}
    Y_{t} = \primtodual{X_{t}}; \quad dY_{t} = -\gradpotential{X_{t}} dt + \sqrt{2}~\hessmirror{X_{t}}^{\nicefrac{1}{2}} ~dB_{t}~\tag{\textsf{MLD}}.
\end{equation*}
As noted for the Langevin dynamics in \ref{eq:LD}, without the Brownian motion term \(dB_{t}\), \ref{eq:MLD} resembles the original mirror dynamics as proposed by \cite{nemirovskii1983problem} for optimisation.
Fundamentally, the mirror function serves two key purposes: (1) it changes the geometry of the primal space (\(\primalspace\)) suitably; the underlying metric is given by the Hessian of the mirror function, and (2) it allows us to perform unconstrained \ref{eq:LD}-style diffusion in the dual space (the range of the mirror map) which is unconstrained when the domain is bounded.
Due to the occurrence of the time-varying diffusion matrix in \ref{eq:MLD}, the dynamics can be discretised in a variety of ways.
The most basic discretisation (or also called the Euler-Maruyama discretisation) of \ref{eq:MLD} is commonly referred to as the Mirror Langevin algorithm (\ref{eq:MLA}) \citep{zhang2020wasserstein,li2022mirror}
\begin{equation*}
    \left\{
    \begin{aligned}
    Y_{k} &= \primtodual{X_{k}}~; \\
    Y_{k + 1} - Y_{k} &= -h\cdot \gradpotential{X_{k}} + \sqrt{2h} \cdot \hessmirror{X_{k}}^{\nicefrac{1}{2}}~\xi_{k}~, \qquad \xi_{k} \sim \calN(0, I_{d})~; \\
    X_{k + 1} &= (\primtodual{})^{-1}(Y_{k + 1})~.
    \end{aligned}
    \right.
    \tag{\textsf{MLA}}\label{eq:MLA}
\end{equation*}
Other discretisations have since been proposed and analysed \citep{ahn2021efficient,jiang2021mirror}; however, the advantage of \ref{eq:MLA} over these other discretisations is that \ref{eq:MLA} can be implemented exactly, whereas the latter methods involve simulating a stochastic differential equation (SDE) which cannot be performed exactly, and has to be approximated for practical purposes.

Theoretically, for \ref{eq:MLA}, the distribution of \(X_{k}\) as \(k \to \infty\) is not guaranteed to coincide with \(\Pi\), and such an algorithm is said to be \emph{biased}, where the distance between this limit and \(\Pi\) is termed the \emph{bias}.
This phenomenon is not specific to \ref{eq:MLA}; for instance, the more popular \ref{eq:ULA} is also biased, and this bias is shown to scale with the step size \(h\) under relatively weak assumptions made on the target distribution -- see \cite{durmus2017nonasymptotic,dalalyan2017theoretical,dalalyan2017further,dalalyan2019user,cheng2018convergence,vempala2019rapid,li2022sqrt} for a variety of analyses of \ref{eq:ULA}.
Interestingly, due to the dependence on \(h\), the bias of \ref{eq:ULA} tends to \(0\) as \(h \to 0\) (thus termed a ``vanishing bias'').
In the case of \ref{eq:MLA} however, the first analysis by \cite{zhang2020wasserstein} suggested that \ref{eq:MLA} had a non-vanishing bias, and conjectured that this was unavoidable.
Notably, their analysis assumed a \emph{modified self-concordance condition} on \(\mirrorfunc{}\), which is sufficient to ensure the existence of strong solutions for the continuous-time dynamics (\ref{eq:MLD}).
In contrast, a more natural assumption placed on \(\mirrorfunc{}\) is \emph{self-concordance}, especially considered in optimisation.
A subsequent but different analysis of \ref{eq:MLA} by \cite{jiang2021mirror} assumed the more standard self-concordance condition over \(\mirrorfunc{}\), but their analysis was also unable to improve the non-vanishing bias of \ref{eq:MLA}.
They also presented an analysis of an alternate yet impractical discretisation of \ref{eq:MLD} developed and analysed by \cite{ahn2021efficient} that we refer to as \MLAFD{}, and independent analyses in both of these works showed that \MLAFD{} has a vanishing bias when the more natural self-concordance condition over \(\mirrorfunc{}\) is considered.
More recently, \cite{li2022mirror} improved the result due to \citet{zhang2020wasserstein}, and showed that \ref{eq:MLA} has a vanishing bias when \(\mirrorfunc{}\) satisfies the modified self-concordance condition.
This newer result is based on a mean-square analysis technique developed in \cite{li2019stochastic}, for which the modified self-concordance condition is more amenable.
However, this modified self-concordance condition does not imply some of the more desirable properties that self-concordance functions satisfy, for example \emph{affine invariance}, which we discuss later.

\subsubsection*{A summary of our work}
Despite these recent results, a question that has been left unanswered is if \ref{eq:MLA}, or any other exactly implementable algorithm based on discretising \ref{eq:MLD} has a vanishing bias when \(\mirrorfunc{}\) is simply self-concordant.
In this paper, we give an algorithm that is \emph{unbiased}, and with provable guarantees while assuming self-concordance of \(\mirrorfunc{}\).
Our proposed algorithm applies a Metropolis-Hastings filter to the proposal induced by a single step of \ref{eq:MLA} -- we call this the \emph{Metropolis-adjusted Mirror Langevin algorithm} (or \nameref{alg:mamla} in short).
Each iteration of the algorithm is composed of three key steps, and these are formalised in \Cref{alg:mamla} in the sequel.
\begin{enumerate}[leftmargin=*]
    \setlength{\itemsep}{0pt}
    \item Generate a proposal \(Z_{k}\) from iterate \(X_{k}\) using a single step of \ref{eq:MLA}.
    \item Compute the Metropolis-Hastings acceptance probability \(p_{\mathrm{accept}}\).
    \item With probability \(p_{\mathrm{accept}}\), set \(X_{k + 1} = Z_{k}\) (accept); otherwise set \(X_{k + 1} = X_{k}\) (reject).
\end{enumerate}
In principle, this is similar to the Metropolis-adjusted Langevin algorithm (\MALA{}) \citep{roberts1996exponential}, which applies a Metropolis-Hastings filter to the proposal induced by a single step of \ref{eq:ULA}.
Due to the form of \ref{eq:MLA}, \nameref{alg:mamla} is exactly implementable much like \MALA{}, which crucially relies on the ability to compute the proposal densities.
Moreover by construction, the inclusion of a Metropolis-Hastings filter results in a Markov chain that is reversible with respect to the target distribution.
Consequently, the algorithm converges to the target distribution (and is therefore unbiased) unlike \ref{eq:MLA} and other discretisations of \ref{eq:MLD}, and does so exponentially quickly.
In particular, let \(\potential{}\) be \(\mu\)-strongly convex, \(\lambda\)-smooth, and \(\beta\)-Lipschitz relative to a self-concordant mirror function \(\mirrorfunc{}\) (see \Cref{sec:func-classes} for formal definitions).
We show in \Cref{thm:mix-mamla} that
\begin{enumerate}[leftmargin=*]
    \setlength{\itemsep}{0pt}
    \item when \(\mu > 0\), the \(\delta\)-mixing time of \nameref{alg:mamla} scales as \(\calO\left(\frac{1}{\mu} \cdot \max\left\{d^{3}, d\lambda, \beta^{2}\right\} \cdot \log\left(\frac{1}{\delta}\right)\right)\), and
    \item when \(\mu = 0\), the \(\delta\)-mixing time of \nameref{alg:mamla} scales as \(\calO\left(\nu \cdot \max\{d^{3}, d\lambda, \beta^{2}\} \cdot \log\left(\frac{1}{\delta}\right)\right)\), where \(\nu\) is a constant that depends on the structure of \(\primalspace\) as induced by \(\mirrorfunc{}\).
\end{enumerate}
We obtain these guarantees through the classical one-step overlap technique due to \cite{lovasz1993random}, which was also used to give mixing time guarantees for \MALA{} in \cite{dwivedi2018log,chewi2021optimal}.
Our analysis is however not a direct consequence of these aforementioned results specifically due to the occurrence of the mirror map \(\primtodual{}\) and metric \(\hessmirror{}\) in \ref{eq:MLA}, which poses certain difficulties.
We handle these newer quantities by strongly leveraging the self-concordant nature of the mirror function \(\mirrorfunc{}\) and use recent isoperimetry results for distributions supported on a Hessian manifold induced by a self-concordant function \citep{gopi2023algorithmic} to provide guarantees for \nameref{alg:mamla}.
We compare these mixing time upper bounds to those obtained for other unadjusted discretisations of \ref{eq:MLD} in more detail later this work, but here briefly remark on the mixing time upper bound for \ref{eq:MLA} shown in \cite{li2022mirror} which scales as \(\calO(\nicefrac{d\lambda^{2}}{\mu^{3}\delta^{2}})\), while assuming that \(\mirrorfunc{}\) satisfies a modified self-concordance condition.

\paragraph*{A special case where \(\mirrorfunc{} = \potential{}\)}
In this case, the Mirror Langevin dynamics (\ref{eq:MLD}) specialises to the \emph{Newton Langevin dynamics} (\ref{eq:NLD}), which is the sampling analogue of the continuous time version of Newton's method in optimisation.
\begin{equation*}
\label{eq:NLD}
    Y_{t} = \gradpotential{X_{t}}; \quad dY_{t} = -\gradpotential{X_{t}} dt + \sqrt{2}~\nabla^{2}\potential{X_{t}}^{\nicefrac{1}{2}} ~dB_{t}~\tag{\textsf{NLD}}.
\end{equation*}
\citet{chewi2020exponential} showed that when \(\potential{}\) is strictly convex over \(\primalspace\), \ref{eq:NLD} converges exponentially quickly to the target in \(\chi^{2}\)-divergence at a rate that is invariant to affine transformations of \(\primalspace\).
This result is due to the Brascamp-Lieb inequality.
Hence, it is reasonable to expect that the Newton Langevin algorithm, which is the specialisation of the Mirror Langevin algorithm (\ref{eq:MLA}) with \(\mirrorfunc{} = \potential{}\), also has mixing time guarantees that are invariant to affine transformations.
However, specialising the most recent analysis of the Mirror Langevin algorithm in \cite{li2022mirror} with \(\mirrorfunc{} = \potential{}\) is unable to show this, owing to the modified self-concordance assumption, which we elaborate on later in \cref{sec:intro:cond-num-indep}.
Despite this, our analysis does indeed yield a guarantee for the \emph{Metropolis-adjusted Newton Langevin algorithm} that is invariant to affine transformations of \(\primalspace\).

\subsection{Related work}
The original version of \textsf{DikinWalk} \citep{kannan2009random} was developed with the goal of sampling uniformly from polytopes, and the iterative algorithm was based on generating either uniform samples from Dikin ellipsoids centered at the iterates, or samples from a Gaussian distribution centered at the iterates with a covariance given by the Dikin ellipsoid \citep{sachdeva2016mixing,narayanan2016randomized}.
Other related algorithms are the \textsf{JohnWalk} \citep{gustafson2018john}, \textsf{VaidyaWalk} \citep{chen2018fast}, and \textsf{WeightedDikinWalk} \citep{laddha2020strong}.
\cite{narayanan2017efficient} modify the filter in \textsf{DikinWalk} to make it amenable for the more general constrained sampling problem, and was recently modified by \cite{mangoubi2023sampling} in proposing a modified \textsf{Soft-Threshold DikinWalk}, which primarily focuses on the case where \(\primalspace\) is a polytope.
\cite{kook2023efficiently} develop a broader theory of interior point sampling methods while generalising the ideas from the analysis of \textsf{DikinWalk} in \cite{laddha2020strong}.
As remarked in \cite{narayanan2017efficient}, while the existence of such a self-concordant barrier is a stronger assumption than a separation oracle for \(\primalspace\), the self-concordant barrier enables leveraging the geometry of \(\primalspace\) better.

The aforementioned methods related to \textsf{DikinWalk} do not require access to gradients of \(\potential{}\).
Drawing inspiration from projected gradient descent for optimising functions over constrained feasibility sets, \cite{bubeck2018sampling} propose a projected Langevin algorithm, where each step of \ref{eq:ULA} is projected onto the domain.
Another class of approaches \citep{brosse2017sampling,gurbuzbalaban2022penalized} propose obtaining a good approximation \(\widetilde{\targetdist}\) of the target \(\targetdist\) such that the support of \(\widetilde{\targetdist}\) is \(\bbR^{d}\), but this does not eliminate the likelihood of iterates lying outside \(\primalspace\).
This enables running the simpler unadjusted Langevin algorithm over \(\widetilde{\targetdist}\) to obtain approximate samples from \(\targetdist\).
\textsf{MALA} \citep{roberts1996exponential,dwivedi2018log,chewi2021optimal}, or Hamiltonian Monte Carlo (\textsf{HMC}) \citep{neal2011mcmc,durmus2017convergence,bou2020coupling} is also a viable option for sampling from log-concave distributions over \(\primalspace\); however, guarantees for these methods assume Euclidean geometry of \(\primalspace\), which poses difficulties.
Riemannian Hamiltonian Monte Carlo (\textsf{RHMC}) \citep{girolami2011riemann,lee2018convergence} extends \textsf{HMC} to allow for sampling in such settings and with provable guarantees, and this has recently been modified to handle constrained domains better in \cite{kook2022sampling,noble2023unbiased}, while assuming that there exists a self-concordant barrier for \(\primalspace\) that is computable.
\cite{girolami2011riemann} also proposed the Riemannian Langevin algorithm as a generalisation of \ref{eq:ULA}, and this was subsequently analysed in \citet{gatmiry2022convergence,li2023riemannian}.
Relatedly, drawing from recent developments in proximal methods for sampling, \citet{gopi2023algorithmic} propose a novel proximal sampler for sampling over non-Euclidean spaces based on the log-Laplace transform.
While this proximal sampler has desirable and interesting theoretical properties, it is yet to be experimented with in practice.

\vspace*{-1mm}
\paragraph*{\textsc{Condition number independence}}
\label{sec:intro:cond-num-indep}
In this context, the \emph{condition number} is a property of the domain \(\primalspace\) and is defined as \(\calC_{\primalspace} = \frac{R}{r}\), where \(R\) is the radius of the smallest ball containing \(\primalspace\) and \(r\) is the radius of the largest ball contained entirely in \(\primalspace\) \citep{kannan2009random}.
A Euclidean ball in \(\bbR^{d}\) of arbitrary radius has condition number \(1\), and non-unitary affine transformations of this ball results in sets with different condition numbers.
Constrained sampling methods over \(\primalspace\) that have mixing time guarantees that are independent of \(\calC_{\primalspace}\) are valuable, and this is because the complexity of approximate sampling from distributions with affine invariant properties over a possibly ill-conditioned domain \(\primalspace\) is the same as approximate sampling from similar distributions over a well-conditioned (or \emph{isotropic}) domain like a ball.
It is important to note that \(\calC_{\primalspace}\) is not the same as the value \(\kappa = \nicefrac{\lambda}{\mu}\); the latter corresponds to the conditioning of the potential relative to a mirror function defined over \(\primalspace\).
The dependence on \(\kappa\) is expected; intuitively, it is harder to sample from sharper / peakier distributions which correspond to a higher value of \(\kappa\), even when supported over a well-conditioned domain.

The classical Hit-and-Run and ball walk algorithms for approximate uniform sampling from \(\primalspace\) are not affine invariant, and the mixing time of these algorithms scaled as \(\calC_{\primalspace}^{2}\).
The Dikin walk is an essential improvement over these methods because its mixing time guarantees are independent of \(\calC_{\primalspace}\), and this is primarily due to the involvement of the Dikin ellipsoids.
There is continued interest in developing methods whose mixing time guarantees are independent of \(\calC_{\primalspace}\): recent analyses of \textsf{RHMC} have shown to satisfy this \citep{lee2018convergence,kook2023condition}.
As far as discretisations of \ref{eq:MLD} are concerned, the analysis of \MLAFD{} in \cite{ahn2021efficient} gives mixing time guarantees that depend on constants in the assumptions made.
Notably, these assumptions are affine invariant, which implies these constants are independent of \(\calC_{\primalspace}\), and results in mixing time guarantees that are \emph{independent} of \(\calC_{\primalspace}\).
On the other hand, the recent analysis of \ref{eq:MLA} due to \citet{li2022mirror} assumes a \emph{modified self-concordance} condition.
This is not affine invariant in the sense that the parameter in this condition can change when the domain is transformed by an affine map \citep[\S D]{li2022mirror}.
Put simply, this parameter could depend on \(\calC_{\primalspace}\).
As this parameter appears in their mixing time guarantee,  it suggests that the mixing time of \ref{eq:MLA} is also possibly dependent on \(\calC_{\primalspace}\).
Our mixing time guarantees for \nameref{alg:mamla} crucially does not depend on \(\calC_{\primalspace}\) much like the guarantees for \MLAFD{}, and this is due to the affine invariance of the assumptions we make for the analysis.

\vspace*{-1mm}
\paragraph*{Organisation}
The remainder of this paper is organised as follows: we first review notation and definitions that are key for this work in \cref{sec:prelims}. In \cref{sec:algo}, we define the algorithm (\nameref{alg:mamla}) formally, discuss its implementation, provide a general mixing time guarantee for it as well as some corollaries of this general guarantee for a variety of sampling tasks.
We showcase some numerical experiments to demonstrate \nameref{alg:mamla} in \cref{sec:expts}.
In \cref{sec:proofs}, we provide proofs for the theoretical statements made in this work, and finally conclude with a discussion of our results in \cref{sec:conclusion}.

\section{Preliminaries}
\label{sec:prelims}
We begin by introducing some general notation that will be used throughout this work.
We remind the reader that \(\primalspace\) is a compact, convex subset of \(\bbR^{d}\).

\paragraph*{Notation}
The set \(\{1, \ldots, m\}\) is denoted by \([m]\).
We denote the set of positive reals by \(\bbR_{+}\).
Let \(A \in \bbR^{d \times d}\) be a symmetric positive definite matrix.
For \(x, y \in \bbR^{d}\), we define \(\langle x, y\rangle_{A} = \langle x, Ay\rangle\), and \(\|x\|_{A} = \sqrt{\langle x, x\rangle_{A}}\).
When the subscript is omited as in \(\|x\|\), then this corresponds to \(A = I_{d \times d}\), or the Euclidean norm of \(x\).
For a set \(\calA\), \(\interior{\calA}\) denotes its interior, and collection of all measurable subsets of \(\calA\) is denoted by \(\calF(\calA)\).
For a measurable map \(T : \calA \to \calB\) and a distribution \(P\) supported on \(\calA\), we denote \(T_{\#}P\) to be the pushforward measure of \(P\) using \(T\); in other words, \(T_{\#}P\) is the law of \(T(x)\) where \(x\) is distributed according to \(P\), and \(T_{\#}P\) has support \(\calB\).
Unless specified explicitly, the density of a distribution \(P\) (if it exists) at \(x\) is denoted by \(dP(x)\).

\subsection{Function classes}
\label{sec:func-classes}

\paragraph*{Legendre type}
For both the mirror descent algorithm in optimisation and \ref{eq:MLA}, a method to compute the inverse of \(\primtodual{}\) is required, as respective updates are made in the set given by range of \(\primtodual{}\) (termed the \emph{dual space}), and then mapped back to \(\primalspace\), which can appear difficult at first glance.
When the mirror function \(\mirrorfunc{}\) is of \emph{Legendre type}, this is less daunting, since \((\primtodual{})^{-1}\) can be computed in terms of the \emph{convex conjugate} of \(\mirrorfunc{}\).
A function with domain \(\primalspace\) is of Legendre type if it is differentiable and strictly convex in \(\interior{\primalspace}\), and has gradients that become unbounded as we approach the boundary of \(\primalspace\); a precise definition of Legendre type functions can be found in \citet[Chap. 26]{rockafellar1997convex}.
The convex conjugate\footnote{also referred to as the \emph{Fenchel-Legendre dual}.} of \(\mirrorfunc{}\) denoted by \(\mirrorfunc{}^{*}\) is defined as
\begin{equation*}
    \mirrorfunc{}^{*}(y) = \max_{x \in \primalspace} ~\langle x, y\rangle - \mirrorfunc{x}~.
\end{equation*}
The domain of \(\mirrorfunc{}^{*}\) is the range of \(\primtodual{}\).
Crucially, when \(\mirrorfunc{}\) is of Legendre type, then \(\primtodual{}\) is an invertible map between \(\primalspace\) and the domain of \(\mirrorfunc{}^{*}\), and the inverse of this map is \(\dualtoprim{}\).
For a given \(y\), the maximiser \(\hat{x}\) which achieves \(\mirrorfuncdual{y} = \langle \hat{x}, y\rangle - \mirrorfunc{\hat{x}}\) by optimality satisfies \(y = \primtodual{\hat{x}}\), or \(\hat{x} = \dualtoprim{y}\).
When it is possible to express \(\mirrorfuncdual{}\) in closed form, \(\dualtoprim{}\) can be computed as well.
Otherwise, \(\dualtoprim{y}\) can be estimated by approximately solving the maximisation problem in the definition of \(\mirrorfunc{}^{*}\).
Additionally, Corollary 13.3.1 from \citet{rockafellar1997convex} posits that when \(\primalspace\) is a bounded subset of \(\bbR^{d}\), the domain of \(\mirrorfunc{}^{*}\) is \(\bbR^{d}\).
This implies that for any \(y \in \bbR^{d}\), there exists a unique \(x \in \interior{\primalspace}\) such that \(\dualtoprim{y} = (\primtodual{})^{-1}(y) = x\).
This fact is essential for \ref{eq:MLA} as the Gaussian random vector \(\xi_{k}\) in \ref{eq:MLA} is unrestricted, and \(X_{k + 1}\) could be undefined unless the domain of \(\mirrorfunc{}^{*} = \bbR^{d}\).
Henceforth in this work, the mirror function \(\mirrorfunc{}\) is assumed to be of \emph{Legendre type}; this is an assumption also made in prior work that study \ref{eq:MLD} and its discretisations (\ref{eq:MLA}, \MLAFD{}).

\paragraph*{Self-concordance}
Self-concordant functions are ubiquitous in the study and design of interior-point methods in optimisation, and are defined as follows.
\begin{definition}[{\citet[\S 5.1.3]{nesterov2018lectures}}]
\label{def:self-concord}
A thrice differentiable strictly convex function \(\psi : \primalspace \to \bbR \cup \{\infty\}\) is said to be self-concordant with parameter \(\alpha \geq 0\) if for all \(x \in \interior{\primalspace}\) and \(u \in \bbR^{d}\)~,
\begin{equation*}
    |\nabla^{3}\psi(x)[u, u, u]| \leq 2\alpha\|u\|_{\nabla^{2}\psi(x)}^{3}~.
\end{equation*}
\end{definition}

\paragraph*{Relative convexity and smoothness}
These classes of functions can be viewed as generalisations of convex and smooth functions with different geometrical properties, and were independently studied in \cite{bauschke2017descent} and \cite{lu2018relatively}.

\begin{definition}
\label{def:rel-convex}
Let \(g,~\psi : \primalspace \to \bbR \cup \{\infty\}\) be differentiable convex functions.
We say that \(g\) is \(\mu\)-relatively convex with respect to \(\psi\) if \(g - \mu \cdot \psi\) is convex, where \(\mu \geq 0\).
Equivalently, when \(g,~\psi\) are twice differentiable, then \(\mu \cdot \nabla^{2}\psi(x) \preceq \nabla^{2}g(x)\) for all \(x \in \interior{\primalspace}\)~.
\end{definition}

\begin{definition}
\label{def:rel-smooth}
Let \(g,~\psi : \primalspace \to \bbR \cup \{\infty\}\) be differentiable convex functions.
We say that \(g\) is \(\lambda\)-relatively smooth with respect to \(\psi\) if \(\lambda \cdot \psi - g\) is convex, where \(\lambda \geq 0\).
Equivalently, when \(g,~\psi\) are twice differentiable, then \(\lambda \cdot \nabla^{2}\psi(x) \succeq \nabla^{2}g(x)\) for all \(x \in \interior{\primalspace}\)~.
\end{definition}

Let \(\psi : \primalspace \to \bbR \cup \{\infty\}\) be a differentiable convex function.
The Bregman divergence of \(\psi\) at \(y\) with respect to \(x\) is defined as \(D_{\psi}(y; x) = \psi(y) - \psi(x) - \langle \nabla\psi(x), y- x\rangle\).
The \(\mu\)-relative convexity and \(\lambda\)-relative smoothness of a function \(g\) with respect to \(\psi\) imply respectively that for any \(x, y \in \interior{\primalspace}\),
\begin{align*}
    g(y) &\geq g(x) + \langle \nabla g(x), y - x\rangle + \mu \cdot D_{\psi}(y; x)~, \\
    g(y) &\leq g(x) + \langle \nabla g(x), y - x\rangle + \lambda \cdot D_{\psi}(y; x)~.
\end{align*}

When \(\psi(x) = \frac{\|x\|^{2}}{2}\), and \(\primalspace = \bbR^{d}\), \(D_{\psi}(y; x) = \frac{\|y - x\|^{2}}{2}\) for any \(x, y \in \bbR^{d}\), and substituting the equation for \(D_{\psi}\) in the inequalities above recovers the standard first order definitions of convexity and smoothness \citep[\S 2.1]{nesterov2018lectures}.

\paragraph*{Relative Lipschitz continuity}

This class of functions is a generalisation of Lipschitz continuity of a differentiable function, and has been useful in the analysis of \MLAFD{} \citep{ahn2021efficient}.

\begin{definition}[{\cite{jiang2021mirror,ahn2021efficient}}]
\label{def:rel-lipschitz}
Let \(g : \primalspace \to \bbR \cup \{\infty\}\) be a differentiable function.
We say that \(g\) is \(\beta\)-relatively Lipschitz continuous with respect to a twice differentiable strictly convex function \(\psi : \primalspace \to \bbR \cup \{\infty\}\) with parameter \(\beta\) if for all \(x \in \interior{\primalspace}\), it holds that
\begin{equation*}
    \|\nabla g(x)\|_{\nabla^{2}\psi(x)^{-1}} \leq \beta \enskip \Leftrightarrow \enskip \beta^{2} \cdot \nabla^{2}\psi(x) \succeq \nabla g(x)\nabla g(x)^{\top}~.
\end{equation*}
\end{definition}

Similar to the relative convexity and smoothness definition, when \(\psi(x) = \frac{\|x\|^{2}}{2}\) with \(\primalspace = \bbR^{d}\), \(\beta\)-relative Lipschitz continuity of \(g\) with respect to \(\psi\) reduces to \(\|\nabla g(x)\| \leq \beta\), which is an equivalent characterisation of Lipschitz continuity of \(g\).
A special case to make note of is when \(\psi = g\), and a function \(g\) that satisfies \(\|\nabla g(x)\|_{\nabla^{2}g(x)^{-1}} \leq \beta\) for all \(x \in \interior{\primalspace}\) is termed a \emph{barrier function} \cite[\S 5.3.2]{nesterov2018lectures}.
This property is useful in the analysis of Newton's method in optimisation.
To avoid clashing with this terminology, we will strictly refer to \(\mirrorfunc{}\) as the \emph{mirror function}.

\textbf{Remark}~~
The above properties: self-concordance, relative convexity, relative smoothness and relative Lipschitz continuity are invariant to affine transformations of the domain.
More precisely, let \(T_{\textrm{Aff}}\) be an affine transformation of suitable dimensions.
\begin{itemize}[leftmargin=*]
\setlength{\itemsep}{0pt}
\item If \(\psi\) is a self-concordant function with parameter \(\alpha\), then \(\psi \circ T_{\textrm{Aff}}\) is also a self-concordant function with parameter \(\alpha\) \citep[Thm. 5.1.2]{nesterov2018lectures}.
\item If \(g\) is a \(\mu\)-relatively convex (or \(\lambda\)-relatively smooth, or \(\beta\)-relatively Lipschitz continuous) function with respect to \(\psi\), then \(g \circ T_{\mathrm{Aff}}\) is also a \(\mu\)-relatively convex (or \(\lambda\)-relatively smooth, or \(\beta\)-relatively Lipschitz continuous, resp.) function with respect to \(\psi \circ T_{\textrm{Aff}}\) (see \citet[Prop. 1.2]{lu2018relatively}; \citet[Thm. 5.3.3]{nesterov2018lectures}).
\end{itemize}

\paragraph*{Symmetric Barrier}
Symmetric barriers were introduced in \cite{laddha2020strong} where it was used to obtain mixing time bounds for the Dikin walk and weighted variant.
This property was originally introduced in \cite{gustafson2018john} in the development of the John walk.

\begin{definition}[{\citet[Def. 2]{laddha2020strong}}]
\label{def:sym-barrier}
Let \(\psi : \primalspace \to \bbR \cup \{\infty\}\) be a twice differentiable function, and let \(\calE_{x}^{\psi}(r) = \{y : \|y - x\|_{\nabla^{2}\psi(x)} \leq r\}\) be the Dikin ellipsoid of radius \(r\) centered at \(x\).
We say that \(\psi\) is a symmetric barrier with parameter \(\nu > 0\) if for all \(x \in \interior{\primalspace}\)~,
\begin{equation*}
    \calE_{x}^{\psi}(1) \subseteq \primalspace \cap (2x - \primalspace) \subseteq \calE_{x}^{\psi}(\sqrt{\nu})~.
\end{equation*}
\end{definition}
Examples of such barriers are the log-barrier of a polytope formed by \(m\) constraints is a symmetric barrier with parameter \(\nu = m\) (\Cref{lem:log-barrier-polytope-sym-barrier}), and the log-barrier of an ellipsoid with unit radius is also a symmetric barrier with parameter \(\nu = 2\) (\Cref{lem:log-barrier-ellipsoid-sym-barrier}).
Notably, the symmetric barrier parameters of these log-barriers is independent of the conditioning of these sets.

\subsection{Markov chains, conductance and mixing time}
\label{sec:prelims-markov-chains}

\paragraph*{Markov chains}

Let the domain of interest be \(\primalspace\).
A (time-homogeneous) Markov chain over \(\primalspace\) is characterised by a set of transition kernels \(\bfP = \{\calP_{x} : x \in \primalspace\}\), where \(\calP_{x}\) is the one-step distribution that maps measurable subsets of \(\primalspace\) to non-negative values.
With this setup, we have a transition operator \(\bbT_{\bfP}\) on the space of probability measures defined by
\begin{equation*}
    (\bbT_{\bfP}\mu)(S) = \int_{\primalspace} \calP_{y}(S) \cdot d\mu(y) \qquad \forall~S \in \calF(\primalspace)~.
\end{equation*}
The distribution after \(k\) applications of the transition operator to \(\mu\) is denoted by \(\bbT_{\bfP}^{k}\mu\).
A probability measure \(\pi\) is called a \emph{stationary} measure of a Markov chain \(\bfP\) if \(\pi = \bbT_{\bfP}\pi\).
A Markov chain \(\bfP\) is said to be \emph{reversible} with respect to a measure \(\pi\) if for any \(A, B \in \calF(\primalspace)\),
\begin{equation*}
    \int_{A} \calP_{x}(B) \cdot d\pi(x) = \int_{B} \calP_{y}(A) \cdot d\pi(y)~.
\end{equation*}
If \(\bfP\) is reversible with respect to \(\pi\), then \(\pi\) is a stationary measure of \(\bfP\); this can be checked by substituting the equality due to reversibility with \(A = \calK, ~B = S\).

\paragraph*{Conductance, total variation distance and mixing time}

The \emph{conductance} of a Markov chain quantifies how likely a Markov chain would visit low probability regions; thus the mixing of Markov chains depends on how low its conductance is.
Formally, the conductance of a Markov chain \(\bfP = \{\calP_{x} : x \in \calK\}\) with stationary measure \(\pi\) supported on \(\primalspace\) is defined as
\begin{equation*}
    \Phi_{\bfP} = \inf_{A \in \calF(\primalspace)}~ \frac{1}{\min\{\pi(A), 1 - \pi(A)\}} \int_{x \in A} \calP_{x}(\primalspace \setminus A) \cdot d\pi(x)~.
\end{equation*}

To measure how quickly a Markov chain mixes to its stationary distribution, we use the \emph{total variation (TV) distance}, which is a metric in the space of probability measures.
The TV distance between two distributions \(\mu\) and \(\nu\) with support \(\primalspace\) is defined as
\begin{equation*}
    \TVdist(\mu, \nu) = \sup_{A \in \calF(\primalspace)} \mu(A) - \nu(A)~.
\end{equation*}
Let \(\bfP\) be a Markov chain with reversible distribution \(\pi\).
For \(\delta \in (0, 1)\), the \emph{\(\delta\)-mixing time} from an initial distribution \(\mu_{0}\) of \(\bfP\), denoted by \(\mixingtime{\delta; \bfP, \mu_{0}}\), is defined as the least number of applications of \(\bbT_{\bfP}\) to \(\mu_{0}\) to achieve a distributions that is at most \(\delta\) away from \(\pi\).
Formally,
\begin{equation*}
    \mixingtime{\delta; \bfP, \mu_{0}} = \inf \{k \geq 0: \TVdist(\bbT_{\bfP}^{k}\mu_{0}, \pi) \leq \delta\}~.
\end{equation*}

Finally, we introduce the notion of a \emph{warm distribution}, which is useful for obtaining mixing time guarantees.
A distribution \(\mu\) supported on \(\primalspace\) is said to be a \(M\)-warm with respect to another distribution \(\Pi\) also supported on \(\primalspace\) if
\begin{equation*}
    \sup_{A \in \calF(\primalspace)} \frac{\mu(A)}{\Pi(A)} = M~.
\end{equation*}

\section{Metropolis-adjusted Mirror Langevin algorithm}
\label{sec:algo}
\begin{algorithm}[t]
\DontPrintSemicolon

\caption{Metropolis-adjusted Mirror Langevin algorithm (\MAMLA{})}
\algotitle{\MAMLA{}}{alg:mamla}

\SetKwInOut{Input}{Input}
\SetKwInOut{Output}{Output}
\Input{Potential \(\potential{} : \primalspace \to \bbR\), mirror function \(\mirrorfunc{} : \primalspace \to \bbR \cup \{\infty\}\), iterations \(K\), initial distribution \(\Pi_{0}\), step size \(h > 0\)}

Sample \(x_{0} \sim \Pi_{0}\).

\For{\(k \leftarrow 0\) \KwTo \(K - 1\)}{
    Sample a random vector \(\xi_{k} \sim \calN(0, I)\).

    Generate proposal \(z = \dualtoprim{\primtodual{x_{k}} - h \cdot \gradpotential{x_{k}} + \sqrt{2h \cdot \hessmirror{x_{k}}} ~\xi_{k}}\).\label{alg:proposal-step}

    Compute acceptance ratio \(p_{\mathrm{accept}}(z; x_{k}) = \min\left\{1, \frac{\pi(z) p_{z}(x_{k})}{\pi(x_{k}) p_{x_{k}}(z)}\right\}\) using \cref{eq:proposal-density}.\label{alg:acceptance-ratio-step}

    Obtain \(U \sim \mathrm{Unif}([0, 1])\).

    \eIf{\(U \leq p_{\mathrm{accept}}(z; x_{k})\)}{
        Set \(x_{k + 1} = z\).
    }{
        Set \(x_{k + 1} = x_{k}\).
    }
}

\Output{\(x_{K}\)}
\end{algorithm}

In this section, we introduce the Metropolis-adjusted Mirror Langevin algorithm (\nameref{alg:mamla}).
Let \(\bfP\) be a Markov chain which defines a collection of proposal distributions at each \(x \in \primalspace\) with densities \(\{p_{x} : x \in \primalspace\}\).
The \emph{acceptance ratio} of \(z\) with respect to \(x\) (given a target density \(\targetdens\)) is defined as
\begin{equation}
\label{eq:accept-ratio-def}
    p^{\bfP}_{\mathrm{accept}}(z; x) = \min\left\{1, \frac{\targetdens(z) p_{z}(x)}{\targetdens(x) p_{x}(z)}\right\}~.
\end{equation}
We recap the general outline of the Metropolis-adjustment of a Markov chain from \cref{sec:intro}.
From a point \(x\), the Markov chain \(\bfP\) generates a proposal \(z \sim \calP_{x}\).
This proposal \(z\) is accepted to be the next iterate with probability \(p^{\bfP}_{\mathrm{accept}}(z; x)\), and if not accepted, \(x\) is retained.
Let \(\bfT \defeq \{\calT_{x} : x \in \primalspace\}\) denote the Metropolis-adjusted Markov chain, where \(\calT_{x}\) is the one-step distribution after this adjustment.
As noted earlier, by construction, \(\bfT\) is reversible with respect to \(\targetdist\) with density \(\targetdens\).

In our setting, the Markov chain \(\bfP\) is induced by one step of \ref{eq:MLA}.
Let \(\xi\) be an independent standard normal vector in \(d\) dimensions.
From any \(x \in \interior{\primalspace}\), one step of \ref{eq:MLA} returns the point
\begin{equation*}
    x' = \dualtoprim{\primtodual{x} - h \cdot \gradpotential{x} + \sqrt{2h \cdot \hessmirror{x}} ~\xi}~.
\end{equation*}
We consider \(\calP_{x}\) to be the law of such \(x'\) for a given \(x\), and \(\bfP = \{\calP_{x} : x \in \primalspace\}\)\footnote{The transition kernel for any \(x \in \partial\primalspace\) is analytically undefined since \(\mirrorfunc{}\) is of Legendre type, but this will not have any influence since the boundary \(\partial\primalspace\) is a Lebesgue null set as \(\primalspace\) is convex.}.
To compute the acceptance ratio (\cref{eq:accept-ratio-def}), we require the density of the \(\calP_{x}\) for each \(x\), and this can be obtained by the change of the variable formula, which states that given an differentiable invertible map \(T\) and probability measure \(\mu\) with density function \(d\mu\),
\begin{equation*}
    dT_{\#}\mu(x) = d\mu(T^{-1}(x)) \cdot |\det J T^{-1}(x)|~.
\end{equation*}
where \(J T^{-1}(x)\) is the Jacobian of \(T^{-1}\) evaluated at \(x\).
Since \(\mirrorfunc{}\) is assumed to be of Legendre type, \(\primtodual{}\) (and equivalently, \(\dualtoprim{}\)) is an invertible map.
Let \(\calN(x; \mu, \Sigma)\) denote the density of a multivariate normal distribution with mean \(\mu\) and covariance \(\Sigma\) at \(x \in \bbR^{d}\).
Then, for any \(z \in \interior{\primalspace}\),
\begin{align*}
    p_{x}(z) &= \calN((\dualtoprim{})^{-1}(z); \primtodual{x} - h \cdot \gradpotential{x}, 2h \cdot \hessmirror{x}) \cdot |\det J (\dualtoprim{})^{-1}(z)| \\
    &= \frac{\det \hessmirror{z}}{(4h\pi)^{\nicefrac{d}{2}} \cdot \sqrt{\det \hessmirror{x}}} \exp\left(-\frac{\|\primtodual{z} - \primtodual{x} + h\cdot \gradpotential{x}\|^{2}_{\hessmirror{x}^{-1}}}{4h}\right) ~.\numberthis\label{eq:proposal-density}
\end{align*}

This completes the definition of \MAMLA{} (see \cref{alg:mamla}).

\subsection{Mixing time analysis}

Here, we state our main theorem concerning the mixing time of \nameref{alg:mamla}, under assumptions made on both the potential \(\potential{}\), and the mirror function \(\mirrorfunc{}\).
Recall that \(\mirrorfunc{}\) is assumed to be of Legendre type.
The other key assumptions are
\begin{assumplist}
\item \label{assump:self-concord} \(\mirrorfunc{}\) is a self-concordant function with parameter \(\alpha\) (\cref{def:self-concord}),
\item \label{assump:rel-convex-smooth} \(\potential{}\) is \(\mu\)-relatively convex and \(\lambda\)-relatively smooth with respect to \(\mirrorfunc{}\) (\cref{def:rel-convex,def:rel-smooth}),
\item \label{assump:rel-lipschitz} \(\potential{}\) is \(\beta\)-relatively Lipschitz continuous with respect to \(\mirrorfunc{}\) (\cref{def:rel-lipschitz}), and
\item \label{assump:symm-barrier} \(\mirrorfunc{}\) is a symmetric barrier with parameter \(\nu\) (\cref{def:sym-barrier}).
\end{assumplist}

Define the constants \(\bar{\alpha} = \max\{1, \alpha\}\) and \(\gamma = \frac{\lambda}{2} + \alpha \cdot \beta\), which appear in the theorem below.

\begin{theorem}
\label{thm:mix-mamla}
Consider a distribution \(\Pi\) with density \(\targetdens(x) \propto e^{-\potential{x}}\) that is supported on a compact and convex set \(\primalspace \subset \bbR^{d}\), and mirror map \(\mirrorfunc{} : \primalspace \to \bbR \cup \{\infty\}\).
If \(\potential{}\) and \(\mirrorfunc{}\) satisfy assumptions \ref{assump:self-concord}-\ref{assump:rel-lipschitz}, then there exists a maximum step size \(h_{\max} > 0\) given by
\begin{equation}
\label{eq:h-max-form}
    h_{\max} = \min\left\{1, ~\frac{C^{(1)}}{\alpha^{2}d^{3}} ~,~ \frac{C^{(2)}}{d\gamma} ~,~ \frac{C^{(3)}}{\alpha^{\nicefrac{4}{3}}\beta^{\nicefrac{2}{3}}} ~,~ \frac{C^{(4)}}{\beta^{\nicefrac{2}{3}}\gamma^{\nicefrac{2}{3}}}~,~ \frac{C^{(5)}}{\beta^{2}}\right\}
\end{equation}
for universal constants \(C^{(1)}, \ldots, C^{(5)}\) such that for any \(0 < h \leq h_{\max}\), \(\delta \in (0, 1)\), and \(M\)-warm initial distribution \(\Pi_{0}\) with respect to \(\targetdist\), \nameref{alg:mamla} has the following mixing time guarantees.
\begin{description}
    \item [When \(\mu > 0\) (strongly convex),]
\begin{equation}
    \label{eq:mix-time-strong}
    \mixingtime{\delta; \bfT, \Pi_{0}} = \calO\left(\max\left\{1, \frac{\bar{\alpha}^{4}}{\mu \cdot h}\right\} \cdot \log\left(\frac{\sqrt{M}}{\delta}\right)\right)~.
\end{equation}

\item [When \(\mu = 0\) (weakly convex),] and additionally assuming that \(\mirrorfunc{}\) satisfies \ref{assump:symm-barrier},
\begin{equation}
    \label{eq:mix-time-weak}
    \mixingtime{\delta; \bfT, \Pi_{0}} = \calO\left(\max\left\{1,  \frac{\nu \cdot \bar{\alpha}^{2}}{h}\right\} \cdot \log\left(\frac{\sqrt{M}}{\delta}\right)\right)~.
\end{equation}
\end{description}
\end{theorem}

The proof of \cref{thm:mix-mamla} is given in \cref{sec:proofs:full-proof-thm}.

\subsubsection*{Mixing time for a Metropolis-adjusted Newton Langevin algorithm}
As noted previously, a specialisation of \ref{eq:MLD} / \ref{eq:MLA} is \ref{eq:NLD} / \textsf{NLA}, where \(\mirrorfunc{}\) is set to be \(\potential{}\).
This setting satisfies assumption \ref{assump:rel-convex-smooth} with constants \(\mu = \lambda = 1\), and when \(\potential{}\) is self-concordant with parameter \(\alpha\), and a barrier function with parameter \(\beta\), assumptions \ref{assump:self-concord} and \ref{assump:rel-lipschitz} are satisfied as well.
This leads to a corollary of \Cref{thm:mix-mamla} for the Metropolis-adjusted Newton Langevin algorithm, which is also unbiased with respect to the target, as stated in the following corollary.
In comparison, \cite{li2022mirror} provide a mixing time guarantee for \textsf{NLA} when \(\potential{}\) satisfies a modified self-concordance condition, which is a less desirable property for reasons elucidated in \cref{sec:intro:cond-num-indep}.

\begin{corollary}
\label{corr:mix-manla}
Consider a distribution \(\Pi\) with density \(\targetdens(x) \propto e^{-\potential{x}}\) that is supported on compact and convex set \(\primalspace \subset \bbR^{d}\).
When \(\potential{}\) is both a self-concordant function with parameter \(\alpha\) and a barrier function with parameter \(\beta\), there exists a maximum step size \(h_{\max} > 0\) of the form in \cref{eq:h-max-form} where \(\gamma = \frac{1}{2} + \alpha \cdot \beta\), such that for any \(0 < h \leq h_{\max}\), \(\delta \in (0, 1)\), and \(M\)-warm initial distribution \(\Pi_{0}\) with respect to \(\Pi\), \nameref{alg:mamla} with \(\mirrorfunc{} = \potential{}\) satisfies
\begin{equation*}
    \mixingtime{\delta; \bfT, \Pi_{0}} = \calO\left(\max\left\{1, \frac{\bar{\alpha}^{4}}{h}\right\} \cdot \log\left(\frac{\sqrt{M}}{\delta}\right)\right)~.
\end{equation*}
\end{corollary}

\subsubsection{A discussion of the result in \cref{thm:mix-mamla}}

We begin by discussing the assumptions \ref{assump:rel-convex-smooth}-\ref{assump:symm-barrier}.
\ref{assump:rel-convex-smooth} states that \(\potential{}\) is a relatively convex and smooth function with respect to \(\mirrorfunc{}\).
Prior works that analyse \ref{eq:MLA} and other discretisations of \ref{eq:MLD} \citep{zhang2020wasserstein,ahn2021efficient,jiang2021mirror,li2022mirror} consider this assumption.
As commented previously in \cref{sec:func-classes}, when \(\primalspace = \bbR^{d}\) and \(\mirrorfunc{x} = \frac{\|x\|^{2}}{2}\), this is equivalent to assuming strong convexity and smoothness of \(\potential{}\) in the usual sense which is instrumental in the analysis of \ref{eq:ULA}, \textsf{MALA}, and other unconstrained sampling methods.
We note that \cite{jiang2021mirror} assumes that \(\targetdist\) satisfies a  mirror log-Sobolev inequality in lieu of relative convexity of \(\potential{}\) with respect to \(\mirrorfunc{}\) to analyse \ref{eq:MLA}, \MLAFD{}, and another discretisation proposed in their work that we refer to as \MLABD{}.
While the usual strong convexity of \(\potential{}\) implies the log-Sobolev inequality \citep{bakry1984diffusions}, it is not known if relative strong convexity of \(\potential{}\) with respect to \(\mirrorfunc{}\) yields a mirror log-Sobolev inequality.
This makes it hard to assess whether this substitution in \cite{jiang2021mirror} is a weaker assumption than relative (strong) convexity in \ref{assump:rel-convex-smooth}.
Additionally, the mixing time guarantees for \ref{eq:MLA} in \cite{li2022mirror} (who work with a subset of assumptions in \cite{zhang2020wasserstein}) is only meaningful when \(\mu > 0\) in \ref{assump:rel-convex-smooth}, and a guarantee for \ref{eq:MLA} in the case where \(\potential{}\) is (weakly) convex (\(\mu = 0\)) is unknown still.
On the other hand, \cite{ahn2021efficient} give an analysis of \(\MLAFD{}\) for both cases i.e., \(\mu > 0\) and \(\mu = 0\).
Next, \ref{assump:rel-lipschitz} states that the gradients of \(\potential{}\) are bounded in the local norm \(\|.\|_{\hessmirror{.}^{-1}}\).
This is used in the analysis of \MLAFD{} and \MLABD{} in \cite{ahn2021efficient} and \cite{jiang2021mirror} respectively, but is not used in any existing analysis of \ref{eq:MLA}.
Finally, \ref{assump:symm-barrier} is a geometric property of \(\mirrorfunc{}\), and is useful to obtain guarantees in the case where \(\mu = 0\) in \ref{assump:rel-convex-smooth}.
In this case, we rely on isoperimetric inequalities for sampling from log-concave densities over convex bodies \citep{vempala2005geometric}, which we get through \ref{assump:symm-barrier}.
This assumption has been employed in prior work to analyse the Dikin walk \citep{laddha2020strong,kook2023efficiently}, constrained \textsf{RHMC} \citep{kook2023condition}, hybrid \textsf{RHMC} \citep{gatmiry2023sampling}.

\begin{table}[t]
\centering
\renewcommand{\arraystretch}{1.75}
\begin{tabular}{ccccc}
Setting of \(\mu\) in \ref{assump:rel-convex-smooth} & \makecell{\nameref{alg:mamla} \\ {\scriptsize (this work)}} & \ref{eq:MLA}\textsuperscript{*} & \MLAFD{} & \MLABD{} \\
\hline
\(\mu > 0\) & \(\widetilde{\calO}\left(\frac{1}{\mu} \max\left\{d^{3}, d\lambda\right\} \log\left(\frac{1}{\delta}\right)\right)\) & \(\widetilde{\calO}\left(\frac{d \lambda^{2}}{\mu^{3} \delta^{2}}\right)\) & \(\widetilde{\calO}\Big(\frac{d \lambda}{\mu\delta^{2}}\Big)\) & \(\widetilde{\calO}\left(\frac{\sqrt{\lambda^{2} + d^{3}}}{\mu^{\nicefrac{3}{2}}\delta}\right)\) \\
[4pt] \arrayrulecolor{lightgray}\hline 
\(\mu = 0\) & \(\widetilde{\calO}\left(\nu \max\Big\{d^{3}, d\lambda\Big\} \log\left(\frac{1}{\delta}\right)\right)\) & N/A & \(\widetilde{\calO}\left(\frac{d^{2} \lambda}{\delta^{4}}\right)\) & N/A \\
\end{tabular}
\caption{Comparison of mixing time guarantees for algorithms based on discretisations of \ref{eq:MLD}.
We use the \(\widetilde{\calO}\) notation to only showcase the dependence on \(\lambda, \mu, d\), and \(\delta\), and omit the dependence on other parameters.
For \MLABD{} in \cite{jiang2021mirror}, \(\mu\) is the constant in the mirror log-Sobolev condition, as relative strong convexity is not considered, and \(\lambda\) is the constant in the alternative relative smoothness condition assumed on \(\potential{}\).
}
\label{tab:comparison}
\end{table}

Under these assumptions, we are interested in how the mixing time guarantees scale with the error tolerance \(\delta\), and the dimension \(d\) of \(\primalspace\).
\Cref{tab:comparison} summarises the comparison between \nameref{alg:mamla} and other algorithms based on discretisations of \ref{eq:MLD} that have been discussed above.
The analyses of \ref{eq:MLA}, \MLAFD{}, and \MLABD{} in the aforementioned works consider varying definitions of mixing time to the \(\delta\)-mixing time in TV distance we define in \Cref{sec:prelims-markov-chains}, and additional assumptions which we highlight as follows.
The most recent mixing time guarantees for \ref{eq:MLA} in \cite{li2022mirror} is given in terms of the mirrored 2-Wasserstein distance, as was previously done in \cite{zhang2020wasserstein}; this is indicated by an asterisk in \Cref{tab:comparison}.
Both of these works also assume that \(\mirrorfunc{}\) satisfies a \emph{modified self-concordance} condition instead of self-concordance as we do in \ref{assump:self-concord}.
Furthermore, the relation between the mirrored 2-Wasserstein distance and a more canonical functional like the KL divergence (\(\KLdist{}\)) / TV distance (\(\TVdist{}\)) cannot be easily established without assuming that \(\mirrorfunc{}\) satisfies additional properties.
For \MLAFD{}, \cite{ahn2021efficient} establish bounds on \(K\) such that \(\KLdist(\overline{\Pi}_{K}, \Pi) \leq \delta\).
Here, \(\overline{\Pi}_{K}\) is a uniform mixture distribution composed of the sequence of iterates \(\{\Pi_{k}\}_{k=1}^{K}\), where \(\Pi_{k}\) is the distribution at iteration \(k\).
For \MLABD{}, \cite{jiang2021mirror} gives upper bounds on \(K\) such that \(\KLdist(\Pi_{K}, \Pi) \leq \delta\).
To do so, they place a different relative smoothness assumption over \(\potential{}\), and additionally assume that \(\mirrorfunc{}\) is strongly convex.
For \MLAFD{} and \MLABD{}, we infer mixing time guarantees in TV distance from guarantees in KL divergence using Pinsker's inequality which states that \(\TVdist(\mu, \nu) \leq \sqrt{\frac{1}{2}\KLdist(\mu, \nu)}\).
We call \nameref{alg:mamla} \emph{fast} due to the dependence on \(\delta\) in the mixing time guarantees (\(\log(\nicefrac{1}{\delta})\)), which is exponentially better than the dependence of \(\delta\) in the mixing time guarantees for \ref{eq:MLA}, \MLAFD{}, and \MLABD{} (\(\mathrm{poly}(\nicefrac{1}{\delta})\)).
This echoes the improvement observed in the mixing time guarantees for \MALA{} relative to \ref{eq:ULA} \citep{dwivedi2018log,chewi2021optimal}.
In contrast, the dependence on \(d\) is better in the mixing time guarantees for the unadjusted Mirror Langevin discretisations; fixing \(\beta\), \(\lambda\), \(\mu\), the mixing time bound for \nameref{alg:mamla} scales as \(d^{3}\) compared to \(d^{\gamma}\) with \(\gamma \in \{1, 1.5, 2\}\) in the mixing time bounds for the other methods.

\subsection{Applications of \nameref{alg:mamla} with provable guarantees}
\label{sec:algo:applications-thm}

In this subsection, we discuss some applications of \nameref{alg:mamla} for which we can infer mixing time guarantees from \cref{thm:mix-mamla}.
These are (1) uniform sampling from polytopes and regions defined by the intersection of ellipsoids, and (2) sampling from Dirichlet distributions.
The proofs of the statements given in this subsection are stated in \cref{sec:proofs:corollaries}.
We use \(C\) to denote a universal positive constant in the corollaries, which can change between corollaries.

\subsubsection{Uniform sampling over polytopes and intersection of ellipsoids}

For uniform sampling, the target density \(\targetdens\) is a constant function, and consequently \(\gradpotential{x} = 0\) for any \(x \in \interior{\primalspace}\).
In this setting, \nameref{alg:mamla} can be viewed as a Gaussian \textsf{DikinWalk} in the dual space \((\bbR^{d}, \hessmirrorinv{})\).
We are interested in approximate uniform sampling from the following sets.
\begin{itemize}[leftmargin=*]
\item \(\mathsf{Polytope}(A, b)\): a bounded polytope with non-zero volume defined by \(\{x \in \bbR^{d} : Ax \leq b\}\) for matrix \(A \in \bbR^{m \times d}\) and vector \(b \in \bbR^{m}\), and
\item \(\mathsf{Ellipsoids}(\{(c_{i}, M_{i})\}_{i=1}^{m})\): a non-empty region defined by the intersection of ellipsoids \(\{x \in \bbR^{d} : \|x - c_{i}\|_{M_{i}}^{2} \leq 1~\forall ~i \in [m]\}\) for a sequence of \(d \times d\) symmetric positive definite matrices \(\{M_{i}\}_{i=1}^{m}\) and centres \(\{c_{i}\}_{i=1}^{m}\).
The radius of \(1\) does not affect the generality of this region.
\end{itemize}
The following corollary establishes mixing time guarantees for uniform sampling over these sets.

\begin{corollary}
\label{corr:mixing-time-polytope-ellipsoids}
Let \(\Pi_{0}\) be a \(M\)-warm distribution with respect to the uniform distribution over either \(\primalspace = \mathsf{Polytope}(A, b)\), or \(\primalspace = \mathsf{Ellipsoids}(\{(c_{i}, M_{i})\}_{i=1}^{m})\), and let \(\mirrorfunc{}\) be the log-barrier of \(\primalspace\).
Then, for any \(\delta \in (0, 1)\), the mixing time of \nameref{alg:mamla} is
\begin{equation*}
    C \cdot m \cdot d^{3} \cdot \log\left(\frac{\sqrt{M}}{\delta}\right)~.
\end{equation*}
\end{corollary}

\textbf{Remark.}~~
For polytopes, \citet{laddha2020strong} show that \textsf{DikinWalk} satisfies a mixing time of \(m \cdot d\) owing to a \emph{strong self-concordance} condition that holds in this setting.
They also propose \textsf{WeightedDikinWalk}, which has a mixing time that scales as \(d^{2}\) (independent of the number of constraints).
In a similar vein, \cite{gatmiry2023sampling} propose a modification to \textsf{RHMC}, and prove a mixing time guarantee that scales as \(m^{\nicefrac{1}{3}} \cdot d^{\nicefrac{4}{3}}\) for their method.
\cite{kook2023condition} study a constrained \textsf{RHMC} algorithm applicable to this setting i.e., uniform sampling from both polytopes and intersection of ellipsoids, and show a mixing time guarantee for this algorithm that scales as \(m \cdot d^{3}\).
Our analysis echoes the mixing time guarantee of the latter method and scales as \(m \cdot d^{3}\).

\subsubsection{Sampling from Dirichlet distributions}

The Dirichlet distribution is the multi-dimensional generalisation of the Beta distribution.
A sample from the Dirichlet distribution \(x' \in \bbR^{d + 1}\) satisfies \(x'_{i} \geq 0\) for all \(i \in [d + 1]\), and \(\bm{1}^{\top}x' = 1\), and thus an element of \(\Delta_{d}\).
Equivalently, we can also express this sample with the first \(d\) elements \(x \in \bbR^{d}_{+}\) satisfying an inequality constraint \(\bm{1}^{\top}x \leq 1\), and write \(x'_{d + 1} = 1 - \bm{1}^{\top}x\).
We work with the latter definition, and \(\primalspace = \{x \in \bbR^{d}_{+} : \bm{1}^{\top}x \leq 1\}\).
The Dirichlet distribution is parameterised by a vector of positive reals\footnote{More generally, \(a_{i} > -1\), but we focus on when \(a_{i} > 0\). The case where \(a_{i} \in (-1, 0)\) results in antimodes.} \(\bm{a} \in \bbR^{d + 1}_{+}\) called the concentration parameter.
The density is
\begin{equation}
\label{eq:potential-dirichlet}
    \targetdens(x) \propto \exp(-\potential{x})~; \qquad \potential{x} = -\sum_{i=1}^{d}a_{i} \cdot \log x_{i} - a_{d + 1} \cdot \log\left(1 - \sum_{i = 1}^{d}x_{i}\right)~.
\end{equation}
We use \(\bm{a}_{\min}\) to denote the minimum of \(\bm{a}\).
Since the sample space is a special polytope, the log-barrier of \(\primalspace\) is a natural consideration for the mirror function \(\mirrorfunc{}\), and with this choice of \(\mirrorfunc{}\), we generate approximate samples from a Dirichlet distribution using \nameref{alg:mamla}.
The following corollary states a mixing time upper bound for this task.

\begin{corollary}
\label{corr:dirichlet-sampling}
Let \(\Pi_{0}\) be a \(M\)-warm distribution with respect to a Dirichlet distribution parameterised by \(\bm{a} \in \bbR^{d + 1}_{+}\).
Let \(\mirrorfunc{}\) be the log-barrier of \(\calK\), and \(\potential{}\) be as defined in \cref{eq:potential-dirichlet}.
If \(\|\bm{a}\| \geq 1\), then for any \(\delta \in (0, 1)\), the mixing time of \nameref{alg:mamla} is
\begin{equation*}
    C \cdot \frac{\max\{d^{3},~ d \cdot \|\bm{a}\|,~ \|\bm{a}\|^{2}\}}{\bm{a}_{\min}} \cdot \log\left(\frac{\sqrt{M}}{\delta}\right)~.    
\end{equation*}
\end{corollary}

\textbf{Remark.}~~
In comparison, a corollary of the guarantee for \ref{eq:MLA} as shown in \cite{li2022mirror} gives a mixing time guarantee that scales as \(\frac{d \bm{a}_{\max}^{2}}{\bm{a}_{\min}^{3}} \cdot \frac{1}{\delta^{2}}\) for this task.
It is worth noting that this holds only when the \emph{modified self-concordance} parameter used in its analysis is at most \(\bm{a}_{\min}\), and that this is in the mirrored 2-Wasserstein distance.
\MLAFD{} \citep{ahn2021efficient} on the other hand satisfies a mixing time guarantee that scales as \(\frac{d \|\bm{a}\|}{\bm{a}_{\min}} \cdot \frac{1}{\delta^{2}}\).
We cannot establish mixing time guarantees for \MLABD{} in this case due to the use of a mirror log-Sobolev inequality instead of relative strong convexity, and this former condition is hard to verify in general.

\subsection{Implementation details}
For practical purposes, a discussion about the implementation of \nameref{alg:mamla} is in order.
We give the following proposition which quantifies the complexity of each iteration of the algorithm.
\begin{proposition}
    Assume the invariant that at each iteration \(k\) of \nameref{alg:mamla} that we have the primal iterate \(x_{k}\) and the dual iterate \(y_{k} = \primtodual{x_{k}}\).
    If the cost of computing \(\hessmirror{}\), \(\gradpotential{}\), and \(\dualtoprim{}\) is \(\calO(d^{3})\), then each iteration can be implemented with cost \(\calO(d^{3})\).
\end{proposition}

We next describe a procedure that satisfies this complexity claim.
For convenience, we use the shorthand notation \(\costhessmirror\), \(\costgradpotential\), \(\costdualtoprim\) to denote the computational costs associated with computing the Hessian, the gradient of the potential and the inverse mirror map respectively at any point \(x \in \interior{\primalspace}\).

The key steps in \nameref{alg:mamla} are \cref{alg:proposal-step,alg:acceptance-ratio-step}.
We first examine the cost of \cref{alg:proposal-step}.
Note that \cref{alg:proposal-step} can be equivalently written as
\begin{equation*}
    \underbrace{\tilde{\xi} \sim \calN(0, \hessmirror{x_{k}})}_{\text{Step 1}}~; \quad \underbrace{\tilde{z} = y_{k} - h \cdot \gradpotential{x_{k}} + \sqrt{2h} \cdot \tilde{\xi}~}_{\text{Step 2}}; \quad \underbrace{z = \dualtoprim{\tilde{z}}}_{\text{Step 3}}~.
\end{equation*}
The computational cost of \cref{alg:proposal-step} is the sum of costs of these steps above.
\begin{description}
\item [Step 1:] this involves computing \(\hessmirror{x_{k}}\).
The most efficient way to sample from \(\calN(0, \hessmirror{x_{k}})\) when \(\hessmirror{x_{k}}\) has no specific structure is by computing the Cholesky factor \(L_{x_{k}}\)\footnote{Assumed to be lower triangular.}, and using \(\xi \sim \calN(0, I)\) to obtain \(\tilde{\xi} = L_{x_{k}}\xi\).
This circumvents computing the matrix square root, whose implementation can be imprecise and slow in comparison.
The cost of this step is
\begin{equation*}
    \costhessmirror + \costcholesky{d} + \calO(d^{2})
\end{equation*}
where the additional \(\calO(d^{2})\) is the cost for matrix-vector product cost for computing \(L_{x_{k}}\xi\) and absorbs the \(\calO(d)\) cost for sampling from \(\calN(0, I)\).
\item [Step 2:] here we require computing \(\gradpotential{x_{k}}\).
The remainder of this step is scalar-vector multiplication and vector addition, which gives the cost as
\begin{equation*}
    \costgradpotential + \calO(d)~.
\end{equation*}
\item [Step 3:] this involves computing the inverse mirror map only, costing \(\costdualtoprim\).
\end{description}
The net cost of \cref{alg:proposal-step} is hence \(\costhessmirror + \costcholesky{d} + \costgradpotential + \costdualtoprim + \calO(d^{2})\).
Next, we examine the cost of \cref{alg:acceptance-ratio-step}.
The key operation is computing the ratio inside the minimum.
For precision purposes, it is preferrable to work in the log-scale, and the explicit form of the log of this ratio is
\begin{align*}
    \log &~ \frac{\pi(z) p_{z}(x_{k})}{\pi(x_{k}) p_{x_{k}}(z)} = \underbrace{\potential{x_{k}} - \potential{z}}_{\text{Step 4}} + \underbrace{\frac{3}{2}\left\{\log \det \hessmirror{x_{k}} - \log \det \hessmirror{z}\right\}}_{\text{Step 5}} \numberthis \\
    &+ \underbrace{\frac{1}{4h}\left(\|\primtodual{z} - \primtodual{x_{k}} + h \cdot \gradpotential{x_{k}}\|^{2}_{\hessmirror{x_{k}}^{-1}} - \|\primtodual{x_{k}} - \primtodual{z} + h \cdot \gradpotential{z}\|^{2}_{\hessmirror{z}^{-1}}\right)}_{\text{Step 6}}~.
\end{align*}

Given \(L_{z}\) such that \(\hessmirror{z} = L_{z}L_{z}^{\top}\), step 6 can be simplfied as follows.
\begin{align*}
    \|\primtodual{z} - (\primtodual{x_{k}} - h \cdot \gradpotential{x_{k}})\|_{\hessmirror{x_{k}}^{-1}}^{2} &= 2h \cdot \|\xi\|^{2} ~,\\
    \|(\primtodual{x_{k}} - \primtodual{z}) + h \cdot \gradpotential{z}\|_{\hessmirror{z}^{-1}}^{2} &= \|h \cdot (\gradpotential{x_{k}} + \gradpotential{z}) - \sqrt{2h} \cdot L_{x_{k}}\xi\|_{\hessmirror{z}^{-1}}^{2} \\
    &= h^{2} \cdot \|L_{z}^{-1}(\gradpotential{x_{k}} + \gradpotential{z})\|^{2} + 2h \cdot \|L_{z}^{-1}L_{x_{k}}\xi\|^{2} \\
    &\qquad - (2h)^{\nicefrac{3}{2}}\langle L_{z}^{-1}(\gradpotential{x_{k}} + \gradpotential{z}), L_{z}^{-1}L_{x_{k}}\xi\rangle~.
\end{align*}
Above, we have demonstrated that computing the Cholesky factor of \(\hessmirror{z}\) is also beneficial, since the log determinant of a lower triangular matrix is the sum of log of the values on the diagonal, and triangular solves can be performed extremely efficiently (due to the echelon form).
The cost of computing the log acceptance ratio is the sum of costs of these steps above.
\begin{description}
    \item [Step 4:] this involve 2 calls to the potential, costing \(2 \cdot \costpotential\)
    \item [Step 5:] here, we have to compute \(\hessmirror{z}\) and its Cholesky factor.
    As noted previously, given the Cholesky factors, the log determinant can be computed efficiently since \(\log \det \hessmirror{x_{k}} = \log (\det L_{x_{k}})^{2} = 2 \log \det L_{x_{k}}\) (and respectively for \(\log \det \hessmirror{z}\)).
    The cost of this step is
    \begin{equation*}
        \costhessmirror + \costcholesky{d} + \calO(d)~.
    \end{equation*}
    \item [Step 6:] first, we have to compute \(\gradpotential{z}\).
    From the simplifications above, we have to perform two triangular solves to obtain \(L_{z}^{-1}L_{x_{k}}\xi\) and \(L_{z}^{-1}(\gradpotential{x_{k}} + \gradpotential{z})\).
    Recall that \(L_{x_{k}}\xi = \tilde{\xi}\) from Step 1.
    After performing these, we have to compute 3 inner products, which cost \(\calO(d)\).
    The net cost of this step is
    \begin{equation*}
        \costgradpotential + 2 \cdot \costtrisolve{d} + \calO(d)~.
    \end{equation*}
\end{description}
This gives the total cost of an iteration assuming the invariant as
\begin{equation*}
    2 \cdot (\costhessmirror + \costgradpotential + \costpotential + \costcholesky{d} + \costtrisolve{d}) + \costdualtoprim + \calO(d^{2})~. 
\end{equation*}

Without any assumptions about the structure of the matrices, \(\costtrisolve{d}\) is \(\calO(d^{2})\) and \(\costcholesky{d}\) is \(\calO(d^{3})\) \citep{trefethen1997numerical}.
When \(\costgradpotential\), \(\costhessmirror\), and \(\costdualtoprim\) is \(\calO(d^{3})\), the total cost is \(\calO(d^{3})\), thus asserting the proposition.

In practice, the cost of computing \(\hessmirror{}\), \(\gradpotential{}\), and \(\dualtoprim{}\) depends on the forms of \(\potential{}\) and \(\phi\), and is also further influenced by parallelisation and vectorisation functionality in many popular scientific computing packages which are used to implement such algorithms in practice, and these optimisations  can lower the net computational cost.
Additionally, due to such hardware optimisations, multiple samples can be simultaneously obtained in a ``batched'' manner.
When \(\phi\) is the log-barrier of a polytope, \(\costhessmirror = \calO(md^{2})\) where \(m\) is the number of linear constraints.
However, computing the inverse mirror map can be more complex.
For general polytopes, this is best computed by solving the dual problem \(\dualtoprim{y} = \argmin_{x \in \primalspace} \phi(x) - x^{\top}y\) using state-of-the-art convex solvers.
There are cases where this can be avoided using closed form expression, or reductions to binary search.
In the case where \(\primalspace\) is a box given by \([-b_{1}, b_{1}] \times \ldots \times [-b_{d}, b_{d}]\) for a positive sequence \(\{b_{i}\}_{i=1}^{d}\), or when \(\primalspace\) is an ellipsoid, the \(\costdualtoprim = \calO(d)\) using closed form expressions.
The case of the box is also special because the Cholesky decompositions and triangular solves are not necessary, since \(\hessmirror{x}\) is diagonal for all \(x\), and the cost of an iteration is reduced to \(\calO(d)\) in total.
When \(\primalspace\) is a simplex, it is not possible to obtain a closed form expression for \(\dualtoprim{}\), but this can be computed to machine precision by reducing it to binary search over an interval, and the net cost in this case is \(\calO(d)\) as well, and we give an algorithm to compute this in \cref{app:sec:simplex-inverse-mirror-map}.

\section{Numerical experiments}
\label{sec:expts}
In this section, we showcase numerical experiments that we conducted to assess \nameref{alg:mamla} on a selection of problems which have previously been discussed in \cref{sec:algo:applications-thm}.\footnote{Code for these experiments can be found at \url{https://github.com/vishwakftw/metropolis-adjusted-MLA}.}
We are interested in three settings: (1) drawing approximately uniform samples from certain polytopes and an ellipsoid (\cref{sec:uniform-sampling-empirics}), (2) drawing approximate samples from a Dirichlet distribution (\cref{sec:dirichlet-sampling-empirics}), and (3) analysing the effect of step size on the acceptance rate of the Metropolis-Hastings filter in \nameref{alg:mamla} (\cref{sec:accept-rate-analysis-empirics}).
Since we only have samples at each iteration of such MCMC algorithms, estimating the TV distance is not possible.
To circumvent this, we use other suitable approximate metrics that are amenable to empirical measures, based on prior works.
We describe these metrics as we proceed because they vary from problem to problem.

\subsection{Uniform sampling}
\label{sec:uniform-sampling-empirics}

In this problem, we are interested in drawing approximately uniform samples from three kinds of domains.
The first kind of domain is a \(d\)-dimensional box defined as \([-b_{1}, b_{1}] \times \ldots \times [-b_{d}, b_{d}]\) for a positive sequence \(\{b_{i}\}_{i=1}^{d}\).
This can equivalently be written as a polytope with \(2d\) constraints, and we will use the notation \(\mathsf{Box}(\{b_{i}\}_{i=1}^{d})\) for convenience.
The second kind of domain is a \(d\)-dimensional ellipsoid centred at \(\bm{0}\), which we denote as \(\mathsf{Ellipsoid}(M) = \{x : \|x\|_{M} \leq 1\}\).
The third kind of domain is a simplex in \((d + 1)\) dimensions embedded in \(\bbR^{d}\), which we denote as \(\mathsf{S}_{d} = \{x \in \bbR^{d} : x_{i} \geq 0 ~\forall i \in [d] ~, \bm{1}^{\top}x \leq 1\}\).

\begin{figure}[t]
    \centering
    \includegraphics[width=0.87\linewidth]{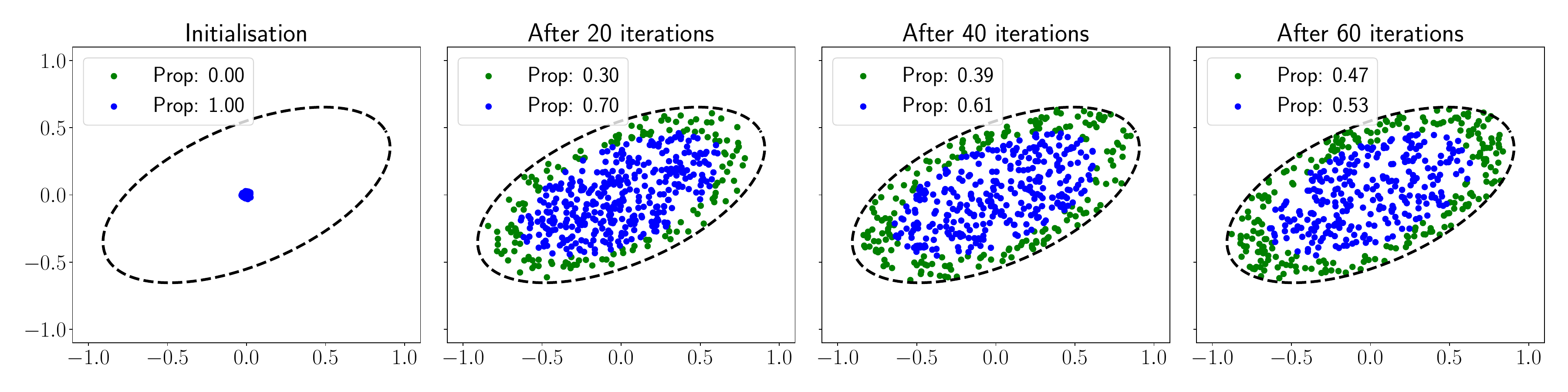}
    \caption{Progression of \nameref{alg:mamla} on an \(\mathsf{Ellipsoid}(M)\) with \(d = 2\) and \(\lambda_{1}(M) = 1\) and \(\lambda_{2}(M) = 4\).
    Points coloured in green are contained in \(\mathsf{Ellipsoid}^{\nicefrac{1}{2}}(M)\).
    ``\textsf{Prop}'' is the proportion of points in the regions.
    Note that \(\widehat{\tau}_{\mathrm{mix}}\) is at most \(60\) in this case.
    }
    \label{fig:progress-ellipsoid}
\end{figure}

We define the (empirical) mixing time for uniform sampling over these kinds of sets next.
To do this, we first identify the following sets.
\begin{align*}
    \mathsf{Box}(\{b_{i}\}_{i=1}^{d})^{\nicefrac{1}{2}} &= \{x \in \bbR^{d} : b_{i} \cdot 2^{-\nicefrac{1}{d}} < |x_{i}| \leq b_{i} ~\forall ~i \in [d]\} ~,\\
    \mathsf{Ellipsoid}(M)^{\nicefrac{1}{2}} &= \{x \in \bbR^{d} : 2^{-\nicefrac{1}{d}} < \|x\|_{M}^{2} \leq 1\}~, \\
    \mathsf{S}_{d}^{\nicefrac{1}{2}} &= \{x \in \bbR^{d}_{+} : 2^{-\nicefrac{1}{d}} < \bm{1}^{\top}x \leq 1\}~.
\end{align*}
The key properties of these sets is that these are subsets of \(\mathsf{Box}(\{b_{i}\}_{i=1}^{d})\), \(\mathsf{Ellipsoid}(M)\) and \(\mathsf{S}_{d}\) respectively, and their volumes are exactly half of the larger sets they are contained in.
Let \(\widehat{\bbT}^{k}\) be the empirical measure of samples after \(k\) iterations of \nameref{alg:mamla} obtained from multiple chains (\(N\)).
Then, the (empirical) mixing time for uniform sampling over domain \(\calK\) \citep{chen2018fast} is
\begin{equation*}
    \widehat{\tau}_{\mathrm{mix}}(\delta; \primalspace) = \inf\left\{k \geq 0 : \widehat{\bbT}^{k}(\calK^{\nicefrac{1}{2}}) \geq \frac{1}{2} - \delta\right\}~, \quad \calK \in \{\mathsf{Box}(\{b_{i}\}_{i=1}^{d})~, \mathsf{Ellipsoid}(M)~, \mathsf{S}_{d}\}.
\end{equation*}

We use \(\delta = \nicefrac{1}{20}\), and use the shorthand \(\widehat{\tau}_{\mathrm{mix}}\) when the domain is clear from context.
This metric seeks to quantify how quickly they spread out uniformly to a region closer to the boundary, especially when the initial points are close together.
It must be noted that for methods without a filter (like \ref{eq:MLA}), a large step size would cause points to move to the boundary rapidly, and give the illusion of quick mixing.
In \cref{fig:progress-ellipsoid}, we show the progression of \nameref{alg:mamla} to sample from a 2-D \(\mathsf{Ellipsoid}\), and highlight points that lie in the associated \(\mathsf{Ellipsoid}^{\nicefrac{1}{2}}\).

\begin{table}[t]
\centering
\renewcommand{\arraystretch}{1.3}
\begin{tabular}{rccc}
    & & \(h \propto d^{-1}\) & \(h \propto d^{-\nicefrac{3}{2}}\) \\
    \hline
    \multirow{2}{*}{\(\primalspace = \mathsf{Box}\)} & \(i = 1\) & 1.213~\textsubscript{(0.016)} & 1.589~\textsubscript{(0.013)}\\
    & \(i = 2\) & 1.213~\textsubscript{(0.014)} & 1.611~\textsubscript{(0.013)} \\
    \arrayrulecolor{lightgray}\hline
    \multirow{2}{*}{\(\primalspace = \mathsf{Ellipsoid}\)} & \(i = 1\) & 1.133~\textsubscript{(0.016)} & 1.628~\textsubscript{(0.026)} \\
    & \(i = 2\) & 1.167~\textsubscript{(0.016)} & 1.599~\textsubscript{(0.026)} \\
\end{tabular}
\caption{The slopes (and standard errors) of the best fit lines between \(\log(\widehat{\tau}_{\mathrm{mix}})\) and \(\log(d)\) for the sequence of domains \(\calD_{i}^{\primalspace}\) for two domain types \(\primalspace \in \{\mathsf{Box}, \mathsf{Ellipsoid}\}\) and two step size choices.
}
\label{tab:mixing-time-cond-num-indep}
\end{table}

\subsubsection{Mixing time versus dimension}

Here, our objective is two-fold: to empirically validate that the mixing time of \nameref{alg:mamla} (1) does not depend on the conditioning of the domain, and (2) scales as \(h^{-1}\) as shown by \cref{thm:mix-mamla}.
For the first, we consider two sequences of boxes and ellipsoids.
In the first sequence \(\calD_{1}^{\primalspace}\), the conditioning of \(\primalspace\) (which is either a \(\mathsf{Box}\) or an \(\mathsf{Ellipsoid}\)) scales as \(d^{2}\), and in the second sequence \(\calD_{2}^{\primalspace}\), the conditioning of the \(\primalspace\) scales as \(e^{d}\), and both of these sequences are indexed by the dimension \(d\).
More precisely, we let \(\kappa_{1} = \frac{d^{2}}{4}, \kappa_{2} = e^{\nicefrac{d}{4}}\), and define for \(i \in \{1, 2\}\)
\begin{gather*}
    \calD_{i}^{\mathsf{Box}}(d) = \mathsf{Box}(\{b_{i}\}_{i=1}^{d})~, \quad b_{1} = \cdots = b_{d - 1} = 1,~ b_{d} = \frac{1}{\kappa_{i}} ~,\\
    \calD_{i}^{\mathsf{Ellipsoid}}(d) = \mathsf{Ellipsoid}(M)~, \quad \lambda_{j}(M) = 1 + \frac{(\kappa_{i} - 1)(j - 1)}{d - 1},~ j \in [d].
\end{gather*}
The purpose of \(\kappa_{1}\) and \(\kappa_{2}\) as stated above is to empirically identify a dependency of the condition number of the domain that is either polynomial or poly-logarithmic in nature.
Such a dependence, if any, would be reflected in the slope of a log-log plot between \(\widehat{\tau}_{\mathrm{mix}}\) and \(d\).
We obtain an approximate sample over each of these domains by running \nameref{alg:mamla} with \(\potential{} = 0\), \(\mirrorfunc{}\) given by the log-barrier of the domain for \(2000\) iterations, and generate samples from \(N = 2000\) independent chains with two choices of step size \(h\): \(h = \frac{C}{d}\) and \(h = \frac{C}{d^{\nicefrac{3}{2}}}\), where \(C\) is a small constant depending on the kind of domain\footnote{\(C = 0.25\) for the box, \(C = 0.05\) for the ellipsoid, and \(C = 0.1\) for the simplex.}.
To capture variation across such generating processes, we perform \(10\) independent runs.

\begin{figure}[t]
    \centering
    \begin{subfigure}{0.32\linewidth}
        \centering
        \includegraphics[width=0.95\linewidth]{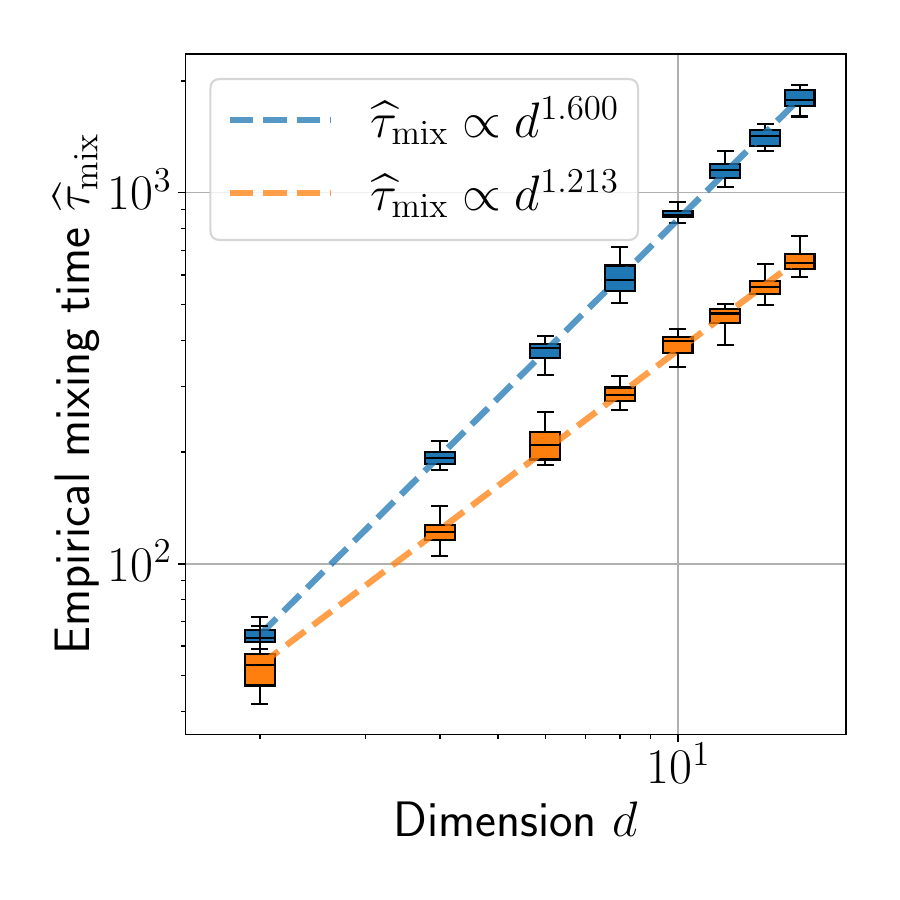}
        \caption{\(\calK = \mathsf{Box}\)}        
    \end{subfigure}
    \hfill
    \begin{subfigure}{0.32\linewidth}
        \centering
        \includegraphics[width=0.95\linewidth]{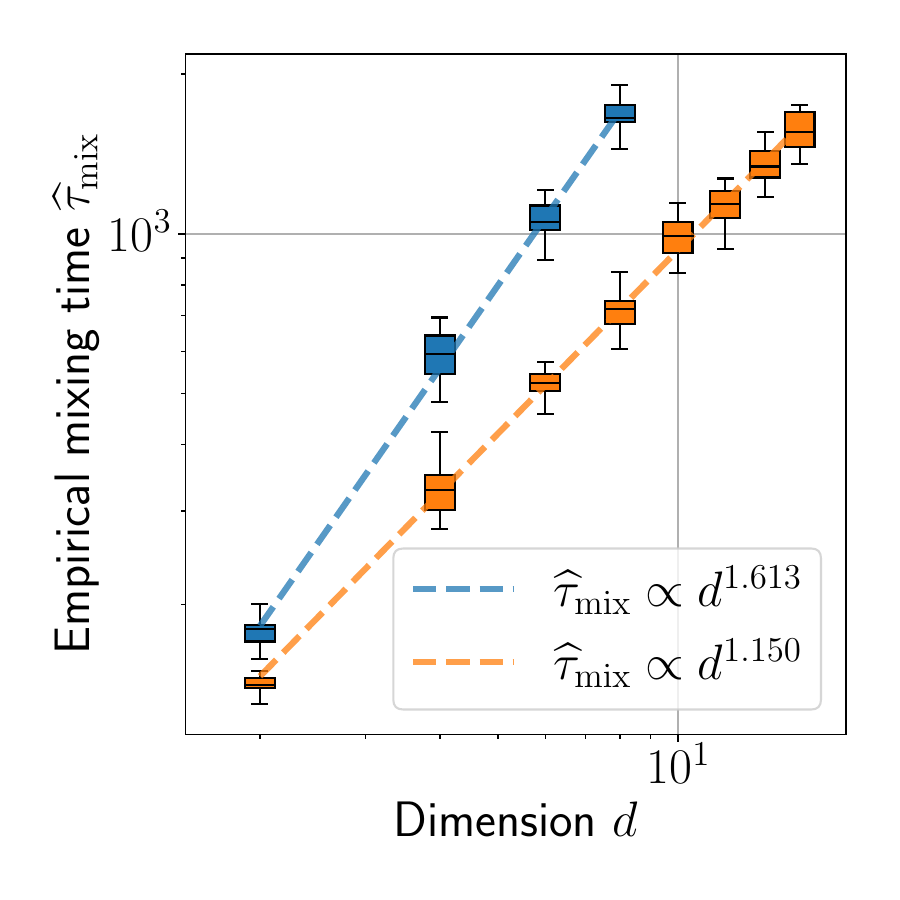}
        \caption{\(\calK = \mathsf{Ellipsoid}\)}        
    \end{subfigure}
    \hfill
    \begin{subfigure}{0.32\linewidth}
        \centering
        \includegraphics[width=0.95\linewidth]{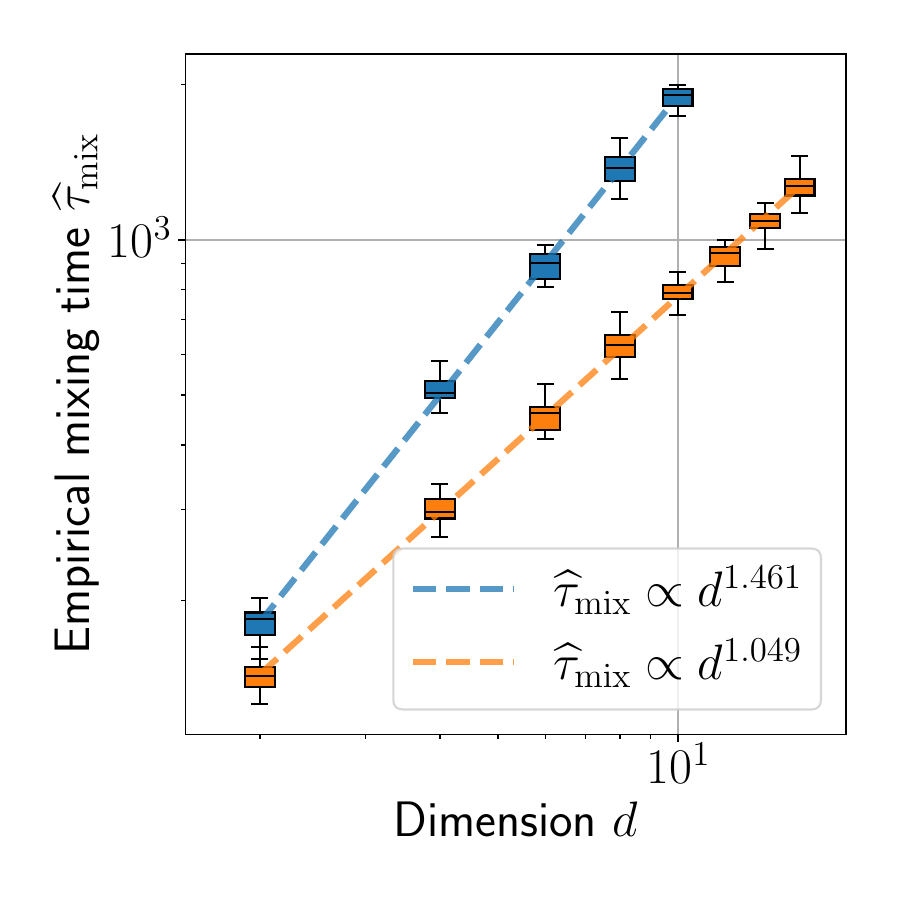}
        \caption{\(\calK = \mathsf{S}_{d}\)}        
    \end{subfigure}
    \caption{(Empirical) Mixing time versus dimension.
    The orange line corresponds to \(h \propto d^{-1}\), and the blue line corresponds to \(h \propto d^{-\nicefrac{3}{2}}\).}
    \label{fig:mixing-time-boxes-ellipsoids-simplex}
\end{figure}

From \cref{tab:mixing-time-cond-num-indep}, we see that the conditioning of the domain does not have a discernible and consistent effect on the empirical mixing time.
This is consistent with the theoretical guarantee in \cref{corr:mixing-time-polytope-ellipsoids} which does not depend on the condition number of the domains.
This accomplishes our first objective.
In \cref{fig:mixing-time-boxes-ellipsoids-simplex}, we plot the variation of the empirical mixing time for with dimension for each of the domain types: box, ellipsoid and simplex.
For each of the domain types, we observe a mixing time that is closely proportional to \(h^{-1}\).
Recall that the mixing time guarantee from \cref{corr:mixing-time-polytope-ellipsoids} scales as \(\frac{\nu}{h}\).
For the ellipsoid, \(\nu = 2\), and hence the observation is consistent with our theoretical guarantee.
However, for the box and simplex, we observe a better mixing time guarantee as \(\nu\) is \(2d\) for the box, and \(d + 1\) for the simplex domain types.
This suggests that there is perhaps more structure for these domains that can leveraged.
The box has a product structure, which reduces to independent \(1\)-dimensional sampling for each of the coordinates, and the simplex has \(d\) non-negativity constraints and only 1 constraint involving more than one coordinate.

\subsection{Sampling from Dirichlet distributions}
\label{sec:dirichlet-sampling-empirics}

Here, the goal is to draw approximate samples from a Dirichlet distribution defined in \cref{eq:potential-dirichlet}.
First, we discuss our metric for measuring the mixing time of \nameref{alg:mamla}.
Let \(\widehat{\bbT}^{k}\) be the empirical measure of samples obtained from multiple chains (\(N\)) after \(k\) iterations of \nameref{alg:mamla}, and let \(\widehat{\Pi}\) be the empirical measure of samples from the true Dirichlet distribution (which can be obtained from common scientific computing packages).
We show an example of the progression of samples from \nameref{alg:mamla} in \cref{fig:progress-dirichlet}.
To quantify how close the distribution obtain after running \(k\) iterations of \nameref{alg:mamla}, we compute the empirical 2-Wasserstein distance \(\widetilde{W^{2}_{2}}(\widehat{\bbT}^{k}, \widehat{\Pi})\) using the Sinkhorn algorithm \citep{cuturi2013sinkhorn} with a small regularisation parameter\footnote{We choose this to be \(0.001\).}.

\begin{figure}[H]
    \centering
    \includegraphics[width=0.87\linewidth]{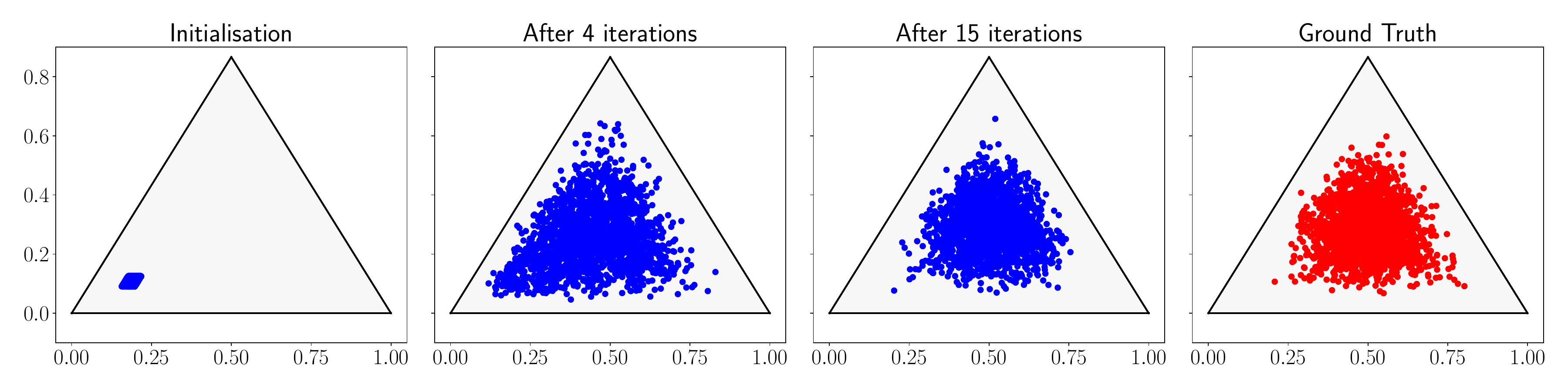}
    \caption{Progression of \nameref{alg:mamla} for sampling from a Dirichlet with \(d = 2\), \(a_{i} = 6\) for \(i \in [3]\).}
    \label{fig:progress-dirichlet}
\end{figure}

\subsubsection{Mixing time versus dimension}

The (empirical) mixing time for this task is defined as
\begin{equation*}
    \widehat{\tau}_{\mathrm{mix}}(\delta) = \inf \left\{k \geq 0 : \widetilde{W^{2}_{2}}(\widehat{\bbT}^{k}, \widehat{\Pi}) \leq \delta\right\}~
\end{equation*}
where \(\widehat{\Pi}\) is the empirical distribution of samples from the Dirichlet distribution we are interested in sampling from.
We set \(\delta = \nicefrac{1}{100}\).
For any dimension \(d\), we set the Dirichlet parameter to \(a_{i} = 3\) for all \(i \in [d]\).
We obtain an approximate sample for a given Dirichlet distribution by running \nameref{alg:mamla} with \(\potential{}\) as defined in \cref{eq:potential-dirichlet}, \(\mirrorfunc{}\) given by the log-barrier of the domain for \(2000\) iterations, and generate samples from \(N = 2000\) chains with step size \(h\).
To capture variation, we perform \(10\) such sampling runs, each independently, and with two choices of step sizes: \(h = \frac{1}{4d^{\nicefrac{3}{2}}}\) and \(h = \frac{1}{4d^{2}}\).
In \cref{fig:mixing-time-dirichlet}, we plot the variation of the mixing time with dimension.
We see that the mixing time closely scales as \(h^{-1}\) (\(d^{1.764}\) for \(h \propto d^{-1.5}\) and \(d^{2.215}\) for \(h \propto d^{-2}\)) as suggested by \cref{corr:dirichlet-sampling}. 

\subsubsection{Comparison to \ref{eq:MLA}}

\begin{figure}[t]
\begin{minipage}{0.33\linewidth}
\begin{figure}[H]
\centering
\includegraphics[width=0.95\linewidth]{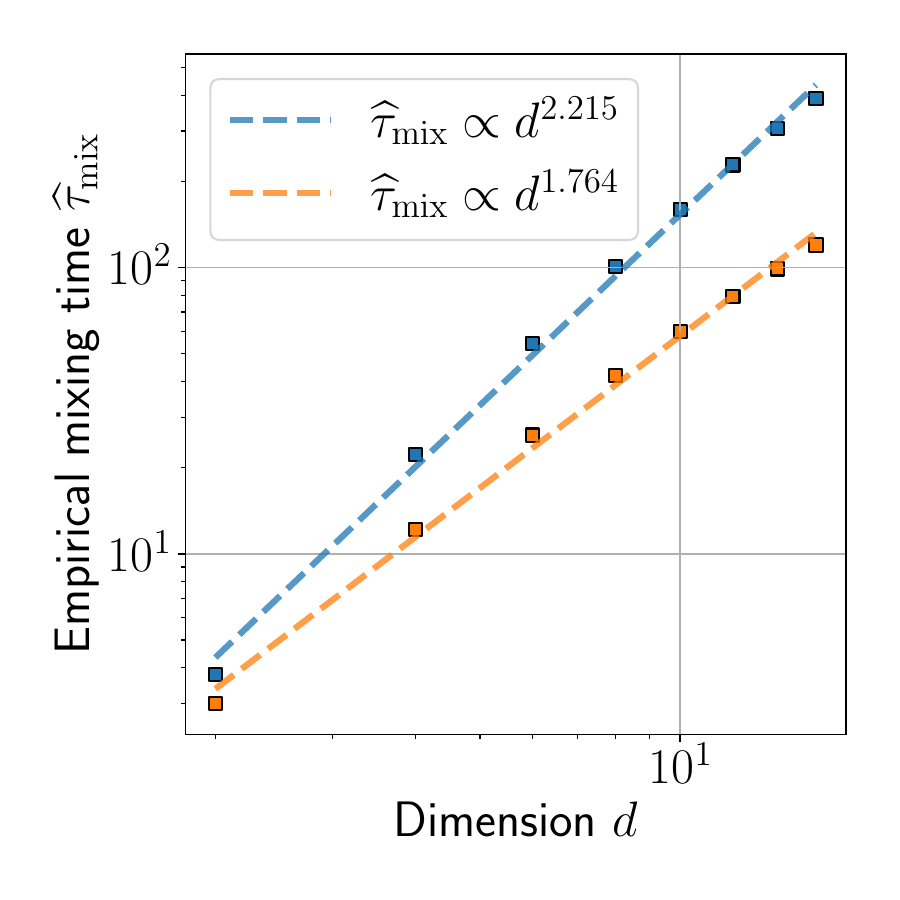}
\caption{(Empirical) Mixing time versus dimension.
Orange corresponds to \(h \propto d^{-\nicefrac{3}{2}}\), and blue corresponds to \(h \propto d^{-2}\).}
\label{fig:mixing-time-dirichlet}
\end{figure}
\end{minipage}
\hfill
\begin{minipage}{0.65\linewidth}
\begin{figure}[H]
\centering
\begin{subfigure}{0.49\linewidth}
\centering
\includegraphics[width=\linewidth]{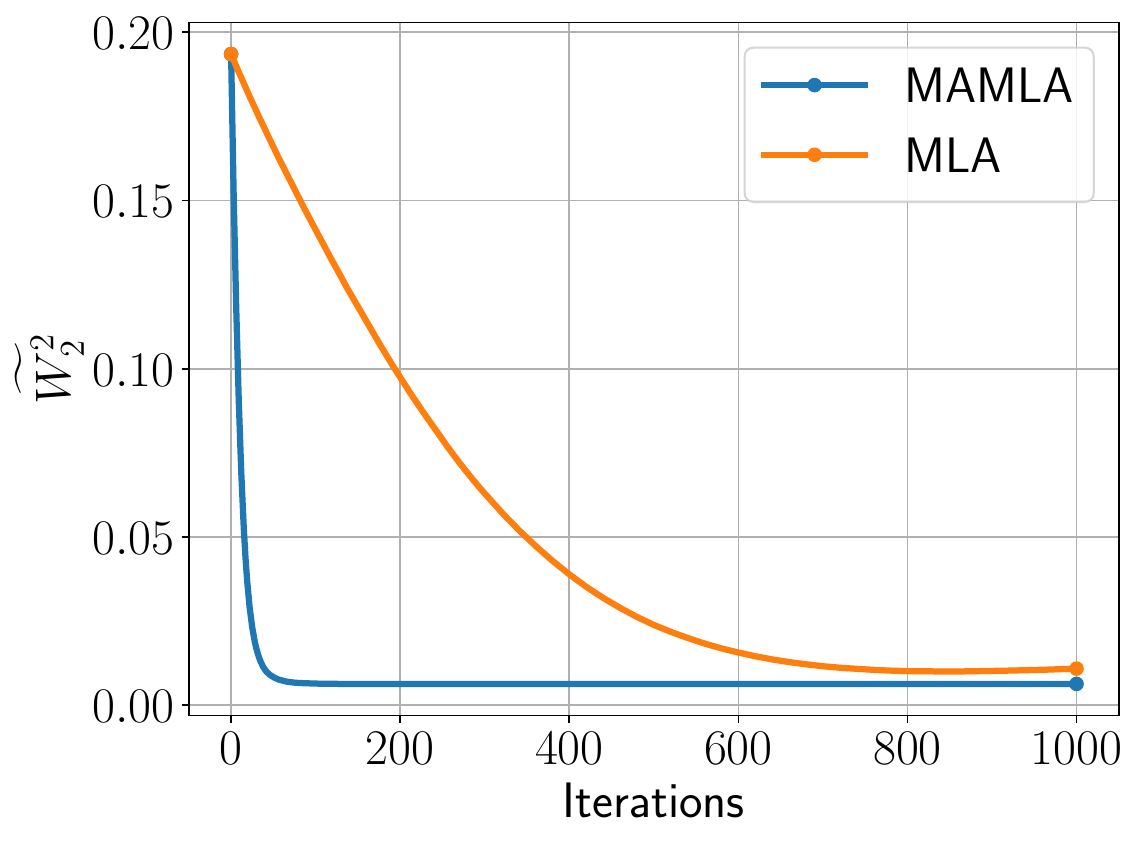}
\caption{\(d = 8\)}
\end{subfigure}
\hfill
\begin{subfigure}{0.49\linewidth}
\centering
\includegraphics[width=\linewidth]{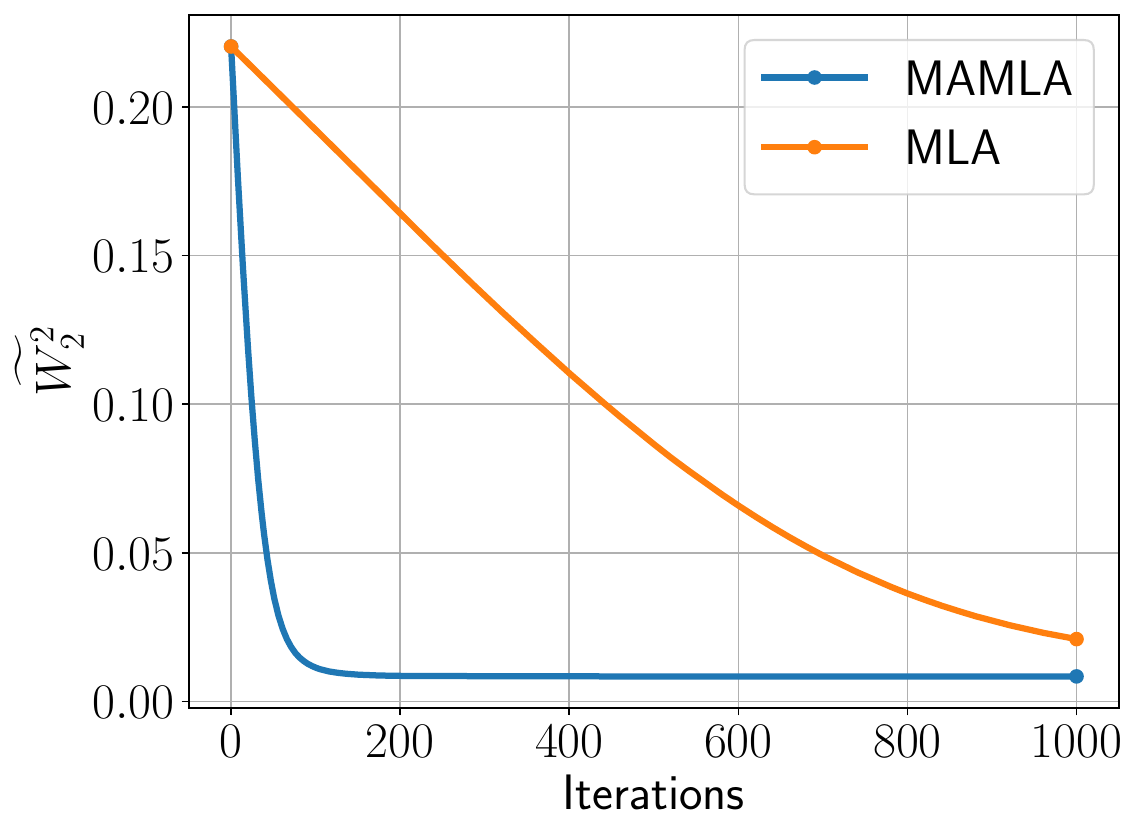}
\caption{\(d = 16\)}
\end{subfigure}
\caption{Variation of empirical \(2\)-Wasserstein distance with iterations for \nameref{alg:mamla} and \ref{eq:MLA}.}
\label{fig:variation-empirical-W2-comp}
\end{figure}
\end{minipage}
\end{figure}

For \(d = 8, 16\) specifically, we compare \nameref{alg:mamla} to \ref{eq:MLA} over the same Dirichlet setup discussed previously.
We run both \nameref{alg:mamla} and \ref{eq:MLA} for 1000 iterations with the same initial conditions, and with step sizes \(h_{\MAMLA{}} = \frac{1}{4d^{\nicefrac{3}{2}}}\) and \(h_{\MLA{}} = \frac{C}{d}\) respectively.
The step size for \ref{eq:MLA} is suggested by the analysis in \cite{li2022mirror}, and \(C\) is a constant that we tune suitably.
In \cref{fig:variation-empirical-W2-comp}, we plot the variation of the empirical \(2\)-Wasserstein distance \(\widetilde{W_{2}^{2}}\) averaged over \(10\) runs for \nameref{alg:mamla} and \ref{eq:MLA}.
For both \(d = 8\) and \(d = 16\), we see that \nameref{alg:mamla} is much faster than \ref{eq:MLA}.
Additionally, for \(d = 16\), the rate of decrease is noticeably slower than for \(d = 8\), as expected.

\subsection{Acceptance rate versus dimension}
\label{sec:accept-rate-analysis-empirics}

Recall from \cref{thm:mix-mamla} that a step size that scales as \(d^{-3}\) ensures fast mixing of \nameref{alg:mamla}.
Therefore, the choices of step size we used for the uniform sampling and Dirichlet sampling tasks discussed earlier are larger than what is sufficient theoretically.
It must be noted that this theoretical suggestion is due to a more conservative, worst-case analysis of the accept-reject step, and a choice of step size that decays slower than \(d^{-3}\) in this worst-case setting causes the probability of acceptance to decay to \(0\) as \(d\) increases.

It would therefore be instructive to understand if indeed for large \(d\) this does happen, because this can help diagnose any looseness in the analysis, and this is the focus of this subsection.
Our empirical setup is motivated by \cite{dwivedi2018log}.
We define the average acceptance rate for a choice of step size to be the average number of accepted proposals over \(4500\) iterations after a burn-in period of \(500\) iterations, and over \(2000\) particles for that step size.
For uniform sampling, we consider step sizes \(h \propto d^{-\gamma}\) for \(\gamma \in \{0.5, 0.75, 1.0, 1.5\}\).
For Dirichlet sampling, we consider step sizes \(h \propto d^{-\gamma}\) for \(\gamma \in \{0.75, 1.0, 1.5, 2.0\}\).
In \cref{fig:accept-rate-all}, we plot the variation of average acceptance rate over dimension for these step size choices and sampling tasks.

\begin{figure}[t]
    \captionsetup[subfigure]{justification=centering}
    \centering
    \begin{subfigure}[t]{0.325\linewidth}
        \centering
        \includegraphics[width=\linewidth]{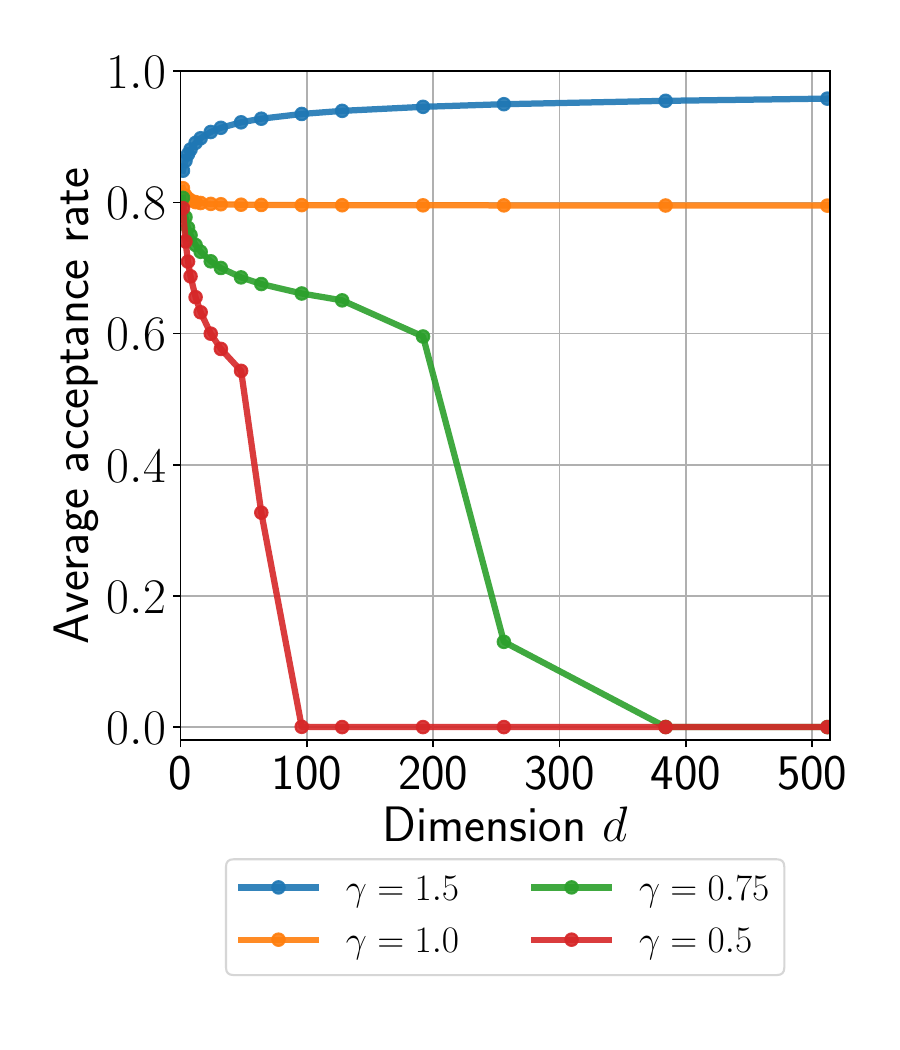}
        \caption{Uniform sampling over \\\(\primalspace = \mathsf{Ellipsoid}\)}    
    \end{subfigure}
    \hfill
    \begin{subfigure}[t]{0.325\linewidth}
        \centering
        \includegraphics[width=\linewidth]{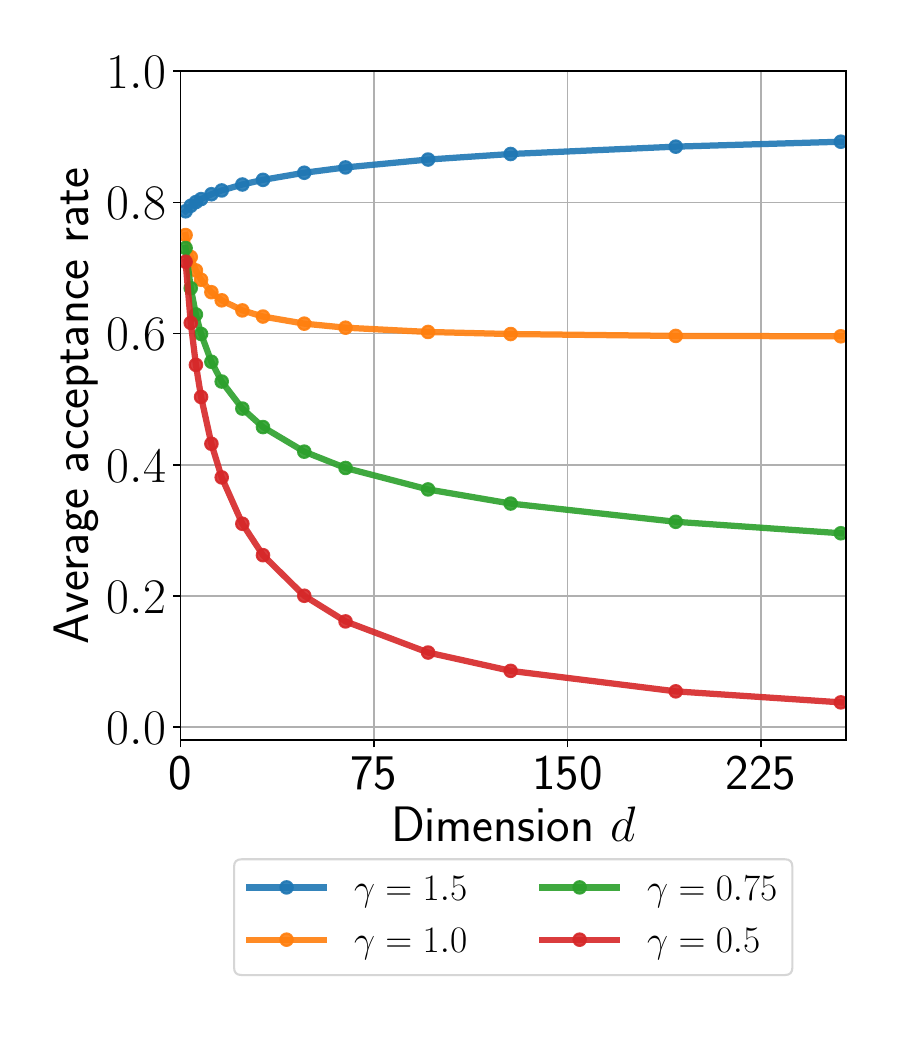}
        \caption{Uniform sampling over \\\(\primalspace = \mathsf{S}_{d}\)}    
    \end{subfigure}
    \hfill
    \begin{subfigure}[t]{0.325\linewidth}
        \centering
        \includegraphics[width=\linewidth]{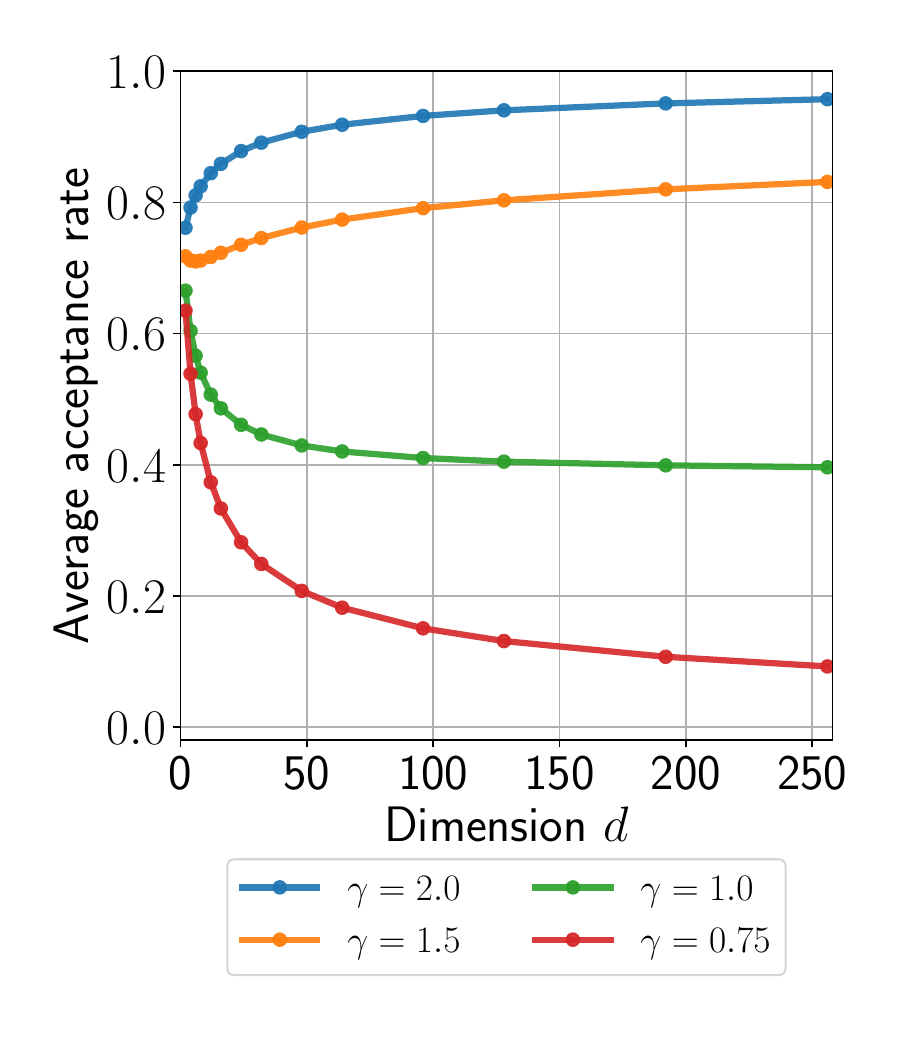}
        \caption{Sampling from \\Dirichlet distribution}
    \end{subfigure}
    \caption{Variation of average acceptance rate with step size \(h \propto d^{-\gamma}\) for various sampling tasks.}
    \label{fig:accept-rate-all}
\end{figure}

From these plots, we can see that the larger step sizes considered indeed work for large dimensions as well, pointing out the conservative nature of the analysis.
For the uniform sampling task, when \(\gamma < 1\), we see that the acceptance rate decays to \(0\) as \(d\) increases, and for \(\gamma \geq 1\), the acceptance rate is sufficiently high, and seems to plateau when \(\gamma = 1\).
For the Dirichlet sampling task, when \(\gamma > 1\), we see that the acceptance rate is sufficiently high.
When \(\gamma = 1\), it appears as though the average acceptance rate has plateaued, however at a level lower than the nominal \(50\%\).
The purpose of considering the average acceptance rate is due to its connection (in the population limit) to the mixing time analysis.
Informally, let \(h^{\star}\) be the largest step size such which results in a non-decaying average acceptance rate.
Then, as noted in \cite{dwivedi2018log}, the mixing time is expected to be proportional to \(\frac{1}{h^{\star}}\).
In this section, we have empirically demonstrated that there are step size choices that scale much higher than \(d^{-3}\) which result in a non-decaying acceptance rate, which would consequently imply an improved mixing time guarantee for \nameref{alg:mamla}.
We leave obtaining such an improved mixing time guarantee for future work.

\section{Proofs}
\label{sec:proofs}
This section is devoted providing the proofs of the main theorem and its corollaries in \cref{sec:algo}.
In \cref{sec:proofs:main-lemmas}, we give the main lemmas that form the proof of the \cref{thm:mix-mamla}, whose proof follows in \cref{sec:proofs:full-proof-thm}.
The proofs of these main lemmas are given in \cref{sec:proofs:main-lemmas-proofs}, and finally in \cref{sec:proofs:corollaries}, we give the proofs of the corollaries from \cref{sec:algo:applications-thm}.

\subsection{Results for conductance using one-step overlap}
\label{sec:proofs:main-lemmas}

The following classical result by \cite{lovasz1993random} is the basis of the mixing time guarantee of \nameref{alg:mamla}, which is a common tool for studying reversible Markov chains.
\begin{proposition}
\label{prop:ls93-mixing-result}
Let \(\bfQ\) be a lazy, reversible Markov chain over \(\primalspace\) with stationary distribution \(\Pi\) and conductance \(\Phi_{\bfQ}\).
For any \(\Pi_{0}\) that is \(M\)-warm with respect to \(\Pi\), we have for all \(k \geq 1\) that
\begin{equation*}
    \TVdist(\bbT_{\bm{Q}}^{k}\Pi_{0}, \Pi) \leq \sqrt{M} \left(1 - \frac{\Phi_{\bfQ}^{2}}{2}\right)^{k}~.
\end{equation*}
\end{proposition}

The above result indicates that it suffices to show a non-vacuous lower bound on the conductance of a reversible Markov chain to derive mixing time upper bounds for it.
To do so, we use the \emph{one-step overlap} technique developed in \cite{lovasz1999hit}.
The gist of this technique is to show that when given a reversible Markov chain \(\bfQ = \{\calQ_{x} : x \in \primalspace\}\), the distance \(\TVdist(\bbT_{\bm{Q}}(\delta_{x}), \bbT_{\bm{Q}}(\delta_{y}))\) (also equal to \(\TVdist(\calQ_{x}, \calQ_{y})\)) is uniformly bounded away from \(1\) for any \(x, y \in \interior{\primalspace}\) that are close.
For \nameref{alg:mamla}, recall that \(\bfT = \{\calT_{x} : x \in \primalspace\}\) denotes the Markov chain induced by a single iteration of the algorithm, and \(\calP_{x}\) is proposal distribution defined whose density is given in \cref{eq:proposal-density}.

\cref{lem:conductance-lower-strong-RelConv,lem:conductance-lower-weak} precisely state how a one-step overlap can yield lower bounds on the conductance of \(\bfT\), whose stationary distribution has density \(\targetdens \propto e^{-\potential{}}\).
These are two separate lemmas, one for when \(\potential{}\) is \(\mu\)-strongly convex relative to a self-concordant mirror function \(\mirrorfunc{}\) with \(\mu > 0\), and the other for when the potential \(\potential{}\) is simply convex i.e., when \(\mu = 0\).

\begin{lemma}
\label{lem:conductance-lower-strong-RelConv}
Let \(\mirrorfunc{} : \primalspace \to \bbR \cup \{\infty\}\) be a self-concordant function with parameter \(\alpha\), and let \(\potential{}\) be \(\mu\)-relatively convex with respect to \(\mirrorfunc{}\).
Assume that for any \(x, y \in \interior{\primalspace}\), if \(\|x - y\|_{\hessmirror{y}} \leq \Delta\) where \(\Delta \leq \frac{1}{2\alpha}\), then \(\TVdist(\calT_{x}, \calT_{y}) \leq \frac{1}{4}\).
Then, for any measurable partition \(\{A_{1}, A_{2}\}\) of \(\primalspace\),
\begin{equation*}
    \int_{A_{1}} \calT_{x}(A_{2}) \targetdens(x) dx \geq \frac{3}{16} \cdot \min\left\{1, \frac{\sqrt{\mu} \cdot \Delta}{2 \cdot (8\alpha + 4)}\right\} \cdot \min\{\targetdist(A_{1}), \targetdist(A_{2})\}~.
\end{equation*}
\end{lemma}

\begin{lemma}
\label{lem:conductance-lower-weak}
Let \(\mirrorfunc{} : \primalspace \to \bbR \cup \{\infty\}\) be a symmetric barrier with parameter \(\nu\), and let \(\potential{} : \primalspace \to \bbR \cup \{\infty\}\) be a convex function.
Assume that for any \(x, y \in \interior{\primalspace}\), if \(\|x - y\|_{\hessmirror{y}} \leq \Delta\), then \(\TVdist(\calT_{x}, \calT_{y}) \leq \frac{1}{4}\).
Then, for any measurable partition \(\{A_{1}, A_{2}\}\) of \(\primalspace\),
\begin{equation*}
    \int_{A_{1}} \calT_{x}(A_{2})\targetdens(x) dx \geq \frac{3}{16} \cdot \min\left\{1, \frac{1}{8} \cdot \frac{\Delta}{\sqrt{\nu}}\right\} \cdot \min\{\targetdist(A_{1}), \targetdist(A_{2})\}~.
\end{equation*}
\end{lemma}

These lemmas are based on two different isoperimetric inequalities.
Specifically, \cref{lem:conductance-lower-strong-RelConv} relies on an isoperimetric inequality for distributions whose potentials \(\potential{}\) are relatively convex with respect to a self-concordant function, and \cref{lem:conductance-lower-weak} uses another isoperimetric inequality for log-concave distributions supported on bounded sets from \cite{lovasz2003hit} in conjunction with the notion of symmetric barriers to arrive at the result.

Next, we establish the one-step overlap that we have assumed in the aforementioned lemmas.
The triangle inequality for the TV distance gives for any \(x, y \in \interior{\primalspace}\)
\begin{equation*}
    \TVdist(\calT_{x}, \calT_{y}) \leq \TVdist(\calT_{x}, \calP_{x}) + \TVdist(\calP_{x}, \calP_{y}) + \TVdist(\calP_{y}, \calT_{y})~.
\end{equation*}
The goal is to show that \(\TVdist(\calT_{x}, \calT_{y})\) is sufficiently bounded away from \(1\), and one method to achieve this is by showing that each of the three quantities above are also sufficiently bounded away from \(1\).
In the following lemma, we give bounds for each of these quantities.
First, we define the function \(b : [0, 1) \to (0, 1)\) as
\begin{equation}
\label{eq:h-max-precise}
    b(\varepsilon; d, \alpha, \beta, \lambda) = \min\left\{1, \frac{\calC_{1}(\varepsilon)}{\alpha^{2} \cdot d^{3}}~,~ \frac{\calC_{2}(\varepsilon)}{d \cdot \gamma}~,~ \frac{\calC_{3}(\varepsilon)}{\beta^{\nicefrac{2}{3}} \cdot \alpha^{\nicefrac{4}{3}}}~,~ \frac{\calC_{4}(\varepsilon)}{\beta^{\nicefrac{2}{3}} \cdot \gamma^{\nicefrac{2}{3}}}~,~ \frac{\varepsilon}{4 \cdot \beta^{2}}\right\}~,
\end{equation}
where \(\calC_{1}(\varepsilon), \ldots, \calC_{4}(\varepsilon)\) are functions of \(\varepsilon\) defined in \cref{eq:h-max-constants-precise} later.

\begin{lemma}
\label{lem:one-step-overlap}
Let \(\mirrorfunc{}\) and \(\potential{}\) satisfy \ref{assump:self-concord}, \ref{assump:rel-convex-smooth}, and \ref{assump:rel-lipschitz}.
Then, the following statements hold.
\begin{enumerate}
    \item For any two points \(x, y \in \interior{\primalspace}\) such that \(\|x - y\|_{\hessmirror{y}} \leq \frac{\sqrt{h}}{10 \cdot \bar{\alpha}}\) and any step size \(h \in (0, 1)\),
    \begin{equation*}
        \TVdist(\calP_{x}, \calP_{y}) \leq \frac{1}{2}\sqrt{\frac{h \cdot d}{20} + \frac{1}{80} + \frac{9}{2} \cdot h \cdot \beta^{2}}~.
    \end{equation*}
    \item For any \(x \in \interior{\primalspace}\), \(\varepsilon < 1\), and step size \(h \in (0, b(\varepsilon; d, \alpha, \beta, \lambda)]\).
    \begin{equation*}
        \TVdist(\calT_{x}, \calP_{x}) \leq \frac{\varepsilon}{2}~.
    \end{equation*}
\end{enumerate}
where \(\bar{\alpha} = \max\{1, \alpha\}\), \(\gamma = \frac{\lambda}{2} + \alpha \cdot \beta\).
\end{lemma}

The proofs of \cref{lem:conductance-lower-strong-RelConv,lem:conductance-lower-weak,lem:one-step-overlap} are given in \cref{sec:proofs:main-lemmas-proofs}.
With these results, we proceed to the proof of \cref{thm:mix-mamla}.

\subsection{Complete proof of \cref{thm:mix-mamla}}
\label{sec:proofs:full-proof-thm}

\begin{proof}
Let \(h \in (0, b(\varepsilon; d, \alpha, \beta, \lambda)] \subseteq (0, 1]\).
From \cref{lem:one-step-overlap}, we have that for any \(x, y \in \interior{\primalspace}\) such that \(\|x - y\|_{\hessmirror{y}} \leq \frac{\sqrt{h}}{10 \cdot \bar{\alpha}}\),
\begin{equation*}
    \TVdist(\calP_{x}, \calP_{y}) \leq \frac{1}{2}\sqrt{\frac{h \cdot d}{20} + \frac{1}{80} + \frac{9}{2} \cdot h \cdot \beta^{2}} \leq \frac{1}{2}\sqrt{\frac{1}{20} + \frac{1}{80} + \frac{9}{8} \cdot \varepsilon}~.
\end{equation*}
For \(\varepsilon = \nicefrac{1}{24}\),
\begin{gather*}
    \TVdist(\calP_{x}, \calP_{y}) \leq \frac{1}{2}\sqrt{\frac{1}{20} + \frac{1}{80} + \frac{3}{64}} \leq \frac{1}{6} \\
    \TVdist(\calT_{x}, \calP_{x})~, \TVdist(\calT_{y}, \calP_{y}) \leq \frac{1}{48}
\end{gather*}
Consequently, \(\TVdist(\calT_{x}, \calT_{y}) \leq \frac{1}{24} + \frac{1}{6} \leq \frac{1}{4}\).

Let \(\{A_{1}, A_{2}\}\) be an arbitrary measurable partition of \(\primalspace\).
When \(\potential{}\) is \(\mu\)-strongly convex with respect to \(\mirrorfunc{}\) such that \(\mu > 0\), we invoke \cref{lem:conductance-lower-strong-RelConv} with \(\Delta = \frac{\sqrt{h}}{10 \cdot \bar{\alpha}}\) to obtain
\begin{equation*}
    \int_{A_{1}} \calT_{x}(A_{2})\targetdens(x)dx \geq \frac{3}{16} \cdot \min\left\{1, ~\frac{C_{\alpha} \cdot \sqrt{\mu} \cdot \sqrt{h}}{20 \cdot \bar{\alpha}}\right\} \cdot \min\{\targetdist(A_{1}), \targetdist(A_{2})\}~;\quad C_{\alpha} = \frac{1}{8\alpha + 4}~.
\end{equation*}
Consequently,
\begin{equation*}
    \Phi_{\bfT} \geq \frac{3}{16} \cdot \min\left\{1, \frac{C_{\alpha} \cdot \sqrt{\mu} \cdot \sqrt{h}}{20 \cdot \bar{\alpha}} \right\}~.
\end{equation*}

Instead, when \(\mu = 0\) i.e., when \(\potential{}\) is convex, the additional assumption about \(\mirrorfunc{}\) being a symmetric barrier with parameter \(\nu\) allows us to use \cref{lem:conductance-lower-weak} with \(\Delta = \frac{\sqrt{h}}{10 \cdot \bar{\alpha}}\) to get
\begin{equation*}
    \int_{A_{1}} \calT_{x}(A_{2}) \targetdens(x) dx \geq \frac{3}{16} \cdot \min\left\{1, \frac{\sqrt{h}}{80 \cdot \sqrt{\nu} \cdot \bar{\alpha}}\right\} \cdot \min\{\targetdist(A_{1}), \targetdist(A_{2})\}~,
\end{equation*}
and as a result,
\begin{equation*}
    \Phi_{\bfT} \geq \frac{3}{16} \cdot \min\left\{1, \frac{\sqrt{h}}{80 \cdot \sqrt{\nu} \cdot \bar{\alpha}}\right\} ~.
\end{equation*}

Since \(\Pi_{0}\) is \(M\)-warm with respect to \(\targetdist\), we instantiate \cref{prop:ls93-mixing-result}, which provides bounds on the mixing time for any \(\delta \in (0, 1)\) in the two cases discussed above.
To be precise,
\begin{equation*}
    \TVdist(\bbT^{k}_{\bfT}\Pi_{0}, \targetdist) \leq \sqrt{M} \cdot \left(1 - \frac{\Phi_{\bfT}^{2}}{2}\right)^{k} \leq \sqrt{M} \cdot \exp\left(-k \cdot \frac{\Phi_{\bfT}^{2}}{2}\right)~.
\end{equation*}
When \(k = \frac{2}{\Phi_{\bfT}^{2}} \log\left(\frac{\sqrt{M}}{\delta}\right)\), \(\TVdist(\bbT^{k}_{\bfT}\Pi_{0}, \Pi) \leq \delta\).
Thus,
\begin{description}
    \item [when \(\mu > 0\)]
    \begin{equation*}
        \mixingtime{\delta; \bfT, \Pi_{0}} = \calO\left(\max\left\{1, \frac{\bar{\alpha}^{4}}{\mu \cdot h} \right\} \cdot \log\left(\frac{\sqrt{M}}{\delta}\right)\right)~,\text{ and}
    \end{equation*}
    \item [when \(\mu = 0\)]
    \begin{equation*}
        \mixingtime{\delta; \bfT, \Pi_{0}} = \calO\left(\max\left\{1, \frac{\bar{\alpha}^{2} \cdot \nu}{h}\right\} \cdot \log\left(\frac{\sqrt{M}}{\delta}\right)\right)~.
    \end{equation*}
\end{description}

\end{proof}

\subsection{Proofs of conductance results and one-step overlap in \cref{sec:proofs:main-lemmas}}
\label{sec:proofs:main-lemmas-proofs}

Prior to presenting the the proofs of \cref{lem:conductance-lower-strong-RelConv,lem:conductance-lower-weak}, we first introduce and elaborate on two specific isoperimetric lemmas which we use in the proofs, which follow after this introduction.

The proof of \cref{lem:conductance-lower-strong-RelConv} uses an isoperimetric lemma for a distribution supported on a compact, convex subset \(\primalspace \subset \bbR^{d}\) with density \(\targetdens \propto e^{-\potential{}}\) where \(\potential{}\) is relatively strongly convex with respect to a self-concordant function \(\psi\) with domain \(\primalspace\).
This lemma is due to \cite{gopi2023algorithmic}; their original lemma assumes the self-concordance parameter of \(\psi\) to be \(1\) while we generalise it for \(\psi\) with an arbitrary self-concordance parameter \(\alpha\).
Additionally, we have the following metric \(d_{\psi} : \primalspace \times \primalspace \to \bbR_{+} \cup \{0\}\).
For any two points \(x, y \in \primalspace\), \(d_{\psi}(x, y)\) is the Riemannian distance between \(x, y\) as measured with respect to the Hessian metric \(\nabla^{2}\psi\) over \(\primalspace\).
We overload this notation to take sets as arguments; for two disjoint sets \(A, B\) of \(\primalspace\), \(d_{\psi}(A, B)\) is defined as \(d_{\psi}(A, B) = \inf_{x \in A, y \in B} d_{\psi}(x, y)\).
We state the lemma due to \cite{gopi2023algorithmic} as stated in \cite{kook2023efficiently}, which addresses the subtlety of \(\potential{}\) only taking finite values in \(\interior{\primalspace}\), and also highlight the dependence on the self-concordance parameter.

\begin{proposition}[{\citet[Lemma 2.7]{kook2023efficiently}}]
\label{lem:self-concordant-isoperimetry}
Let \(\psi : \primalspace \to \bbR \cup \{\infty\}\) be a self-concordant function with parameter \(\alpha\), and let \(\potential{}\) be a \(\mu\)-relatively convex function with respect to \(\psi\).
Then, for any given partition \(\{A_{1}, A_{2}, A_{3}\}\) of \(\primalspace\),
\begin{equation*}
    \targetdist(A_{3}) \geq C_{\alpha} \cdot \sqrt{\mu} \cdot d_{\psi}(A_{1}, A_{2}) \cdot \min\left\{\targetdist(A_{1}) ~, \targetdist(A_{2})\right\}
\end{equation*}
for positive constant \(C_{\alpha} = \frac{1}{8\alpha + 4}\).
\end{proposition}

To operationalise this in order to obtain lower bounds on the conductance using the one-step overlap which uses a closeness criterion in terms of \(\|x - y\|_{\hessmirror{y}}\), we require a relation between \(d_{\mirrorfunc{}}(x, y)\) and \(\|x - y\|_{\hessmirror{y}}\).
The self-concordance property of \(\mirrorfunc{}\) ensures that the metric \(\nabla^{2}\mirrorfunc{}\) is relatively stable between two points that are close, and thus implying that the Riemannian distance between \(x, y\) is approximately the local norm \(\|x - y\|_{\hessmirror{y}}\).
The following lemma from \cite{nesterov2002riemannian} formalises this intuition.
The lemma below is a slight modification of their original lemma that makes the dependence on the self-concordance parameter explicit.

\begin{proposition}[{\citet[Lemma 3.1]{nesterov2002riemannian}}]
\label{lem:hessian-metric-to-hessian-norm}
Let \(\psi : \primalspace \to \bbR \cup \{\infty\}\) be a convex and self-concordant with parameter \(\alpha\).
For any \(\Delta \in [0, \nicefrac{1}{\alpha})\), if for any \(x, y \in \interior{\primalspace}\) it holds that \(d_{\psi}(x, y) \leq \Delta - \alpha \cdot \frac{\Delta^{2}}{2}\), then \(\|x - y\|_{\nabla^{2}\psi(x)} \leq \Delta\).
\end{proposition}

Next, for the proof of \cref{lem:conductance-lower-weak}, we use an isoperimetric inequality for a log-concave distribution \(\targetdist\) from \cite{lovasz2003hit}.
Before we state this, we introduce the notion of the cross ratio between two points in \(\primalspace\), which is compact.
Let \(x, y \in \interior{\primalspace}\), and consider a line segment between \(x\) and \(y\).
Let the endpoints of this extension of the line segment to the boundary of \(\primalspace\) be \(p\) and \(q\) respectively (thus forming a chord), with points in the order \(p, x, y, q\).
The cross ratio between \(x\) and \(y\) (with respect to \(\primalspace\)) is
\begin{equation*}
    \crossratio{x}{y}{\primalspace} = \frac{\|y - x\| \cdot \|q - p\|}{\|y - q\| \cdot \|p - x\|}~.
\end{equation*}
We overload this notation to define the cross ratio between sets as
\begin{equation*}
    \crossratio{A_{2}}{A_{1}}{\primalspace} = \inf\limits_{y \in A_{2},~x \in A_{1}} \crossratio{y}{x}{\primalspace}~.
\end{equation*}

\begin{proposition}[{\citet[Thm. 2.2]{lovasz2003hit}}]
\label{lem:log-concave-isoperimetry}
    Let \(\targetdist\) be a log-concave distribution on \(\bbR^{d}\) with support \(\primalspace\).
    For any measurable partition \(\{A_{1}, A_{2}, A_{3}\}\) of \(\primalspace\),
    \begin{equation*}
        \targetdist(A_{3}) \geq \crossratio{A_{2}}{A_{1}}{\primalspace} \cdot \targetdist(A_{2}) \cdot \targetdist(A_{1})~.
    \end{equation*}
\end{proposition}
As seen previously, to operationalise this isoperimetric inequality, we have to relate the cross ratio between two points to the local norm \(\|x - y\|_{\hessmirror{y}}\).
When \(\mirrorfunc{}\) is a symmetric barrier, this is possible due to the following lemma.
\begin{proposition}[{\citet[Lem. 2.3]{laddha2020strong}}]
\label{lem:sym-barrier-cross-ratio}
    Let \(\psi : \primalspace \to \bbR \cup \{\infty\}\) be a symmetric barrier with parameter \(\nu\).
    Then, for any \(x, y \in \interior{\primalspace}\), \(\crossratio{y}{x}{\primalspace} \geq \frac{\|y - x\|_{\nabla^{2}\psi(y)}}{\sqrt{\nu}}\).
\end{proposition}

With all of these ingredients, we give the proofs of \cref{lem:conductance-lower-strong-RelConv,lem:conductance-lower-weak}.
The proofs of both lemmas only differ in the isoperimetric inequalities used, and hence are presented together.

\subsubsection{Proofs of \cref{lem:conductance-lower-strong-RelConv,lem:conductance-lower-weak}}
\label{prf:lem:conductance-lower}
\begin{proof}
By the reversibility of the Markov chain \(\bfT\), we have for any partition \(\{A_{1}, A_{2}\}\) of \(\primalspace\) that
\begin{equation}
\label{eq:equality-transitions}
    \int_{A_{1}}\calT_{x}(A_{2})\targetdens(x) dx = \int_{A_{2}}\calT_{x}(A_{1})\targetdens(x) dx~.
\end{equation}
Define the subsets \(A_{1}'\) and \(A_{2}'\) of \(A_{1}\) and \(A_{2}\) respectively as
\begin{equation*}
    A_{1}' = \left\{x \in A_{1} ~:~ \calT_{x}(A_{2}) < \frac{3}{8}\right\}; \qquad A_{2}' = \left\{x \in A_{2} ~:~ \calT_{x}(A_{1}) < \frac{3}{8}\right\}~.
\end{equation*}

We consider two cases, which cover all possibilities
\begin{description}
\item [Case 1]: \(\Pi(A_{1}') \leq \frac{\Pi(A_{1})}{2}\) or \(\Pi(A_{2}') \leq \frac{\Pi(A_{2})}{2}\).
\item [Case 2]: \(\Pi(A_{i}') > \frac{\Pi(A_{i})}{2}\) for \(i \in \{1, 2\}\)
\end{description}

\paragraph*{Lower bound from case 1:}
Let \(\Pi(A_{1}') \leq \frac{\Pi(A_{1})}{2}\).
Then,
\begin{equation*}
    \Pi(A_{1} \cap (\primalspace \setminus A_{1}')) = \Pi(A_{1}) - \Pi(A_{1}') \geq \frac{\Pi(A_{1})}{2}.
\end{equation*}
Consequently,
\begin{equation*}
    \int_{A_{1}}\calT_{x}(A_{2})\targetdens(x)dx \geq \int_{A_{1} \cap (\primalspace \setminus A_{1}')} \calT_{x}(A_{2})\targetdens(x)dx \geq \int_{A_{1} \cap (\primalspace \setminus A_{1}')} \frac{3}{8} \cdot \targetdens(x) dx \geq \frac{3}{16} \cdot \targetdist(A_{1})~.
\end{equation*}
The first inequality is due to \(A_{1} \cap (\primalspace \setminus A_{1}') \subset A_{1}\).
The second inequality uses the fact that since \(x \in A_{1} \cap (\primalspace \setminus A_{1}')\), \(x \in A_{1} \setminus A_{1}'\), and for such \(x\), \(\calT_{x}(A_{2}) \geq \frac{3}{8}\).
The final step uses the fact shown right above.
We can use this same technique to show that when \(\targetdist(A_{2}') \leq \frac{\targetdist(A_{2})}{2}\),
\begin{equation*}
    \int_{A_{2}}\calT_{x}(A_{1})\targetdens(x) dx \geq \frac{3}{16}\cdot \targetdist(A_{2})~.
\end{equation*}
Combining these gives,
\begin{equation*}
    \int_{A_{1}}\calT_{x}(A_{2}) \targetdens(x) dx \geq \frac{3}{16}\cdot \min\{\Pi(A_{1}), \Pi(A_{2})\}~.
\end{equation*}

\paragraph*{Lower bound from case 2:}
We start by considering arbitrary \(x' \in A_{1}'\) and \(y' \in A_{2}'\).
From \cref{eq:equality-transitions},
\begin{align*}
    \int_{A_{1}}\calT_{x}(A_{2})\targetdens(x)dx &= \frac{1}{2}\left(\int_{A_{1}}\calT_{x}(A_{2})\targetdens(x)dx + \int_{A_{2}}\calT_{y}(A_{1})\targetdens(y)dy \right) \\
    &\geq \frac{1}{2}\left(\int_{A_{1} \cap (\primalspace \setminus A_{1}')}\calT_{x}(A_{2})\targetdens(x)dx + \int_{A_{2} \cap (\primalspace \setminus A_{2}')}\calT_{y}(A_{1})\targetdens(y)dy \right) \\
    &\geq \frac{3}{16}\left(\int_{A_{1} \cap (\primalspace \setminus A_{1}')}\targetdens(x)dx + \int_{A_{2} \cap (\primalspace \setminus A_{2}')}\targetdens(y)dy\right) \\
    &= \frac{3}{16} \cdot \targetdist(\primalspace \setminus (A_{1}' \cup A_{2}')) = \frac{3}{16} \cdot \targetdist(\primalspace \setminus A_{1}' \setminus A_{2}')~.\numberthis
    \label{eq:lower-bound-conductance-pre-inequality}
\end{align*}

Also, by the definition of \(\TVdist(\calT_{x'}, \calT_{y'})\), and using the fact that \(A_{1}\) and \(A_{2}\) form a partition,
\begin{equation*}
    \TVdist(\calT_{x'}, \calT_{y'}) \geq \calT_{x'}(A_{1}) - \calT_{y'}(A_{1}) = 1 - \calT_{x'}(A_{2}) - \calT_{y'}(A_{1}) > 1 - \frac{3}{8} - \frac{3}{8} = \frac{1}{4}.
\end{equation*}
At this juncture, the proofs of \cref{lem:conductance-lower-strong-RelConv,lem:conductance-lower-weak} will differ due to the different isoperimetric inequalities discussed previously.

\paragraph{When \(\mu > 0\)}

We first discuss \textbf{Case 2} in the context of \cref{lem:conductance-lower-strong-RelConv}.
From the assumption of the lemma, this implies by contraposition that \(\|x' - y'\|_{\hessmirror{y'}} > \Delta\) for \(\Delta \leq \frac{1}{2\alpha}\).
Let \(A_{3}' = \primalspace \setminus A_{1}' \setminus A_{2}'\).
Note that \(\{A_{1}', A_{2}', A_{3}'\}\) forms a partition of \(\primalspace\).
Since \(\potential{}\) is \(\mu\)-relatively convex with respect to \(\mirrorfunc{}\), which is self-concordant with parameter \(\alpha\), \cref{lem:self-concordant-isoperimetry} then gives
\begin{equation*}
    \targetdist(A_{3}') \geq C_{\alpha} \cdot \sqrt{\mu} \cdot d_{\phi}(A_{2}', A_{1}') \cdot \min\left\{\targetdist(A_{1}'), \targetdist(A_{2}')\right\}~.
\end{equation*}

From \cref{lem:hessian-metric-to-hessian-norm}, if \(\|x' -  y'\|_{\hessmirror{y'}} > \Delta\) for \(\Delta \leq \frac{1}{2\alpha}\), \(d_{\mirrorfunc{}}(y', x') > \Delta - \alpha \cdot \frac{\Delta^{2}}{2}\).
Since \(x' \in A_{1}', y' \in A_{2}'\) are arbitrary, this holds for all pairs of \((x', y') \in A_{1}' \times A_{2}'\), and hence
\begin{equation*}
    d_{\phi}(A_{2}', A_{1}') = \inf_{x' \in A_{1}', y' \in A_{2}'} d_{\phi}(y', x') > \Delta - \alpha \cdot \frac{\Delta^{2}}{2} \geq \frac{\Delta}{2}
\end{equation*}
where the final inequality uses the fact that \(t - \frac{t^{2}}{2} \geq \frac{t}{2}\) for \(t \in [0, 0.5]\) with \(t = \alpha \cdot \Delta\).
Substituting the above two inequalities in \cref{eq:lower-bound-conductance-pre-inequality}, we get the lower bound
\begin{align*}
    \int_{A_{1}}\calT_{x}(A_{2})\targetdens(x) dx &\geq C_{\alpha} \cdot \sqrt{\mu} \cdot \Delta \cdot \frac{3}{32} \cdot \min\{\targetdist(A_{1}'), \targetdist(A_{2}')\} \\
    &\geq C_{\alpha} \cdot \sqrt{\mu} \cdot \Delta \cdot \frac{3}{64} \cdot \min\{\targetdist(A_{1}), \targetdist(A_{2})\}~.
\end{align*}
where the final inequality uses the assumption of the case that \(\targetdist(A_{i}') > \frac{\targetdist(A_{i})}{2}\) for \(i \in \{1, 2\}\).

\paragraph{When \(\mu = 0\)}
Next, we discuss \textbf{Case 2} in the context of \cref{lem:conductance-lower-weak}.
From \cref{lem:log-concave-isoperimetry},
\begin{equation*}
    \targetdist(\primalspace \setminus A_{1}' \setminus A_{2}') \geq \crossratio{A_{2}'}{A_{1}'}{\primalspace} \cdot \targetdist(A_{1}') \cdot \targetdist(A_{2}')~.
\end{equation*}
From \cref{lem:sym-barrier-cross-ratio}, \(\crossratio{y'}{x'}{\primalspace} \geq \frac{\|y' - x'\|_{\hessmirror{y'}}}{\sqrt{\nu}}\).
Since \(x' \in A_{1}', y' \in A_{2}'\) are arbitrary, this holds for all pairs of \((x', y') \in A_{1}' \times A_{2}'\), and hence
\begin{equation*}
    \crossratio{A_{2}'}{A_{1}'}{\primalspace} = \inf_{y' \in A_{2}', x' \in A_{1}'} \crossratio{y'}{x'}{\primalspace} \geq \inf_{y' \in A_{2}', x' \in A_{1}'} \frac{\|y' - x'\|_{\hessmirror{y'}}}{\sqrt{\nu}} \geq \frac{\Delta}{\sqrt{\nu}}~.
\end{equation*}

Substituting the above two inequalities in \cref{eq:lower-bound-conductance-pre-inequality}, we get the lower bound
\begin{align*}
    \int_{A_{1}} \calT_{x}(A_{2})\targetdens(x)dx &\geq \frac{3}{16} \cdot \frac{\Delta}{\sqrt{\nu}} \cdot \Pi(A_{2}') \cdot \Pi(A_{1}') \\
    &\geq \frac{3}{64} \cdot \frac{\Delta}{\sqrt{\nu}} \cdot \Pi(A_{2}) \cdot \Pi(A_{1}) \\
    &\geq \frac{3}{128} \cdot \frac{\Delta}{\sqrt{\nu}} \cdot \min\left\{\Pi(A_{1}), \Pi(A_{2})\right\}
\end{align*}
where the second inequality uses the assumption of the case that \(\targetdist(A_{i}') > \frac{\targetdist(A_{i})}{2}\) for \(i \in \{1, 2\}\), and the final inequality uses the simple fact that \(t(1 - t) \geq 0.5 \cdot \min\{t, (1 - t)\}\) when \(t \in [0, 1]\).

We collate the inequalities from \textbf{Case 1} and \textbf{2} to get the following inequalities
\begin{description}
    \item [when \(\mu > 0\) as in \cref{lem:conductance-lower-strong-RelConv}] \begin{equation*}
            \int_{A_{1}} \calT_{x}(A_{2})\targetdens(x) dx \geq \min\left\{\frac{3}{16}, \frac{3}{64} \cdot C_{\alpha} \cdot \sqrt{\mu} \cdot \Delta \right\} \cdot \min\{\targetdist(A_{1}), \targetdist(A_{2})\}~,
        \end{equation*}
    \item [when \(\mu = 0\) as in \cref{lem:conductance-lower-weak}] \begin{equation*}
        \int_{A_{1}} \calT_{x}(A_{2})\targetdens(x) dx \geq \min\left\{\frac{3}{16}, \frac{3}{128} \cdot \frac{\Delta}{\sqrt{\nu}}\right\} \cdot \min\{\targetdist(A_{1}), \targetdist(A_{2})\}~,
    \end{equation*}
\end{description}
which concludes the proofs.
\end{proof}

\subsubsection{Proof of \cref{lem:one-step-overlap}}

In this subsection, we give the proof of both statements of the lemma.
Specifically,
\begin{description}
    \item [{\hyperlink{prf:part-1-lem:one-step-overlap}{Part 1}}] provides the proof for the first statement about \(\TVdist(\calP_{x}, \calP_{y})\), and
    \item [{\hyperlink{prf:part-2-lem:one-step-overlap}{Part 2}}] provides the proof for the second statement about \(\TVdist(\calT_{x}, \calP_{x})\).
\end{description}

Before giving the proofs, we state some key properties of \(\mirrorfunc{}\) due to its self-concordance.
\begin{proposition}[{\citet[\S 5.1.4]{nesterov2018lectures}}]
\label{prop:self-concordance-props}
Let \(\mirrorfunc{} : \primalspace \to \bbR \cup \{\infty\}\) be a self-concordant function with parameter \(\alpha\).
\begin{enumerate}
    \item For any \(x, y \in \interior{\primalspace}\) such that \(\|x - y\|_{\hessmirror{y}} < \frac{1}{\alpha}\),
    \begin{equation*}
        (1 - \alpha \cdot \|x - y\|_{\hessmirror{y}})^{2} \cdot \hessmirror{y} \preceq \hessmirror{x} \preceq \frac{1}{(1 - \alpha \cdot \|x - y\|_{\hessmirror{y}})^{2}} \cdot \hessmirror{y}~.
    \end{equation*}
    \item For any \(x, y \in \interior{\primalspace}\) such that \(\|x - y\|_{\hessmirror{y}} < \nicefrac{1}{\alpha}\), the matrix
    \begin{equation*}
        G_{y} = \int_{0}^{1} \hessmirror{y + \tau(x - y)} d\tau
    \end{equation*}
    satisfies
    \begin{equation*}
        G_{y} \preceq \frac{1}{1 - \alpha \cdot \|x - y\|_{\hessmirror{y}}} \cdot \hessmirror{y}~.
    \end{equation*}
    \item \(\mirrorfunc{}^{*}\) is also a self-concordant function with parameter \(\alpha\).
\end{enumerate}
\end{proposition}

\begin{proof}[\hypertarget{prf:part-1-lem:one-step-overlap}{Part 1.}]
By definition of \(\calP_{x}\), \(\calP_{x} = (\dualtoprim{})_{\#}\widetilde{\calP}_{x}\), where \(\widetilde{\calP}_{x} = \calN(\primtodual{x} - h \cdot \gradpotential{x}, 2h \cdot \hessmirror{x})\).
Starting with Pinsker's inequality, we have
\begin{align*}
    \TVdist(\calP_{x}, \calP_{y}) &\leq \sqrt{\frac{1}{2}\KLdist(\calP_{x}, \calP_{y})} \\
    &= \sqrt{\frac{1}{2}\KLdist((\dualtoprim{})_{\#}\widetilde{\calP}_{x}, (\dualtoprim{})_{\#}\widetilde{\calP}_{y})} \\
    &= \sqrt{\frac{1}{2}\KLdist(\widetilde{\calP}_{x}, \widetilde{\calP}_{y})}~.
\end{align*}
Since \(\mirrorfunc{}\) is of Legendre type, \(\dualtoprim{}\) is a bijective and differentiable map.
This allows us to instantiate \citet[Lemma 15]{vempala2019rapid} to get the final equality.

To further work with \(\KLdist(\widetilde{\calP}_{x}, \widetilde{\calP}_{y})\), we use the identity
\begin{equation*}
    \KLdist(\calN(m_{1}, \Sigma_{1}), \calN(m_{2}, \Sigma_{2})) = \frac{1}{2}\left(\trace(\Sigma_{2}^{-1}\Sigma_{1} - I) + \log\frac{\det \Sigma_{2}}{\det \Sigma_{1}} + \|m_{2} - m_{1}\|^{2}_{\Sigma_{2}^{-1}}\right)~.
\end{equation*}
In our case,
\begin{equation*}
    \Sigma_{1} = 2h\hessmirror{x} ~,~ \Sigma_{2} = 2h\hessmirror{y} ~,~ m_{1} = \primtodual{x} - h\gradpotential{x} ~,~ m_{2} = \primtodual{y} - h\gradpotential{y}~.
\end{equation*}
This yields a closed form expression \(\KLdist(\widetilde{\calP}_{x}, \widetilde{\calP}_{y}) = \frac{1}{2}\left(T_{1}^{P} + T_{2}^{P}\right)\)
where
\begin{align*}
    T_{1}^{P} &= \trace(\hessmirror{y}^{-1}\hessmirror{x} - I) - \log\det \hessmirror{y}^{-1}\hessmirror{x}~, \\
    T_{2}^{P} &= \frac{1}{2h}\|(\primtodual{y} - \primtodual{x}) - h(\gradpotential{y} - \gradpotential{x})\|^{2}_{\hessmirror{y}^{-1}}~.
\end{align*}

The rest of the proof is dedicated to showing that
\begin{equation*}
    T_{1}^{P} \leq \frac{h \cdot d}{20} ~;\qquad T_{2}^{P} \leq \frac{1}{80} + \frac{9}{2} \cdot h \cdot \beta~,
\end{equation*}
under assumptions made in the statement of the lemma.

For convenience, we denote \(\|x - y\|_{\hessmirror{y}}\) by \(r_{y}\).
From the statement of the first part of the lemma, \(r_{y} \leq \frac{\sqrt{h}}{10 \cdot \bar{\alpha}} \leq \frac{\sqrt{h}}{10 \cdot \alpha}\), since \(\bar{\alpha} = \max\{1, \alpha\}\).

\paragraph{Bounding \(T_{1}^{P}\)}
Owing to the cyclic property of trace and the product property of determinants, we have
\begin{align*}
    \trace(\hessmirror{y}^{-1}\hessmirror{x}) &= \trace(\hessmirror{y}^{-\nicefrac{1}{2}}\hessmirror{x}\hessmirror{y}^{-\nicefrac{1}{2}}) \\
    \det\hessmirror{y}^{-1}\hessmirror{x} &= \det \hessmirror{y}^{-\nicefrac{1}{2}}\hessmirror{x}\hessmirror{y}^{-\nicefrac{1}{2}}.
\end{align*}
Let \(M = \nabla^{2}\phi(y)^{-\nicefrac{1}{2}}\nabla^{2}\phi(x)\nabla^{2}\phi(y)^{-\nicefrac{1}{2}}\) for convenience, and let \(\{\lambda_{i}(M)\}_{i=1}^{d}\) be its eigenvalues.
\begin{equation*}
    T_{1}^{P} = \trace(M - I) - \log \det M = \sum_{i = 1}^{d} \{\lambda_{i}(M) - 1 - \log \lambda_{i}(M)\}~.
\end{equation*}
\cref{lem:bound_x-1-logx} states that \(\lambda_{i}(M) - 1 - \log \lambda_{i}(M) \leq \frac{(\lambda_{i}(M) - 1)^{2}}{\lambda_{i}(M)}\).
Therefore, we have the upper bound
\begin{equation*}
    T_{1}^{P} \leq \sum_{i = 1}^{d}\frac{(\lambda_{i}(M) - 1)^{2}}{\lambda_{i}(M)}~.
\end{equation*}
From \cref{prop:self-concordance-props}(1), each \(\lambda_{i}(M)\) is bounded as
\begin{equation*}
    \lambda_{i}(M) \in \left[\frac{1}{(1 - \alpha \cdot r_{y})^{2}}~,~ (1 - \alpha \cdot r_{y})^{2}\right]~.
\end{equation*}
The function \(t \mapsto \frac{(t- 1)^{2}}{t}\) is strictly convex for \(t > 0\).
This is because its second derivative is \(t \mapsto \frac{1}{t^{3}}\).
This implies that when \(t\) is restricted to a closed interval \([a, b]\), the maximum of \(\frac{(t - 1)^{2}}{t}\) is attained at the end points.
\begin{equation*}
    \max_{t \in [a, b]} \frac{(t - 1)^{2}}{t} = \max\left\{\frac{(a - 1)^{2}}{a}, \frac{(b - 1)^{2}}{b}\right\}~.
\end{equation*}
When \(b = \frac{1}{a}\), the maximum on the right hand side is \(\frac{(a - 1)^{2}}{a}\).
We use this fact with bounds on \(t = \lambda_{i}(M)\), and thus have for all \(i \in [d]\) that
\begin{equation*}
    \frac{(\lambda_{i}(M) - 1)^{2}}{\lambda_{i}(M)} \leq \frac{((1 - \alpha \cdot r_{y})^{2} - 1)^{2}}{(1 - \alpha \cdot r_{y})^{2}} = \frac{\alpha^{2} \cdot r_{y}^{2} \cdot (2 - \alpha \cdot r_{y})^{2}}{(1 - \alpha \cdot r_{y})^{2}} \leq \frac{4 \cdot \alpha^{2} \cdot r_{y}^{2}}{(1 - \alpha \cdot r_{y})^{2}}~.
\end{equation*}
and consequently,
\begin{equation*}
    T_{1}^{P} \leq d \cdot \frac{4 \cdot \alpha^{2} \cdot r_{y}^{2}}{(1 - \alpha \cdot r_{y})^{2}}~.
\end{equation*}
The function \(t \mapsto \frac{4t^{2}}{(1 - t)^{2}}\) is an increasing function since it is a product of two increasing functions \(t \mapsto 4t^{2}\) and \(t \mapsto \frac{1}{(1 - t)^{2}}\).
As stated earlier, \(r_{y} \leq \frac{\sqrt{h}}{10 \cdot \alpha}\), which implies
\begin{equation*}
    T_{1}^{P} \leq d \cdot \frac{4 \cdot \frac{h}{100}}{(1 - \frac{\sqrt{h}}{10})^{2}} \leq \frac{h \cdot d}{20}~.
\end{equation*}
The final inequality uses the fact that \(h \leq 1\) which implies \(\left(1 - \frac{\sqrt{h}}{10}\right)^{-2} \leq \frac{100}{81} \leq \frac{5}{4}\).

\paragraph{Bounding \(T_{2}^{P}\)}
Using the fact that \(\|a + b\|_{M}^{2} \leq 2\|a\|^{2}_{M} + 2\|b\|_{M}^{2}\) successively,
\begin{align*}
    T_{2}^{P} &= \frac{1}{2h}\|(\primtodual{y} - \primtodual{x}) - h(\gradpotential{y} - \gradpotential{x})\|^{2}_{\hessmirror{y}^{-1}} \\
    &\leq \frac{1}{2h}\left(2 \cdot \|\primtodual{y} - \primtodual{x}\|_{\hessmirror{y}^{-1}}^{2} + 2h^{2} \cdot \|\gradpotential{y} - \gradpotential{x}\|_{\hessmirror{y}^{-1}}^{2}\right) \\
    &\leq \frac{1}{h} \cdot \|\primtodual{y} - \primtodual{x}\|_{\hessmirror{y}^{-1}}^{2} + 2h \cdot \|\gradpotential{y}\|_{\hessmirror{y}^{-1}}^{2} + 2h \cdot \|\gradpotential{x}\|_{\hessmirror{y}^{-1}}^{2}~.
\end{align*}
Note that
\begin{equation*}
    \primtodual{x} - \primtodual{y} = \int_{0}^{1}\frac{d}{d\tau}\primtodual{y + \tau(x - y)} d\tau = \int_{0}^{1}\hessmirror{y + \tau(x - y)}(x - y) d\tau = G_{y} (x - y).
\end{equation*}
From \cref{prop:self-concordance-props}(2), and using the fact that \(r_{y} < \nicefrac{1}{\alpha}\), we have \(G_{y} \preceq \frac{1}{(1 - \alpha \cdot r_{y})} \cdot \hessmirror{y}\).
By virtue of \(\mirrorfunc{}\) being of Legendre type, \(G_{y} \succ 0\).
As a result,
\begin{equation}
\label{eq:interpolate-hessian-loewner-equiv}
    G_{y}^{-1} \succeq (1 - \alpha \cdot r_{y}) \cdot \hessmirror{y}^{-1} \Leftrightarrow G_{y} \preceq \frac{1}{1 - \alpha \cdot r_{y}} \cdot \hessmirror{y}~,
\end{equation}
where the equivalence is due to the Loewner ordering.
Let \(v = G_{y}^{\nicefrac{1}{2}}(x - y)\).
\begin{align*}
    \|\nabla \phi(x) - \nabla \phi(y)\|_{\nabla^{2}\phi(y)^{-1}}^{2} &= \|G_{y}(x - y)\|^{2}_{\hessmirror{y}^{-1}} \\
    &= \langle v, G_{y}^{\nicefrac{1}{2}}\nabla^{2}\phi(y)^{-1}G_{y}^{\nicefrac{1}{2}}v\rangle \\
    &\leq \frac{\langle v, v\rangle}{1 - \alpha \cdot r_{y}} \\
    &= \frac{\langle x - y, G_{y} (x - y)\rangle}{1 - \alpha \cdot r_{y}} \\
    &\leq \frac{r_{y}^{2}}{(1 - \alpha \cdot r_{y})^{2}}~.
\end{align*}
The inequalities above are due to \cref{eq:interpolate-hessian-loewner-equiv}.
Next, from \cref{prop:self-concordance-props}(1), we have \(\hessmirror{x} \succeq \frac{1}{(1 - \alpha \cdot r_{y})^{2}} \cdot \hessmirror{y}\), and using the fact that \(\mirrorfunc{}\) is strictly convex yields
\begin{equation*}
    \|\gradpotential{x}\|_{\hessmirror{y}^{-1}}^{2} \leq \frac{1}{(1 - \alpha \cdot r_{y})^{2}} \cdot \|\gradpotential{x}\|_{\hessmirror{x}^{-1}}^{2}~.
\end{equation*}
This gives a bound for \(T_{2}^{P}\) in terms of \(r_{y}\) as
\begin{align*}
    T_{2}^{P} &\leq \frac{r_{y}^{2}}{h \cdot (1 - \alpha \cdot r_{y})^{2}} + 2h \cdot \|\gradpotential{y}\|_{\hessmirror{y}^{-1}}^{2} + \frac{2h \cdot \|\gradpotential{x}\|_{\hessmirror{x}^{-1}}^{2}}{(1 - \alpha \cdot r_{y})^{2}} \\
    &\leq \frac{r_{y}^{2}}{h \cdot (1 - \alpha \cdot r_{y})^{2}} + 2h \cdot \beta^{2} + \frac{2h \cdot \beta^{2}}{(1 - \alpha \cdot r_{y})^{2}}~,
\end{align*}
where in the final inequality, we have used the fact that \(\potential{}\) is a \(\beta\)-relatively Lipschitz with respect to \(\mirrorfunc{}\).
Finally, since \(\alpha \cdot r_{y} \leq \frac{\sqrt{h}}{10} \leq \frac{1}{10}\), \((1 - \alpha \cdot r_{y})^{-2} \leq \frac{5}{4}\) as noted earlier.
This gives the upper bound
\begin{equation*}
    T_{2}^{P} \leq \frac{5}{4} \cdot \frac{r_{y}^{2}}{h} + \frac{9}{2} \cdot h \cdot \beta^{2}~.
\end{equation*}
From the statement of the lemma, \(r_{y} \leq \frac{\sqrt{h}}{10 \cdot \max\{1, \alpha\}}\) which also implies that \(r_{y} \leq \frac{\sqrt{h}}{10}\).
Finally, substituting this bound over \(r_{y}\), we get
\begin{equation*}
    T_{2}^{P} \leq \frac{5}{4} \cdot \frac{1}{h} \cdot \frac{h}{100} + \frac{9}{2} \cdot h \cdot \beta \leq \frac{1}{80} + \frac{9}{2} \cdot h \cdot \beta^{2}~.
\end{equation*}

Using the bounds derived for \(T_{1}^{P}\) and \(T_{2}^{P}\), we can complete the proof.
\begin{align*}
    \KLdist(\widetilde{\calP}_{x}, \widetilde{\calP}_{y}) &\leq \frac{1}{2}\left(\frac{h \cdot d}{20} + \frac{1}{80} + \frac{9}{2} \cdot h \cdot \beta^{2}\right)~.
\end{align*}
Finally, applying Pinsker's inequality as stated in the beginning of the proof yields
\begin{equation*}
    \TVdist(\calP_{x}, \calP_{y}) \leq \frac{1}{2}\sqrt{\frac{h \cdot d}{20} + \frac{1}{80} + \frac{9}{2} \cdot h \cdot \beta^{2}}~.
\end{equation*}
\end{proof}

\begin{proof}[\hypertarget{prf:part-2-lem:one-step-overlap}{Part 2.}]

For \(\varepsilon \in (0, 1)\), define the following quantities.
\begin{equation}
\label{eq:eps-quantities}
    \sfN_{\varepsilon} = 1 + 2\sqrt{\log\left(\frac{8}{\varepsilon}\right)} + 2\log\left(\frac{8}{\varepsilon}\right); \qquad \sfI_{\varepsilon} = \sqrt{2\log\left(\frac{8}{\varepsilon}\right)}~.
\end{equation}

The explicit forms of the \(\calC_{1}(\varepsilon), \ldots, \calC_{4}(\varepsilon)\) in \cref{eq:h-max-precise} are given below.
\begin{equation}
\label{eq:h-max-constants-precise}
    \left.
    \begin{aligned}
        \calC_{1}(\varepsilon) = \left(\left(\frac{\varepsilon}{24}\right)^{\nicefrac{2}{3}} \frac{1}{6 \cdot \sfN_{\varepsilon}}\right)^{3} &~,~ \calC_{2}(\varepsilon) = \frac{\varepsilon}{16} \cdot \frac{1}{6 \cdot \sfN_{\varepsilon}} \\
        \calC_{3}(\varepsilon) = \left(\left(\frac{\varepsilon}{72}\right)^{2} \cdot \frac{1}{6\sqrt{2} \cdot \sfI_{\varepsilon}}\right)^{\nicefrac{2}{3}} &~,~ \calC_{4}(\varepsilon) = \frac{\varepsilon}{16} \cdot \frac{1}{6\sqrt{2} \cdot \sfI_{\varepsilon}}
    \end{aligned}
    \quad \right\}
\end{equation}
Recall that the transition kernel has an atom i.e., \(\calT_{x}(\{x\}) \neq 0\).
The explicit form of this is given by
\begin{equation*}
    \calT_{x}(\{x\}) = 1 - \int_{z \in \primalspace} \alpha_{\bfP}(z; x) p_{x}(z) dz = 1 - \int_{z \in \primalspace} \min\left\{1, \frac{\targetdens(z) p_{z}(x)}{\targetdens(x) p_{x}(z)}\right\} p_{x}(z) dz~.
\end{equation*}

Using this, we have the expression for \(\TVdist(\calT_{x}, \calP_{x})\) as (c.f. \citet[Eq. 46]{dwivedi2018log})
\begin{equation*}
    \TVdist(\calT_{x}, \calP_{x}) = 1 - \bbE_{z \sim \calP_{x}}\left[\min\left\{1, \frac{\targetdens(z) p_{z}(x)}{\targetdens(x) p_{x}(z)}\right\} \right] ~.
\end{equation*}

Our goal is to bound this quantity from above, which equivalently implies bounding the expectation on the right hand side from below.
This can be achieved using Markov's inequality; for any \(t > 0\),
\begin{equation*}
    \bbE_{z \sim \calP_{x}}\left[\min\left\{1, \frac{\targetdens(z) p_{z}(x)}{\targetdens(x) p_{x}(z)}\right\}\right] \geq t ~\bbP_{z \sim \calP_{x}}\left[\min\left\{1, \frac{\targetdens(z) p_{z}(x)}{\targetdens(x) p_{x}(z)}\right\} \geq t\right]~.
\end{equation*}
If \(t \leq 1\), then
\begin{equation*}
    \bbP_{z \sim \calP_{x}}\left[\min\left\{1, \frac{\targetdens(z) p_{z}(x)}{\targetdens(x) p_{x}(z)}\right\} \geq t\right] = \bbP_{z \sim \calP_{x}}\left[\frac{\targetdens(z) p_{z}(x)}{\targetdens(x) p_{x}(z)} \geq t\right].
\end{equation*}

Due to \cref{eq:proposal-density}, we can write the explicit expression for this ratio
\begin{multline*}
    \frac{\targetdens(z) \cdot p_{z}(x)}{\targetdens(x) \cdot p_{x}(z)} = \exp\left(\potential{x} - \potential{z} + \frac{3}{2}\left\{\log \det \hessmirror{x} - \log \det \hessmirror{z}\right\}\right. \\
    + \frac{1}{4h}\|\primtodual{z} - \primtodual{x} + h \cdot \gradpotential{x}\|_{\hessmirror{x}^{-1}}^{2}\\
    \left.-\frac{1}{4h} \|\primtodual{x} - \primtodual{z} + h \cdot \gradpotential{z}\|_{\hessmirror{z}^{-1}}^{2}\right).
\end{multline*}
Concisely, \(\frac{\pi(z) \cdot p_{z}(x)}{\pi(x) \cdot p_{x}(z)} = \exp(\calA(x, z))\), where \(\calA(x, z) = T_{1}^{A} + T_{2}^{A} + T_{3}^{A} + T_{4}^{A} + T_{5}^{A}\) and each of these terms are
\begin{subequations}
\label{eq:TiA-group}
\begin{align}
    T_{1}^{A} &:= \frac{\|\primtodual{z} - \primtodual{x}\|_{\hessmirror{x}^{-1}}^{2}}{4h} - \frac{\|\primtodual{x} - \primtodual{z}\|_{\hessmirror{z}^{-1}}^{2}}{4h} \\
    T_{2}^{A} &:= \frac{3}{2}\log \det \hessmirror{z}^{-1}\hessmirror{x} \\
    T_{3}^{A} &:= \frac{1}{2}\left(f(x) - f(z) - \langle \gradpotential{z}, \primtodual{x} - \primtodual{z}\rangle_{\hessmirror{z}^{-1}}\right) \\
    T_{4}^{A} &:= \frac{1}{2}\left(f(x) - f(z) - \langle \gradpotential{x}, \primtodual{x} - \primtodual{z}\rangle_{\hessmirror{x}^{-1}}\right) \\
    T_{5}^{A} &:= \frac{h^{2}}{4h}\left(\|\gradpotential{x}\|_{\hessmirror{x}^{-1}}^{2} - \|\gradpotential{z}\|_{\hessmirror{z}^{-1}}^{2}\right).
\end{align}
\end{subequations}
The self-concordant nature of \(\mirrorfunc{}\) enables using \cref{prop:self-concordance-props} to control some of these quantities above, but only when \(\primtodual{x}\) and \(\primtodual{z}\) are close in the local norm.
Therefore, we condition on an event \(\frakE\) which implies that \(\|\primtodual{x} - \primtodual{z}\|_{\hessmirror{x}^{-1}} \leq \frac{3}{10 \cdot \alpha}\).
To be specific,
\begin{align*}
    \bbP_{z \sim \calP_{x}}\left[\frac{\targetdens(z) p_{z}(x)}{\targetdens(x) p_{x}(z)} \geq t\right] &\geq \bbP_{z \sim \calP_{x}}\left[\frac{\targetdens(z) p_{z}(x)}{\targetdens(x) p_{x}(z)} \geq t \wedge \frakE\right] \\
    &= \bbP_{z \sim \calP_{x}}\left[\left.\frac{\targetdens(z) p_{z}(x)}{\targetdens(x) p_{x}(z)} \geq t ~\right|~ \frakE\right] \cdot \bbP_{z \sim \calP_{x}}(\frakE)~.
\end{align*}

\setcounter{secnumdepth}{4}
\crefalias{paragraph}{part}

The remainder of the proof consists of three parts.
\begin{itemize}
\item In \cref{prf:lem:dist-proposal-transition:part1}, we identify an event \(\frakE\) satisfying our requirements, and show that it occurs with high probability.
We show that for the choice of step size \(h\) assumed, \(\frakE\) also implies that \(\|\primtodual{x} - \primtodual{z}\|_{\hessmirror{x}^{-1}} \leq \frac{3}{10 \cdot \alpha}\).
\item Next, in \cref{prf:lem:dist-proposal-transition:part2}, we condition on event \(\frakE\) and give lower bounds for each of the \(T_{i}^{A}\) quantities by leveraging the implication in the previous part.
\item Finally, in \cref{prf:lem:dist-proposal-transition:part3}, we show that for the choice of \(h\) in the statement of the lemma there exists \(t \leq 1\) such that the conditional probability is \(1\), thus reducing the lower bound to \(t \cdot \bbE_{z \sim \calP_{x}}[\frakE]\).
\end{itemize}

\paragraph{Identifying an event \(\frakE\)}\label{prf:lem:dist-proposal-transition:part1}
Since \(\calP_{x} = (\dualtoprim{})_{\#}\widetilde{\calP}_{x}\),
\begin{equation*}
    z \sim \calP_{x} ~\Leftrightarrow~ \primtodual{z} \sim \widetilde{\calP}_{x} = \calN(\primtodual{x} - h \cdot \gradpotential{x}, 2h\hessmirror{x})~.
\end{equation*}

We use this to show that \(\|\primtodual{x} - \primtodual{z}\|_{\hessmirror{x}^{-1}}\) concentrates well, since given \(\xi \sim \calN(0, I)\), \(\primtodual{z}\) is distributionally equivalent to \(\primtodual{x} - h \cdot \gradpotential{x} + \sqrt{2h} \hessmirror{x}^{\nicefrac{1}{2}}\xi\) when \(z \sim \calP_{x}\).

We begin by expanding the squared norm as
\begin{multline*}
    \|\primtodual{z} - \primtodual{x}\|^{2}_{\hessmirror{x}^{-1}} \\
    = h^{2} \cdot \langle \gradpotential{x}, \hessmirror{x}^{-1}\gradpotential{x}\rangle + 2h \cdot \langle \xi, \xi \rangle + 2\sqrt{2}h\sqrt{h} \cdot \langle \xi, \hessmirror{x}^{-\nicefrac{1}{2}}\gradpotential{x}\rangle.
\end{multline*}

Consider the following probability
\begin{align*}
    &\bbP_{z \sim \calP_{x}}(\|\primtodual{x} - \primtodual{z}\|_{\hessmirror{x}^{-1}}^{2} > h^{2} \cdot \|\gradpotential{x}\|^{2}_{\hessmirror{x}^{-1}} + 2h \cdot d \cdot \sfN_{\varepsilon} + 2\sqrt{2}h\sqrt{h} \cdot \sfI_{\varepsilon})\\
    &\quad= \bbP_{\xi \sim \calN(0, I_{d})}(2h \cdot \|\xi\|^{2} + 2\sqrt{2}h\sqrt{h} \cdot \langle \xi, \hessmirror{x}^{-\nicefrac{1}{2}}\gradpotential{x}\rangle > 2h \cdot d\cdot \sfN_{\varepsilon} + 2\sqrt{2}h\sqrt{h} \cdot \beta \cdot \sfI_{\varepsilon}) \\
    &\quad\leq \bbP_{\xi \sim \calN(0, I_{d})}(\|\xi\|^{2} > d \cdot \sfN_{\varepsilon}) + \bbP_{\xi \sim \calN(0, I_{d})}(\langle \xi, \hessmirror{x}^{-\nicefrac{1}{2}}\gradpotential{x}\rangle > \beta\cdot \sfI_{\varepsilon})
\end{align*}
where the final inequality uses the fact that for random variables \(\sfa, \sfb\),
\begin{gather*}
    (\sfa \leq a) \wedge (\sfb \leq b) \Rightarrow \sfa + \sfb \leq a + b \\
    \bbP(\sfa + \sfb > a + b) \leq \bbP(\sfa > a \vee \sfb > b) \leq \bbP(\sfa > a) + \bbP(\sfb > b)~.
\end{gather*}
Through \(\chi^{2}\) concentration guarantees \citep[Lem. 1]{laurent2000adaptive}, we get
\begin{equation*}
    \bbP(\|\xi\|^{2} > d \cdot \sfN_{\varepsilon}) \leq \frac{\varepsilon}{8}~,
\end{equation*}
where \(\sfN_{\varepsilon}\) was previously defined in \cref{eq:eps-quantities}.
Also, since \(\xi\) is Gaussian, \(\langle \xi, \hessmirror{x}^{-\nicefrac{1}{2}}\gradpotential{x}\rangle\) is a mean zero, sub-Gaussian random variable with variance parameter \(\|\gradpotential{x}\|_{\hessmirror{x}^{-1}}^{2}\).
This implies that for any \(x \in \primalspace\),
\begin{equation*}
    \bbP_{\xi \sim \calN(0, I_{d})}(\langle \xi, \hessmirror{x}^{-\nicefrac{1}{2}}\gradpotential{x}\rangle > t) \leq \exp\left(-\frac{t^{2}}{2\|\gradpotential{x}\|_{\hessmirror{x}^{-1}}^{2}}\right) \leq \exp\left(-\frac{t^{2}}{2\beta^{2}}\right)
\end{equation*}
where the final inequality is due to the fact that \(\potential{}\) is \(\beta\)-relatively Lipschitz with respect to \(\phi\).
Therefore,
\begin{equation*}
    \bbP_{\xi \sim \calN(0, I_{d})}(\langle \xi, \hessmirror{x}^{-\nicefrac{1}{2}}\gradpotential{x}\rangle > \beta \cdot \sfI_{\varepsilon}) \leq \frac{\varepsilon}{8}~,
\end{equation*}
where \(\sfI_{\varepsilon}\) was previously defined in \cref{eq:eps-quantities}.
Therefore, with probability at least \(1 - \frac{\varepsilon}{4}\),
\begin{align*}
    \|\primtodual{x} - \primtodual{z}\|^{2}_{\hessmirror{x}^{-1}} &\leq h^{2} \cdot \|\gradpotential{x}\|^{2}_{\hessmirror{x}^{-1}} + 2h \cdot d \cdot \sfN_{\varepsilon} + 2\sqrt{2}h\sqrt{h} \cdot \beta \cdot \sfI_{\varepsilon} \\
    &\leq h^{2} \cdot \beta^{2} + 2h \cdot d \cdot \sfN_{\varepsilon} + 2\sqrt{2}h\sqrt{h} \cdot \beta \cdot \sfI_{\varepsilon}~.
\end{align*}
We hence pick the event \(\frakE\) to be
\begin{equation*}
    \frakE \defeq \left\{ \|\primtodual{x} - \primtodual{z}\|_{\hessmirror{x}^{-1}} \leq \sqrt{h^{2} \cdot \beta^{2} + 2h \cdot d \cdot \sfN_{\varepsilon} + 2\sqrt{2}h\sqrt{h} \cdot \beta \cdot \sfI_{\varepsilon}} \right\}.
\end{equation*}

Using the simple algebraic fact that \(\sqrt{a + b + c} \leq \sqrt{a} + \sqrt{b} + \sqrt{c}\), and from the choice of \(h\) in the lemma it can be verified that
\begin{equation*}
    h \leq \frac{1}{10 \cdot \beta \cdot \alpha}; \quad h \leq \frac{1}{200 \cdot d \cdot \sfN_{\varepsilon} \cdot \alpha^{2}}; \quad h \leq \frac{1}{50 \cdot \sfI_{\varepsilon}^{\nicefrac{2}{3}} \cdot \alpha^{\nicefrac{4}{3}} \cdot \beta^{\nicefrac{2}{3}}}~,
\end{equation*}
and we are assured that \(\|\primtodual{x} - \primtodual{z}\|_{\hessmirror{x}^{-1}} \leq \frac{3}{10 \cdot \alpha}\) as required.

\paragraph{Lower bound on the \(T_{i}^{A}\) quantities}\label{prf:lem:dist-proposal-transition:part2}

In this part, we condition on the event \(\frakE\) which implies that \(\|\primtodual{z} - \primtodual{x}\|_{\hessmirror{x}^{-1}} \leq \frac{3}{10 \cdot \alpha}\) to obtain lower bounds on \(\calA(x, z)\) defined previously.
We do so by giving lower bounds for each of the \(T_{i}^{A}\) quantities individually defined in \cref{eq:TiA-group}.

In this part of the proof, we will make use of the identity
\begin{align}
\label{eq:legendre-hessian-identity}
    \hessmirrorinv{\primtodual{x}} &= \hessmirror{x}^{-1} &\forall ~x \in \primalspace~.
\end{align}
This can be derived by differentiating the identity \(\dualtoprim{\primtodual{x}} = x\) by the invertibility of \(\primtodual{x}\).

\begin{description}
\item [\(T_{1}^{A}\):]
We use the shorthand notation \(v_{x, z} = \primtodual{x} - \primtodual{z}\) for convenience.
\begin{align*}
    T_{1}^{A} &= \frac{1}{4h}\left(\|v_{x, z}\|_{\hessmirror{x}^{-1}}^{2} - \|v_{x, z}\|_{\hessmirror{z}^{-1}}^{2}\right) \\
    &\overset{(i)}= \frac{1}{4h}\left(\|v_{x, z}\|_{\hessmirrorinv{\primtodual{x}}}^{2} - \|v_{x, z}\|_{\hessmirrorinv{\primtodual{z}}}^{2}\right) \\
    &\overset{(ii)}\geq \frac{1}{4h}\left(\|v_{x, z}\|_{\hessmirrorinv{\primtodual{x}}}^{2} - \frac{\|v_{x, z}\|_{\hessmirrorinv{\primtodual{x}}}^{2}}{(1 - \alpha \cdot \|v_{x, z}\|_{\hessmirrorinv{\primtodual{x}}})^{2}}\right) \\
    &= \frac{1}{h\cdot \alpha^{2}} \cdot \frac{(\alpha \cdot \|v_{x, z}\|_{\hessmirrorinv{\primtodual{x}}})^{2}}{4} \cdot \left(1 - \frac{1}{(1 - \alpha \cdot \|v_{x, z}\|_{\hessmirrorinv{\primtodual{x}}})^{2}}\right) \\
    &\overset{(iii)}\geq -\frac{3\alpha}{2h} \cdot \|v_{x, z}\|_{\hessmirrorinv{\primtodual{x}}}^{3}~.
\end{align*}
Step \((i)\) uses \cref{eq:legendre-hessian-identity}.
Step \((ii)\) uses the fact that \(\mirrorfunc{}^{*}\) is self-concordant with parameter \(\alpha\) from \cref{prop:self-concordance-props}(3), and the result of \cref{prop:self-concordance-props}(1) to change the metric from \(\hessmirrorinv{\primtodual{z}}\) to \(\hessmirrorinv{\primtodual{x}}\) since \(\|v_{x, z}\|_{\hessmirrorinv{\primtodual{x}}} < \nicefrac{1}{\alpha}\).
Step \((iii)\) uses an algebraic lemma \cref{lem:t1-lower}.

\item [\(T_{2}^{A}\):]
We can use \cref{eq:legendre-hessian-identity} and the product property of determinants to rewrite \(T_{2}^{A}\) as follows.
\begin{align*}
    T_{2}^{A} &= \frac{3}{2} \log \det \hessmirror{z}^{-1}\hessmirror{x} \\
    &= \frac{3}{2} \log \det \hessmirrorinv{\primtodual{x}}^{-\nicefrac{1}{2}} ~\hessmirrorinv{\primtodual{z}} ~\hessmirrorinv{\primtodual{x}}^{-\nicefrac{1}{2}} \\
    &= \frac{3}{2}\sum_{i = 1}^{d}\log \lambda_{i}(M^{*})
\end{align*}
where \(M^{*} = \hessmirrorinv{\primtodual{x}}^{-\nicefrac{1}{2}} \hessmirrorinv{\primtodual{z}} \hessmirrorinv{\primtodual{x}}^{-\nicefrac{1}{2}}\), and \(\{\lambda_{i}(M^{*})\}\) is the sequence of its eigenvalues.
These eigenvalues lie in an interval due to \(\mirrorfunc{}^{*}\) being self-concordant, and the associated Hessian ordering as given in \cref{prop:self-concordance-props}(1,3) courtesy of \(\|\primtodual{x} - \primtodual{z}\|_{\hessmirrorinv{\primtodual{x}}} < \nicefrac{1}{\alpha}\) as guaranteed by \(\frakE\).
This yields
\begin{equation*}
    \lambda_{i}(M^{*}) \geq (1 - \alpha \cdot \|\primtodual{x} - \primtodual{z}\|_{\hessmirrorinv{\primtodual{x}}})^{2}~,
\end{equation*}
and we use this to obtain the lower bound
\begin{align*}
    T_{2}^{A} &\geq 3d \cdot \log (1 - \alpha \cdot \|\primtodual{x} - \primtodual{z}\|_{\hessmirrorinv{\primtodual{x}}}) \\
    &\geq -\frac{9d}{2} \cdot \alpha \cdot \|\primtodual{x} - \primtodual{z}\|_{\hessmirrorinv{\primtodual{x}}}~.
\end{align*}
The final inequality is due to \cref{lem:t2-lower}.

\item [\(T_{3}^{A} + T_{4}^{A}\): ]

Let \(\psi = f \circ \dualtoprim{} : \bbR^{d} \to \bbR\).
Then,
\begin{align*}
    D_{\psi}(\bar{x}; \bar{z}) &= \psi(\bar{x}) - \psi(\bar{z}) - \langle \nabla \psi(\bar{z}), \bar{x} - \bar{z}\rangle \\
    &= \psi(\bar{x}) - \psi(\bar{z}) - \langle \hessmirrorinv{\bar{z}}\gradpotential{\dualtoprim{\bar{z}}}, \bar{x} - \bar{z}\rangle~.
\end{align*}

The Bregman commutator \(\zeta_{\psi} : \bbR^{d} \times \bbR^{d} \to \bbR\) of \(\psi\) \citep{wibisono2022alternating} is defined as
\begin{equation*}
    \zeta_{\psi}(\bar{x}, \bar{z}) = \frac{1}{2}(D_{\psi}(\bar{x}; \bar{z}) - D_{\psi}(\bar{z}; \bar{x}))~.
\end{equation*}

Substituting \(\bar{x} = \primtodual{x}\) and \(\bar{z} = \primtodual{z}\) gives
\begin{align*}
    D_{\psi}(\primtodual{x}; \primtodual{z}) &= f(x) - f(z) - \langle \gradpotential{z}, \primtodual{x} - \primtodual{z}\rangle_{\hessmirror{z}^{-1}} = T_{3}^{A}\\
    D_{\psi}(\primtodual{z}; \primtodual{x}) &= f(z) - f(x) - \langle \gradpotential{x}, \primtodual{z} - \primtodual{x}\rangle_{\hessmirror{x}^{-1}} = -T_{4}^{A} ~.
\end{align*}

Consequently, the sum of \(T_{3}^{A}\) and \(T_{4}^{A}\) is
\begin{equation*}
    T_{3}^{A} + T_{4}^{A} = \frac{1}{2}\left(D_{\psi}(\primtodual{x}; \primtodual{z}) - D_{\psi}(\primtodual{z}; \primtodual{x})\right) = \zeta_{\psi}(\primtodual{x}, \primtodual{z})~.
\end{equation*}

For convenience, we use the shorthand notation \(\bar{x}\) and \(\bar{z}\) for \(\primtodual{x}\) and \(\primtodual{z}\) respectively.
Define \(p_{t} = \bar{z} + t(\bar{x} - \bar{z})\) for \(t \in [0, 1]\).
We work with the following identity from \citet[Eq. 15]{wibisono2022alternating}
\begin{equation*}
    \zeta_{\psi}(\bar{x}, \bar{z}) = \frac{1}{2}\int_{0}^{1}(1 - 2t) \nabla^{2}\psi(p_{t})[\bar{x} - \bar{z}, \bar{x} - \bar{z}] dt~.
\end{equation*}
The Hessian of \(\psi\) is
\begin{align*}
    \nabla^{2}\psi(p_{t}) &= \nabla^{3}\phi^{*}(p_{t})[\gradpotential{\dualtoprim{p_{t}}}] + \hessmirrorinv{p_{t}} \nabla^{2}\potential{\dualtoprim{p_{t}}} \hessmirrorinv{p_{t}} \\
    &= \nabla^{3}\phi^{*}(p_{t})[\gradpotential{\dualtoprim{p_{t}}}] + \hessmirror{\dualtoprim{p_{t}}}^{-1} \nabla^{2}\potential{\dualtoprim{p_{t}}} \hessmirror{\dualtoprim{p_{t}}}^{-1}
\end{align*}
Consequently, 
\begin{multline*}
    |\zeta_{\psi}(\bar{x}, \bar{z})| \leq \frac{1}{2}\underbrace{\int_{0}^{1}|1 - 2t| \cdot |\nabla^{3}\phi^{*}(p_{t})[\bar{x} - \bar{z}, \bar{x} - \bar{z}, \gradpotential{\dualtoprim{p_{t}}}]| dt}_{I_{1}} \\+ \frac{1}{2}\underbrace{\int_{0}^{1}|1 - 2t| \cdot |\langle \bar{x} - \bar{z}, \hessmirror{\dualtoprim{p_{t}}}^{-1}\nabla^{2}\potential{\dualtoprim{p_{t}}}\hessmirror{\dualtoprim{p_{t}}}^{-1}[\bar{x} - \bar{z}]\rangle| dt}_{I_{2}}~.
\end{multline*}

We can bound \(I_{1}\) and \(I_{2}\) using properties of \(\mirrorfunc{}\), \(\potential{}\) and the fact that \(\bar{x} = \primtodual{x}, \bar{z} = \primtodual{z}\).
From \cref{prop:self-concordance-props}(3), we know that \(\mirrorfunc{}^{*}\) is self-concordant with parameter \(\alpha\).
This implies
\begin{align*}
    I_{1} &\leq \int_{0}^{1} |1 - 2t| \left(2 \alpha \cdot \|\bar{x} - \bar{z}\|_{\hessmirrorinv{p_{t}}}^{2} \cdot \|\gradpotential{\dualtoprim{p_{t}}}\|_{\hessmirrorinv{p_{t}}}\right) dt \\
    &\leq \int_{0}^{1} |1 - 2t| \left(2\alpha \cdot \|\bar{x} - \bar{z}\|_{\hessmirrorinv{p_{t}}}^{2} \cdot \beta\right) dt
\end{align*}
The last inequality uses the fact that \(\potential{}\) is \(\beta\)-relatively Lipschitz with respect to \(\mirrorfunc{}\) and by writing \(\hessmirrorinv{p_{t}} = \hessmirror{\dualtoprim{p_{t}}}^{-1}\) from \cref{eq:legendre-hessian-identity}.

Since \(\potential{}\) is \(\lambda\)-relatively smooth with respect to \(\mirrorfunc{}\), we have for any \(w \in \primalspace\),
\begin{equation*}
    \nabla^{2}\potential{w} \preceq \lambda \cdot \hessmirror{w} \Leftrightarrow \hessmirror{w}^{-\nicefrac{1}{2}}\nabla^{2}f(w)\hessmirror{w}^{-\nicefrac{1}{2}} \preceq \lambda \cdot I~.
\end{equation*}
With this, we can bound \(I_{2}\) as 
\begin{align*}
    I_{2} &\leq \int_{0}^{1}|1 - 2t| \cdot \lambda \cdot |\langle \hessmirror{\dualtoprim{p_{t}}}^{-\nicefrac{1}{2}}(\bar{x} - \bar{z}), \hessmirror{\dualtoprim{p_{t}}}^{-\nicefrac{1}{2}}(\bar{x} - \bar{z})\rangle| dt \\
    &= \int_{0}^{1} |1 - 2t| \cdot \lambda \cdot \|\bar{x} - \bar{z}\|_{\hessmirrorinv{p_{t}}}^{2} dt~.
\end{align*}
Collecting the bounds,
\begin{align*}
    |\zeta_{\psi}(\bar{x}, \bar{z})| &\leq \left(\alpha \cdot \beta + \frac{\lambda}{2}\right) \cdot \int_{0}^{1} |1 - 2t| \cdot \|\bar{x} - \bar{z}\|^{2}_{\hessmirrorinv{p_{t}}} dt \\
    &\overset{(i)}\leq \left(\alpha \cdot \beta + \frac{\lambda}{2}\right) \cdot \int_{0}^{1} |1 - 2t| \cdot \frac{\|\bar{x} - \bar{z}\|^{2}_{\hessmirrorinv{\bar{x}}}}{(1 - \alpha \cdot (1 - t)\|\bar{x} - \bar{z}\|_{\hessmirrorinv{\bar{x}}})^{2}} dt \\
    &\overset{(ii)}\leq \left(\alpha \cdot \beta + \frac{\lambda}{2}\right) \cdot \|\bar{x} - \bar{z}\|^{2}_{\hessmirrorinv{\bar{x}}} \cdot \int_{0}^{1} \frac{100 \cdot |1 - 2t|}{(7 + 3t)^{2}} dt \\
    &\leq \left(\alpha \cdot \beta + \frac{\lambda}{2}\right) \cdot \|\primtodual{x} - \primtodual{z}\|^{2}_{\hessmirrorinv{\primtodual{x}}}~.
\end{align*}
Recall that \(\|\bar{x} - \bar{z}\|_{\hessmirrorinv{\bar{x}}} < \nicefrac{1}{\alpha}\) and \(\|p_{t} - \bar{x}\|_{\hessmirrorinv{\bar{x}}} = (1 - t)\|\bar{z} - \bar{z}\|_{\hessmirrorinv{\bar{x}}} < \nicefrac{1}{\alpha}\) when \(t \in [0, 1]\).
In Step \((i)\), we use this fact to use \cref{prop:self-concordance-props}(1) to get
\begin{equation*}
    \hessmirrorinv{p_{t}} \preceq \frac{1}{(1 - \alpha \cdot (1 - t) \|\bar{x} - \bar{z}\|_{\hessmirrorinv{\bar{x}}})^{2}} \cdot \hessmirrorinv{\bar{x}}
\end{equation*}
Step \((ii)\) uses the fact that \(\|\bar{x} - \bar{z}\|_{\hessmirrorinv{\bar{x}}} \leq \nicefrac{3}{10 \cdot \alpha}\) and that \((1 - x)^{-2}\) is increasing for \(x < 1\).
The final inequality is an upper bound on the integral which can be computed as a closed form expression with value \(\approx 0.73\).
This gives us the final bound
\begin{equation*}
    |\zeta_{\psi}(\primtodual{x}, \primtodual{z})| \leq \gamma \cdot \|\primtodual{x} - \primtodual{z}\|^{2}_{\hessmirror{x}^{-1}}~,
\end{equation*}
where \(\gamma = \frac{\lambda}{2} + \alpha \cdot \beta\), which implies
\begin{equation*}
    T_{3}^{A} + T_{4}^{A} \geq -\gamma \cdot \|\primtodual{x} - \primtodual{z}\|_{\hessmirror{x}^{-1}}^{2}~.
\end{equation*}

\item [\(T_{5}^{A}\): ] Since \(f\) is \(\beta\)-relatively Lipschitz with respect to \(\phi\), this can be simply bounded from below as
\begin{equation*}
    T_{5}^{A} \geq -\frac{h \cdot \beta^{2}}{4}.
\end{equation*}
\end{description}

For convenience, we will use the shorthand notation \(\ell_{x,z}\) to denote \(\|\primtodual{x} - \primtodual{z}\|_{\hessmirror{x}^{-1}}\).
Collating these lower bounds, we get the net lower bound on \(\calA(x, z)\) as
\begin{equation*}
    \calA(x, z) \geq -\frac{3\bar{\alpha}}{2h}\cdot \ell_{x,z}^{3} - \frac{9d \cdot \bar{\alpha}}{2}\cdot \ell_{x, z} - \gamma \cdot\ell_{x,z}^{2} - \frac{h \cdot \beta^{2}}{4}~.
\end{equation*}

\paragraph{Choosing \(t\) in the conditional probability quantity}
\label{prf:lem:dist-proposal-transition:part3}
Under the event \(\frakE\), we can substitute the upper bound on \(\ell_{x,z}\) in the lower bound for \(\calA(x, z)\).

This gives us
\begin{align*}
    \calA(x, z) &\geq -\frac{3\alpha}{2h}\cdot(h^{2} \cdot \beta^{2} + 2h \cdot d \cdot \sfN_{\varepsilon} + 2\sqrt{2}h\sqrt{h} \cdot \beta \cdot \sfI_{\varepsilon})^{\nicefrac{3}{2}} \\
    &\qquad - \frac{9d \cdot \alpha}{2} \cdot (h^{2} \cdot \beta^{2} + 2h \cdot d \cdot \sfN_{\varepsilon} + 2\sqrt{2}h\sqrt{h} \cdot \beta \cdot  \sfI_{\varepsilon})^{\nicefrac{1}{2}} \\
    &\qquad - \gamma \cdot (h^{2} \cdot \beta^{2} + 2h \cdot d \cdot \sfN_{\varepsilon} + 2\sqrt{2}h\sqrt{h} \cdot \beta \cdot \sfI_{\varepsilon}) - \frac{h \cdot \beta^{2}}{4} \\
    &= -\frac{3\alpha \cdot \sqrt{h}}{2}(h \cdot \beta^{2} + 2d \cdot \sfN_{\varepsilon} + 2\sqrt{2h} \cdot \beta \cdot \sfI_{\varepsilon})^{\nicefrac{3}{2}} \\
    &\qquad - \frac{9\sqrt{h} \cdot \alpha}{2}(h \cdot d^{2} \cdot \beta^{2} + 2d^{3} \cdot \sfN_{\varepsilon} + 2\sqrt{2h} \cdot \beta \cdot \sfI_{\varepsilon})^{\nicefrac{1}{2}} \\
    &\qquad -\gamma \cdot h \cdot (h \cdot \beta^{2} + 2d \cdot \sfN_{\varepsilon} + 2\sqrt{2}\sqrt{h} \cdot \beta \cdot \sfI_{\varepsilon}) - \frac{h \cdot \beta^{2}}{4}~.
\end{align*}

Define the following constants.
\begin{gather*}
    C_{\varepsilon}^{(1)} = \left(\frac{\varepsilon}{24}\right)^{\nicefrac{2}{3}}; \quad C_{\varepsilon}^{(2)} = \left(\frac{\varepsilon}{72}\right)^{2}; \quad C_{\varepsilon}^{(3)} = \frac{\varepsilon}{16} \\
    C_{\varepsilon}^{(i,1)} = \frac{C_{\varepsilon}^{(i)}}{3}~;\quad C_{\varepsilon}^{(i, 2)} = \frac{C_{\varepsilon}^{(i)}}{6 \cdot \sfN_{\varepsilon}}~; \quad C_{\varepsilon}^{(i, 3)} = \frac{C_{\varepsilon}^{(i)}}{6\sqrt{2} \cdot \sfI_{\varepsilon}} \qquad \forall ~i \in [3]~.
\end{gather*}
When \(\varepsilon < 1\), note that all of these constants are strictly less than \(1\), since \(\sfN_{\varepsilon} \geq 1\) and \(\sfI_{\varepsilon} \geq 1\).
With this we define some limits on \(h\).
\begin{gather*}
    C^{(1)}_{\max}(\varepsilon) = \min\left\{\left(C_{\varepsilon}^{(1, 1)}\right)^{\nicefrac{3}{4}} \cdot \frac{1}{\beta^{\nicefrac{3}{2}}} \cdot \frac{1}{\sqrt{\alpha}}~,~\left(C_{\varepsilon}^{(1, 2)}\right)^{3}\cdot \frac{1}{\alpha^{2}} \cdot \frac{1}{d^{3}}~,~ \left(C_{\varepsilon}^{(1, 3)}\right)^{\nicefrac{6}{5}} \cdot \frac{1}{\alpha^{\nicefrac{4}{5}}} \cdot \frac{1}{\beta^{\nicefrac{6}{5}}}\right\} ~, \\
    C^{(2)}_{\max}(\varepsilon) = \min\left\{\left(C_{\varepsilon}^{(2, 1)}\right)^{\nicefrac{1}{2}} \cdot \frac{1}{d} \cdot \frac{1}{\alpha} \cdot \frac{1}{\beta} ~,~ \left(C_{\varepsilon}^{(2, 2)}\right) \cdot \frac{1}{\alpha^{2}} \cdot \frac{1}{d^{3}} ~,~ \left(C_{\varepsilon}^{(2, 3)}\right)^{\nicefrac{2}{3}} \cdot \frac{1}{\beta^{\nicefrac{2}{3}}} \cdot \frac{1}{\alpha^{\nicefrac{4}{3}}}\right\} ~, \\
    C^{(3)}_{\max}(\varepsilon) = \min\left\{(C_{\varepsilon}^{(3, 1)})^{\nicefrac{1}{2}} \cdot \frac{1}{\beta} \cdot \frac{1}{\sqrt{\gamma}} ~,~ (C_{\varepsilon}^{(3, 2)}) \cdot \frac{1}{d} \cdot \frac{1}{\gamma} ~,~ (C_{\varepsilon}^{(3, 3)}) \cdot \frac{1}{\beta^{\nicefrac{2}{3}}} \cdot \frac{1}{\gamma^{\nicefrac{2}{3}}}\right\} ~, \\
    C_{\max}^{(4)}(\varepsilon) = \frac{\varepsilon}{4 \cdot \beta^{2}}~.
\end{gather*}
When \(h \leq \min\{C_{\max}^{(1)}(\varepsilon), C_{\max}^{(2)}(\varepsilon), C_{\max}^{(3)}(\varepsilon), C_{\max}^{(4)}(\varepsilon)\}\), it can be verified that
\begin{gather*}
    \frac{3\alpha \cdot \sqrt{h}}{2}(h \cdot \beta^{2} + 2d \cdot \sfN_{\varepsilon} + 2\sqrt{2h} \cdot \beta \cdot \sfI_{\varepsilon})^{\nicefrac{3}{2}} \leq \frac{\varepsilon}{16} \\
    \frac{9\sqrt{h} \cdot \alpha}{2}(h \cdot d^{2} \cdot \beta^{2} + 2d^{3} \cdot \sfN_{\varepsilon} + 2\sqrt{2h} \cdot \beta \cdot \sfI_{\varepsilon})^{\nicefrac{1}{2}} \leq \frac{\varepsilon}{16} \\
    \gamma \cdot h \cdot (h \cdot \beta^{2} + 2d \cdot \sfN_{\varepsilon} + 2\sqrt{2}\sqrt{h} \cdot \beta \cdot \sfI_{\varepsilon}) \leq \frac{\varepsilon}{16} \\
    \frac{h \cdot \beta^{2}}{4} \leq \frac{\varepsilon}{16}
\end{gather*}
and consequently,
\begin{equation*}
    \calA(x, z) \geq -\frac{\varepsilon}{4}~.
\end{equation*}
It can be verified that when \(h \leq b_{\max}(\varepsilon; d, \alpha, \beta, \lambda)\) as defined in \cref{eq:h-max-precise}, with \(\calC_{i}(\varepsilon)\) defined in \cref{eq:h-max-constants-precise} that when \(h \leq b_{\max}(\varepsilon)\), \(h \leq \min\left\{C_{\max}^{(1)}(\varepsilon), C_{\max}^{(2)}(\varepsilon), C_{\max}^{(3)}(\varepsilon), C_{\max}^{(4)}(\varepsilon)\right\}\).
\(b_{\max}(\varepsilon)\) is a condensed version where we have used the facts
\begin{itemize}
\item \(\beta\sqrt{\gamma} \geq \beta^{\nicefrac{3}{2}}\sqrt{\alpha}\),
\item \(d \gamma \geq d \alpha \beta\),
\item if \(\gamma \geq \beta^2\), then \(\beta^{\nicefrac{2}{3}} \gamma^{\nicefrac{2}{3}} \ge \beta \sqrt{\gamma} \ge \beta^2\), and when \(\gamma \leq \beta^{2}\), then \(\beta^{\nicefrac{2}{3}}\gamma^{\nicefrac{2}{3}} \leq \beta \sqrt{\gamma} \leq \beta^{2}\), and
\item if \(\alpha \geq \beta\), then \(\alpha^{\nicefrac{4}{3}} \beta^{\nicefrac{2}{3}} \geq \alpha^{\nicefrac{4}{5}} \beta^{\nicefrac{6}{5}} \geq \beta^2\), and when \(\alpha \leq \beta\), then \(\alpha^{\nicefrac{4}{3}} \beta^{\nicefrac{2}{3}} \leq \alpha^{\nicefrac{4}{5}} \beta^{\nicefrac{6}{5}} \leq \beta^2\).
\end{itemize}
to reduce the cases.

This implies that when \(t = e^{-\nicefrac{\varepsilon}{4}}\) which is less than \(1\),
\begin{equation*}
    \bbP_{z \sim \calP_{x}}\left[\left.\frac{\pi(z) p_{z}(x)}{\pi(x) p_{x}(z)} \geq t ~\right|~ \frakE \right] = 1~.
\end{equation*}

To conclude the proof, from part \ref{prf:lem:dist-proposal-transition:part1}, we have that \(\bbP_{z \sim \calP_{z}}(\frakE) \geq 1 - \frac{\varepsilon}{4}\).
From part \ref{prf:lem:dist-proposal-transition:part3}, setting \(t = e^{-\nicefrac{\varepsilon}{4}}\) ensures that when conditioning on \(\frakE\), \(\bbP_{z\sim \calP_{x}}\left[\exp(\calA(x, z)) \geq t ~|~ \frakE\right] = 1\).
This finally gives
\begin{align*}
    \bbE_{z \sim \calP_{x}}\left[\min\left\{1, \frac{\pi(z) p_{z}(x)}{\pi(x) p_{x}(z)}\right\}\right] &\geq t ~ \bbP_{z \sim \calP_{x}}\left[\left.\frac{\pi(z) p_{z}(x)}{\pi(x) p_{x}(z)} \geq t ~\right|~ \frakE \right] \cdot \bbP_{z \sim \calP_{x}}[\frakE] \\
    &= \exp(-\nicefrac{\varepsilon}{4}) \cdot (1 - \nicefrac{\varepsilon}{4}) \geq (1 - \nicefrac{\varepsilon}{4})^{2} \geq 1 - \nicefrac{\varepsilon}{2}~
\end{align*}
and recalling the form of \(\TVdist(\calT_{x}, \calP_{x})\), we get the bound
\begin{equation*}
    \TVdist(\calT_{x}, \calP_{x}) \leq \frac{\varepsilon}{2}~.
\end{equation*}
\end{proof}

\subsection{Proofs of corollaries in \cref{sec:algo:applications-thm}}
\label{sec:proofs:corollaries}

In this section, we provide the proofs for the corollaries in \cref{sec:algo:applications-thm}.

\subsubsection{Proofs of \cref{corr:mixing-time-polytope-ellipsoids}}
\label{sec:prf:corr:mixing-time-polytope-ellipsoids}

We begin by stating two key facts, each about \(\mathsf{Polytope}(A, b)\) and \(\mathsf{Ellipsoids}(\{(c_{i}, M_{i})\}_{i=1}^{m})\).

\begin{lemma}
\label{lem:log-barrier-polytope-sym-barrier}
Let \(\primalspace = \mathsf{Polytope}(A, b)\) where \(A \in \bbR^{m \times d}, b \in \bbR^{m}\) be a non-empty, bounded polytope, and \(\mirrorfunc{}\) be the log-barrier of \(\primalspace\).
Then, \(\mirrorfunc{}\) is a symmetric barrier with parameter \(m\).
\end{lemma}

\begin{lemma}
\label{lem:log-barrier-ellipsoid-sym-barrier}
Let \(\primalspace = \mathsf{Ellipsoids}(\{M_{i}\}_{i=1}^{m})\) for a sequence of positive definite matrices \(\{M_{i}\}_{i=1}^{m}\), and \(\mirrorfunc{}\) be the log-barrier of \(\primalspace\).
Then, \(\mirrorfunc{}\) is a symmetric barrier with parameter \(2m\).
\end{lemma}

\cref{lem:log-barrier-polytope-sym-barrier} was stated as \citet[Lem. 4.9]{kook2023efficiently}.
To the best of our knowledge, \cref{lem:log-barrier-ellipsoid-sym-barrier} has not been discussed in any prior work, a proof of which is given in \cref{sec:prf:lem:log-barrier-ellipsoid-sym-barrier}.

\begin{proof}
We begin with the following fact, as stated in \citet[Ex. 5.1.1(4)]{nesterov2018lectures}.
Let \(P \in \bbR^{d \times d}\) and \(P \succeq 0\), \(q \in \bbR^{d}, r \in \bbR\).
Then, \(x \mapsto -\log(r + q^{\top}x + x^{\top}Px)\) is a self-concordant function with parameter \(1\).
For \(\mathsf{Polytope}(A, b)\), the log-barrier is the sum of \(m\) such functions where \(r = b_{i}\), \(q = -a_{i}\) and \(P = 0\) for \(i \in [m]\).
For \(\mathsf{Ellipsoids}(\{(c_{i}, M_{i})\}_{i=1}^{m})\), the log-barrier is the sum of \(m\) such functions as well, where \(r = 1 - \|c_{i}\|_{M_{i}}^{2}\), \(q = 2M_{i}c_{i}\), \(P = M_{i}\) for \(i \in [m]\).
By \citet[Thm. 5.1.1]{nesterov2018lectures}, this is a self-concordant function with parameter \(1\) as well.

For uniform sampling, since \(\potential{}\) is a constant function, \(\gradpotential{x} = 0\) for \(x \in \interior{\primalspace}\).
Hence, \(\mu = \lambda = \beta = 0\) in assumptions \ref{assump:rel-convex-smooth} and \ref{assump:rel-lipschitz}.
For assumption \ref{assump:symm-barrier}, \cref{lem:log-barrier-polytope-sym-barrier,lem:log-barrier-ellipsoid-sym-barrier} states that \(\nu = m\) for \(\primalspace = \mathsf{Polytope}(A, b)\), and \(\nu = 2m\) for \(\primalspace = \mathsf{Ellipsoids}(\{(c_{i}, M_{i})\}_{i=1}^{m})\).
Substituting these values for the parameters in the assumptions in the mixing time bound for the (weakly) convex case in \cref{thm:mix-mamla} recovers the statements of the lemmas.
\end{proof}

\subsubsection{Proof of \cref{corr:dirichlet-sampling}}
\begin{proof}
The simplex \(\primalspace\) is a special polytope, and since \(\mirrorfunc{}\) is the log-barrier of \(\primalspace\), the self-concordance parameter is \(1\), as previously discussed in \cref{sec:prf:corr:mixing-time-polytope-ellipsoids}.

We have the following explicit expression for the gradient of \(\potential{}\) and Hessians of \(\phi\) and \(\potential{}\).
\begin{gather*}
    \gradpotential{x} = \left[-\frac{a_{1}}{x_{1}} + \frac{a_{d+1}}{1 - \bm{1}^{\top}x}~, \cdots~, \frac{a_{d}}{x_{d}} + \frac{a_{d+1}}{1 - \bm{1}^{\top}x} \right] \\
    \nabla^{2}\potential{x} = \mathrm{diag}\left(\frac{a_{1}}{x_{1}^{2}}~, \cdots~, \frac{a_{d}}{x_{d}^{2}}\right) + \frac{a_{d+ 1}}{\left(1 - \bm{1}^{\top}x\right)^{2}} \bm{1}_{d\times d}~, \\
    \hessmirror{x} = \mathrm{diag}\left(\frac{1}{x_{1}^{2}}~, \cdots,~, \frac{1}{x_{d}^{2}}\right) + \frac{1}{\left(1 - \bm{1}^{\top}x\right)^{2}} \bm{1}_{d\times d}~.
\end{gather*}
The smallest \(\lambda\) such that \(\lambda \cdot \hessmirror{x} - \nabla^{2}\potential{x} \succeq 0\) is \(\bm{a}_{\max}\).
The largest \(\mu\) such that \(\nabla^{2}\potential{x} - \mu \cdot \hessmirror{x} \succeq 0\) is \(\bm{a}_{\min}\).
Consequently,
\(\potential{}\) is \(\bm{a}_{\max}\)-relatively smooth and \(\bm{a}_{\min}\)-relatively convex with respect to \(\phi\).

From the expressions above, for any \(v \in \bbR^{d}\),
\begin{equation*}
    \langle \gradpotential{x}, v\rangle = -\sum_{i=1}^{d}\frac{a_{i}v_{i}}{x_{i}} + \frac{a_{d+1}}{1 - \bm{1}^{\top}x} \cdot \sum_{i =1}^{d}v_{i}~, \enskip 
    \langle v, \hessmirror{x}v\rangle = \sum_{i=1}^{d}\frac{v_{i}^{2}}{x_{i}^{2}} + \frac{1}{(1 - \bm{1}^{\top}x)^{2}} \cdot \left(\sum_{i=1}^{d}v_{i}\right)^{2}~.
\end{equation*}
From \cref{lem:dirichlet-relative-lipschitz} with \(z_{i} \leftarrow \frac{v_{i}}{x_{i}}\), \(w_{i} \leftarrow v_{i}\), and \(c \leftarrow \frac{1}{1 - \bm{1}^{\top}x}\), we have the inequality for any \(x \in \interior{\primalspace}\)
\begin{equation*}
    \langle \gradpotential{x}, v\rangle^{2} \leq \|\bm{a}\|^{2} \cdot \langle v, \hessmirror{x}v\rangle \Leftrightarrow \langle \gradpotential{x}, \hessmirror{x}^{-1}\gradpotential{x}\rangle \leq \|\bm{a}\|^{2}~,
\end{equation*}
which shows that \(\potential{}\) is \(\|\bm{a}\|\)-relatively Lipschitz with respect to \(\mirrorfunc{}\).

The constant \(\gamma = \frac{\bm{a}_{\max}}{2} + \|\bm{a}\|\) is at most \(\frac{3\|\bm{a}\|}{2}\).
Since \(a_{i} > 0\) for all \(i\), we instantiate the \(\mu > 0\) case of \cref{thm:mix-mamla} to get the mixing time guarantee
\begin{equation}
    \calO\left(\frac{1}{\bm{a}_{\min}} \cdot \max\left\{d^{3},~d \cdot \|\bm{a}\|,~\|\bm{a}\|^{\nicefrac{2}{3}},~\|\bm{a}\|^{\nicefrac{3}{2}},~\|\bm{a}\|^{2}\right\}\cdot \log\left(\frac{\sqrt{M}}{\delta}\right)\right)
\end{equation}
When \(\|\bm{a}\| \geq 1\), \(\|\bm{a}\|^{2}\) dominates \(\|\bm{a}\|^{\nicefrac{3}{2}}\) and \(\|\bm{a}\|^{\nicefrac{2}{3}}\), thus recovering the statement of the lemma.
\end{proof}

\subsection{Proofs of other technical lemmas}
\subsubsection{Proof of \cref{lem:log-barrier-ellipsoid-sym-barrier}}
\label{sec:prf:lem:log-barrier-ellipsoid-sym-barrier}

\begin{proof}
The log-barrier of this set is defined by
\begin{equation*}
    \mirrorfunc{x} = -\sum_{i=1}^{m}\log(1 - \|x - c_{i}\|^{2}_{M_{i}})~.
\end{equation*}

The Hessian of \(\mirrorfunc{}\) is given by
\begin{equation*}
    \hessmirror{x} = \sum_{i=1}^{m} \frac{2M_{i}(1 - \|x - c_{i}\|_{M_{i}}^{2}) + 4M_{i}(x - c_{i})(x - c_{i})^{\top}M_{i}}{(1 - \|x - c_{i}\|_{M_{i}}^{2})^{2}}~.
\end{equation*}
For any \(y \in \calE_{x}^{\mirrorfunc{}}(\sqrt{r})\),
\begin{equation}
\label{eq:dikin-ellipsoid-belonging-def}
    \|y - x\|_{\hessmirror{x}}^{2} \leq r \Leftrightarrow \sum_{i = 1}^{m} \frac{2\|y - x\|_{M_{i}}^{2}(1 - \|x - c_{i}\|_{M_{i}}^{2}) + 4\langle y - x, x - c_{i}\rangle_{M_{i}}^{2}}{(1 - \|x - c_{i}\|_{M_{i}}^{2})^{2}} \leq r~.
\end{equation}
Let \(\calS_{i}\) be the ellipsoid defined by the pair \((c_{i}, M_{i})\).
To show that \(\mirrorfunc{}\) is a symmetric barrier over \(\primalspace = \mathsf{Ellipsoids}(\{(c_{i}, M_{i})\}_{i=1}^{m})\), we need to show that there exists \(r > 0\), such that for all \(x \in \interior{\primalspace}\),
\begin{equation}
\label{eq:containment-chain}
    \calE_{x}^{\mirrorfunc{}}(1) \subseteq \primalspace \cap (2x - \primalspace) \subseteq \calE_{x}^{\mirrorfunc{}}(r)~.
\end{equation}
We begin by giving an equivalent algebraic statement for \(y \in \primalspace \cap (2x - \primalspace)\).
Note that \(\primalspace \cap (2x - \primalspace) = \bigcap_{i = 1}^{m} \left\{\calS_{i} \cap (2x - \calS_{i})\right\}\).

For any \(i \in [m]\), let \(y \in \calS_{i} \cap (2x - \calS_{i})\) which is equivalent to \(y \in \calS_{i} \wedge (2x - y) \in \calS_{i}\).
\begin{align*}
    y \in \calS_{i} &\Leftrightarrow \|y - c_{i}\|_{M_{i}}^{2} \leq 1 \\
    &\Leftrightarrow \|y - x + x - c_{i}\|_{M_{i}}^{2} \leq 1 \\
    &\Leftrightarrow \|y - x\|_{M_{i}}^{2} + \|x - c_{i}\|^{2} + 2\langle y - x, x - c_{i}\rangle_{M_{i}} \leq 1 \\
    &\Leftrightarrow 2\langle y - x, x - c_{i}\rangle_{M_{i}} \leq 1 - \|x - c_{i}\|_{M_{i}}^{2} - \|y - x\|_{M_{i}}^{2}~.\\
    2x - y \in \calS_{i} &\Leftrightarrow \|2x - y - c_{i}\|_{M_{i}}^{2} \leq 1 \\
    &\Leftrightarrow \|x - c_{i} + x - y\|_{M_{i}}^{2} \leq 1 \\
    &\Leftrightarrow \|x - c_{i}\|_{M_{i}}^{2} + \|y - x\|_{M_{i}}^{2} + 2\langle x - c_{i}, x - y\rangle_{M_{i}} \leq 1 \\
    &\Leftrightarrow -2\langle y - x, x - c_{i}\rangle_{M_{i}} \leq 1 - \|x - c_{i}\|_{M_{i}}^{2} - \|y - x\|_{M_{i}}^{2}~.
\end{align*}
Therefore, \(y \in \calS_{i} \cap (2x - \calS_{i})\) if and only if
\begin{equation*}
    2\left|\langle y - x, x - c_{i}\rangle_{M_{i}}\right| \leq 1 - \|x - c_{i}\|_{M_{i}}^{2} - \|y - x\|_{M_{i}}^{2}~.
\end{equation*}
Squaring both sides, and moving terms around we have that \(y \in \primalspace \cap (2x - \primalspace)\) if and only if
\begin{align*}
    \underbrace{\frac{2(1 - \|x - c_{i}\|_{M_{i}}^{2})\|y - x\|_{M_{i}}^{2} + 4\langle y - x, x - c_{i}\rangle_{M_{i}}^{2}}{(1 - \|x - c_{i}\|_{M_{i}}^{2})^{2}}}_{a_{i}} &\leq \underbrace{1 + \frac{\|y - x\|_{M_{i}}^{4}}{(1 - \|x - c_{i}\|_{M_{i}}^{2})^{2}}}_{b_{i}}~ \quad \forall ~i \in [m]~.
\end{align*}
By the definition as stated in \cref{eq:dikin-ellipsoid-belonging-def}, we have following equivalence
\begin{equation*}
    y \in \calE_{x}^{\mirrorfunc{}}(1) \Rightarrow y \in \primalspace \cap (2x - \primalspace) \quad \Leftrightarrow \quad \sum_{i=1}^{m}a_{i} \leq 1 \Rightarrow a_{i} \leq b_{i} \quad \forall~ i \in [m]
\end{equation*}
Suppose there exists \(i \in [m]\) such that \(a_{i} > b_{i}\), then \(a_{i} > 1\) since \(b_{i} \geq 1\).
By definition, \(a_{i} \geq 0\) for all \(i \in [m]\), and this in conjuction with the assumption implies that \(\sum_{i = 1}^{m} a_{i} > 1\).
By contraposition, this is equivalent to stating that \(\sum_{i = 1}^{m} a_{i} \leq 1\) implies \(a_{i} \leq b_{i}\) for all \(i \in [m]\).
This proves that \(\calE_{x}^{\mirrorfunc{}}(1) \subseteq \primalspace \cap (2x - \primalspace)\) as \(y \in \calE_{x}^{\mirrorfunc{}}(1)\) was chosen arbitrarily.

We now show the second statement in the containment chain (\cref{eq:containment-chain}).
If \(y \in \calS_{i} \cap (2x - \calS_{i})\) for some \(i \in [m]\),
\begin{align*}
    \|y - x\|_{M_{i}}^{2} &= \|y - 2x + c_{i} + x - c_{i}\|_{M_{i}}^{2} \\
    &= \|2x - y - c_{i}\|_{M_{i}}^{2} + \|x - c_{i}\|_{M_{i}}^{2} - 2\langle 2x - y - c_{i}, x - c_{i}\rangle_{M_{i}} \\
    &\overset{(a)}\leq 1 + \|x - c_{i}\|_{M_{i}}^{2} - 2\langle x - c_{i} + x - y, x - c_{i}\rangle_{M_{i}} \\
    &= 1 - \|x - c_{i}\|_{M_{i}}^{2} + 2\langle y - x, x - c_{i}\rangle_{M_{i}} \\
    &\overset{(b)}\leq 1 - \|x - c_{i}\|_{M_{i}}^{2} + 1 - \|x - c_{i}\|_{M_{i}}^{2} - \|y - x\|_{M_{i}}^{2} \\
    &\leq 2 - 2\|x - c_{i}\|_{M_{i}}^{2} - \|y - x\|_{M_{i}}^{2}
\end{align*}
Step \((a)\) uses the fact that \(y \in 2x - \calS_{i}\), and step \((b)\) uses the equivalence for \(y \in \calS_{i}\) shown above.
This concludes that if \(y \in \primalspace \cap (2x - \primalspace)\), \(\|y - x\|_{M_{i}}^{2} \leq 1 - \|x\|_{M_{i}}^{2}\) for all \(i \in [m]\).

We can use the above assertion to bound the \(b_{i}\) quantities by \(2\).
In summary, if \(y \in \primalspace \cap (2x - \primalspace)\), then \(a_{i} \leq b_{i} \leq 2\) for all \(i \in [m]\).
Note that this implies \(\sum_{i = 1}^{m} a_{i} \leq 2m\), which is equivalent to stating that if \(y \in \primalspace \cap (2x - \primalspace)\), then \(y \in \calE_{x}^{\mirrorfunc{}}(\sqrt{2m})\).
Due to \(y\) being arbitrary, we have shown that \(\primalspace \cap (2x - \primalspace) \subseteq \calE_{x}^{\mirrorfunc{}}(\sqrt{2m})\), thus completing the proof.
\end{proof}

\section{Conclusion}
\label{sec:conclusion}
To summarise, we introduce the Metropolis-adjusted Mirror Langevin algorithm (\nameref{alg:mamla}), and provide non-asymptotic mixing time guarantees for it.
This algorithm adds a Metropolis-Hastings accept-reject filter to the proposal Markov chain defined by a single step of the Mirror Langevin algorithm (\ref{eq:MLA}) which is the Euler-Maruyama discretisation of the Mirror Langevin dynamics (\ref{eq:MLD}), at each point in \(\primalspace\).
The resulting Markov chain is reversible with respect to the target distribution that we seek to sample from, unlike the Markov chain induced by \ref{eq:MLA}.

Our mixing time guarantees are strongly related to the maximum permissible stepsize that results in non-vanishing acceptance rates, and our current analysis shows that this scales as \(\nicefrac{1}{d^{3}}\) as \(d\) grows.
Motivated by our empirical studies, we believe that for structured domains like polytopes where \(\mirrorfunc{}\) is chosen to be the log-barrier of the domain, a step size that scales as \(\nicefrac{1}{d^{\gamma}}\) for \(\gamma < 3\) can lead to non-vanishing acceptance rates for \nameref{alg:mamla}, and consequently yield better mixing time guarantees.
Another question that we leave for the future is to check if the relative Lipschitz continuity assumption over \(\potential{}\) can be relaxed or removed altogether for the mixing time analysis of \nameref{alg:mamla}, especially when used in conjuction with relative convexity and smoothness of \(\potential{}\).
This is motivated by the analysis in \citet{li2022mirror}, which does not employ such a condition to derive mixing time guarantees for \ref{eq:MLA}.

More conceptually, \nameref{alg:mamla} also strongly uses the dual form of \ref{eq:MLD} i.e., by using \ref{eq:MLA} which is a discretisation of the SDE in the dual space.
Using It\^{o}'s lemma, the equivalent SDE in the primal space is given by \citep{zhang2020wasserstein,li2022mirror}
\begin{equation*}
    dX_{t} = (\nabla \cdot \{\hessmirror{X_{t}}^{-1}\} - \hessmirror{X_{t}}^{-1}\gradpotential{X_{t}}) dt + \sqrt{2 \hessmirror{X_{t}}^{-1}} dB_{t}~.
\end{equation*}
While an algorithm based on the discretision of the primal SDE by itself will not guarantee feasible iterates, the Metropolis-Hastings filter could be used to ensure that (1) the iterates are feasible at all times, and (2) the resulting Markov chain is reversible with respect to the target distribution.
Notably, the primal space SDE introduces a third-order quantity, specifically \(\nabla \cdot \{\hessmirror{x}^{-1}\}\) which might be hard to compute.
However, owing to the fact that the Metropolis-adjustment results in a Markov chain that is reversible with respect to the target, it would be instructive to analyse a proposal that does not contain the aforementioned third order expression in it.

\section*{Acknowledgements}
\label{sec:ack}
The authors would like to thank the reviewers at COLT 2024 for their feedback, and Sinho Chewi for helpful remarks.
The authors acknowledge the MIT SuperCloud and Lincoln Laboratory Supercomputing Center for providing high-performance computing resources that have contributed to the experimental results reported within this work.

\bibliography{references.bib}

\appendix

\section{Addendum}
\label{app:sec:addendum}
\subsection{Miscellaneous algebraic lemmas}

\begin{lemma}
\label{lem:bound_x-1-logx}
Let \(f : (0, \infty) \to (0, \infty)\) be defined as \(f(x) = x - 1 - \log(x)\).
Then for all \(x > 0\),
\begin{equation*}
    f(x) \leq \frac{(x - 1)^{2}}{x}.
\end{equation*}
\end{lemma}
\begin{proof}
We begin with an algebraic manipulation.
\begin{align*}
    g(x) &= f(x) - \frac{(x - 1)^{2}}{x} \\
    &= x - 1 - \log(x) - \frac{(x - 1)^{2}}{x} \\
    &= \frac{x^{2} - x - (x - 1)^{2}}{x} - \log(x) \\
    &= \frac{x - 1}{x} - \log(x).
\end{align*}
The function \(g(x) = \frac{x - 1}{x} - \log(x)\) has derivative \(g'(x) = \frac{(1 - x)}{x^{2}}\).
The only solution to \(g'(x) = 0\) for \(x > 0\) is \(x = 1\).
Moreover, \(g''(x) = -\frac{2}{x^{3}} + \frac{1}{x^{2}}\), and \(g''(1) = -1 < 0\).
This implies that \(g(x)\) attains its maximum at \(x = 1\), and this maximum value is \(1 - 1 - \log(0) = 0\).
As a result of the calculation above, for any \(x > 0\), \(f(x) - \frac{(x - 1)^{2}}{x} \leq \max\limits_{x > 0} g(x) = 0\), which concludes the proof.
\end{proof}

\begin{lemma}
\label{lem:t1-lower}
Let \(t \in [0, 0.5]\).
Then,
\begin{equation*}
    \frac{t^{2}}{4}\left(1 - \frac{1}{(1 - t)^{2}}\right) \geq -\frac{3}{2}t^{3}
\end{equation*}
\end{lemma}
\begin{proof}
We begin with the quantity
\begin{equation*}
    \frac{t^{2}}{4}\left(1 - \frac{1}{(1 - t)^{2}}\right) + \frac{3}{2}t^{3} = \frac{t^{3}(6t^{2} - 11t + 4)}{4(1 - t)^{2}} = \frac{t^{3}(3t - 4)(2t - 1)}{4(1 - t)^{2}}~.
\end{equation*}
When \(t \in [0, 0.5]\), \(t^{3}, (1 - t)^{2} \geq 0\), and \((3t - 4), (2t - 1) \leq 0\).
This implies that
\begin{equation*}
    \frac{t^{3}(3t - 4)(2t - 1)}{4(1 - t)^{2}} \geq 0 \Leftrightarrow \frac{t^{2}}{4}\left(1 - \frac{1}{(1 - t)^{2}}\right) + \frac{3}{2}t^{3} \geq 0~,
\end{equation*}
which completes the proof.
\end{proof}

\begin{lemma}
\label{lem:t2-lower}
Let \(t \in [0, 0.5]\).
Then,
\begin{equation*}
    \log(1 - t) \geq -\frac{3}{2}t
\end{equation*}
\end{lemma}
\begin{proof}
Consider the function \(f(t) = \log(1 - t) + \frac{3}{2}t\).
The second derivative of \(f\) is \(f''(t) = -\frac{1}{(1 - t)^{2}}\).
Note that \(f''(t) < 0\) for \(t \in [0, 0.5]\), which implies that \(f\) is strictly concave in \([0, 0.5]\).
Consequently, the minimum of \(f\) restricted to \([0, 0.5]\) is attained at either \(0\) or \(0.5\).
Hence, for any \(t \in [0, 0.5]\),
\begin{equation*}
    f(t) \geq \min\{f(0), f(0.5)\} = \min\{0, 0.75 + \log(0.5)\} = 0~.
\end{equation*}
\end{proof}

\begin{lemma}
\label{lem:dirichlet-relative-lipschitz}
Let \(a \in \bbR^{d + 1}\).
Then, for any \(w, z \in \bbR^{d}\), \(c \in \bbR\),
\begin{equation*}
    \left(-\sum_{i=1}^{d} a_{i}z_{i} + c \cdot a_{d+1} \cdot \sum_{i=1}^{d}w_{i}\right)^{2} \leq \|a\|^{2} \cdot \left(\sum_{i=1}^{d}z_{i}^{2} + c^{2}\cdot \left\{\sum_{i=1}^{d}w_{i}\right\}^{2}\right)
\end{equation*}
\end{lemma}
\begin{proof}
Construct two vectors of length \(d + 1\) as shown below.
\begin{equation*}
    \sfA = \left[-a_{1}, \cdots, -a_{d}, a_{d + 1}\right] ~;\qquad
    \sfB = \left[z_{1}, \cdots, z_{d}, c \cdot \sum_{i=1}^{d}w_{i}\right]~.
\end{equation*}
The LHS of the inequality in the statement is \((\sfA^{\top}\sfB)^{2}\).
Using the Cauchy-Schwarz inequality,
\begin{equation*}
    (\sfA^{\top}\sfB)^{2} \leq \|\sfA\|^{2}\|\sfB\|^{2} = \|a\|^{2} \cdot \left(\sum_{i=1}^{d}z_{i}^{2} + c^{2} \cdot \left\{\sum_{i=1}^{d}w_{i}\right\}^{2}\right)
\end{equation*}
thus obtaining the statement of the lemma.
\end{proof}

\subsection{An algorithm to compute the inverse mirror map for a simplex}
\label{app:sec:simplex-inverse-mirror-map}

Recall the simplex which is defined as the set
\begin{equation*}
    \primalspace = \left\{x \in \bbR_{+}^{d} : \bm{1}^{\top}x \leq 1\right\}~.
\end{equation*}

The log-barrier of the simplex, which we consider to be the mirror function is defined as
\begin{equation*}
    \mirrorfunc{x} = -\sum_{i=1}^{d}\log x_{i} - \log\left(1 - \sum_{i=1}^{d}x_{i}\right)~.
\end{equation*}
Since the simplex is a compact subset of \(\bbR^{d}\), for every \(y \in \bbR^{d}\) there exists \(x \in \interior{\primalspace}\) such that \(y = \primtodual{x}\).
Element-wise, this is equivalent to stating
\begin{equation*}
    y_{i} = -\frac{1}{x_{i}} + \frac{1}{1 - \bm{1}^{\top}x} \qquad \forall ~i \in [d]~.
\end{equation*}

The goal is to recover \(x\) given some real vector \(y\).
Let \(C = 1 - \bm{1}^{\top}x\).
Then,
\begin{equation}
\label{app:eq:simplex-inverse-element}
    y_{i} = -\frac{1}{x_{i}} + \frac{1}{C} \Leftrightarrow x_{i} = \frac{C}{1 - C y_{i}}~.
\end{equation}

Consequently, \(C\) must satisfy
\begin{equation*}
    C = 1 - \bm{1}^{\top}x = 1 - \sum_{i=1}^{d}x_{i} = 1 - \sum_{i=1}^{d} \frac{C}{1 - C y_{i}} \Leftrightarrow f(C) \defeq C + C\sum_{i=1}^{d} \frac{C}{1 - C y_{i}} - 1 = 0~.
\end{equation*}
By the domain constraints, \(C \in (0, 1)\) (equality is attained when \(x\) lies on the boundary of \(\primalspace\)).

Note the following properties of \(f\):
\begin{enumerate}
    \item the asymptotes of \(f\) are given by \(y_{1}^{-1}, \ldots, y_{d}^{-1}\) with \(\lim\limits_{C \to y_{i}^{+}} f(C) = -\infty\) and \(\lim\limits_{C \to y_{i}^{-}} f(C) = \infty\).
    \item \(f\) is increasing between two consecutive asymptotes, and
    \item \(f(0) = -1\).
\end{enumerate}

We identify ranges to perform binary search on, and this depends on the sign of \(y_{(1)} = \max_{j \in [d]} y_{j}\).
\begin{description}
    \item [\(y_{(1)} \leq 0\):] in this case, since there are no positive asymptotes, \(f\) must be increasing in the domain \((0, \infty)\).
    Since \(f(0) = -1\), the root \(C^{\star}\) could lie anywhere between \((0, 1)\), and therefore the search bounds as \((0, 1)\).
    \item [\(y_{(1)} > 0\):] in this case, there is at least one positive asymptote.
    For the sake of contradiction, let \(C^{\star} > y_{(1)}^{-1}\).
    Then for the \(j\) such that \(y_{j} = y_{(1)}\), we have
    \begin{equation*}
        y_{j} = -\frac{1}{x_{j}} + \frac{1}{C^{\star}} \Rightarrow \frac{1}{x_{j}} = \frac{1}{C^{\star}} - y_{j} < y_{j} - y_{j} < 0~.
    \end{equation*}
    which is a contradiction as \(x_{j} \geq 0\).
    Therefore, the root \(C^{\star}\) must be at most \(\min\{1, y_{(1)}^{-1}\}\).
    This gives the search bounds \(\left(0, \min\{1, y_{(1)}^{-1}\}\right)\).
\end{description}

Having identified the bounds, one can use a binary search method to identify the root.
With \(C^{\star}\), the corresponding \(x\) can be obtained from \cref{app:eq:simplex-inverse-element}.

\end{document}